\documentclass[12pt]{article}
\pdfoutput=1
\usepackage{fullpage}
\def\citep#1{\cite{#1}}

\newif\ifpreprint
\preprintfalse

\usepackage[utf8]{inputenc} 
\usepackage[T1]{fontenc}    

\usepackage{amsfonts}       
\usepackage{amssymb}
\usepackage{amsmath}
\usepackage{euscript}
\usepackage{stmaryrd} 
\usepackage{color}
\usepackage{cases}
\usepackage[colors]{optsys}
\usepackage{causal}
\usepackage{snippets}
\newcommand{\SnippetActiveOe}{
  \sep
  \begin{sentence}
    \begin{input}
      ~~\PY{k+kn}{Definition}~\PY{n+nf}{ActiveOe}~\PY{o}{(}\PY{n+nv}{W}\PY{o}{:}~\PY{n+nb}{set}~\PY{n}{T}\PY{o}{)}~\PY{o}{(}\PY{n+nv}{E}\PY{o}{:}~\PY{n}{relation}~\PY{n}{T}\PY{o}{)}~\PY{o}{:=}~\nl
      ~~~~\PY{o}{[}\PY{n+nb}{set}~\PY{n}{oe}~\PY{o}{:}~\PY{o}{(}\PY{n}{T}\PY{o}{*}\PY{n}{T}\PY{o}{*}\PY{n}{O}\PY{o}{)}~\PY{o}{*}~\PY{o}{(}\PY{n}{T}\PY{o}{*}\PY{n}{T}\PY{o}{*}\PY{n}{O}\PY{o}{)}~\PY{o}{|}~\nl
      ~~~~~~\PY{n}{Oedge}~\PY{n}{E}~\PY{n}{oe}\PY{o}{.}\PY{l+m+mi}{1}~\PY{o}{\(\,\wedge\,\)}~\PY{n}{Oedge}~\PY{n}{E}~\PY{n}{oe}\PY{o}{.}\PY{l+m+mi}{2}~\PY{o}{\(\,\wedge\,\)}~\PY{o}{(}\PY{n}{ChrelO}~\PY{n}{oe}\PY{o}{)}\nl
      ~~~~~~\PY{o}{\(\,\wedge\,\)}~\PY{k+kr}{match}~\PY{o}{(}\PY{n}{oe}\PY{o}{.}\PY{l+m+mi}{1}\PY{o}{.}\PY{l+m+mi}{2}\PY{o}{,}\PY{n}{oe}\PY{o}{.}\PY{l+m+mi}{2}\PY{o}{.}\PY{l+m+mi}{2}\PY{o}{,}~\PY{n}{oe}\PY{o}{.}\PY{l+m+mi}{1}\PY{o}{.}\PY{l+m+mi}{1}\PY{o}{.}\PY{l+m+mi}{2}\PY{o}{)}~\PY{k+kr}{with}~\nl
      ~~~~~~~~\PY{o}{|}~\PY{o}{(}\PY{n}{P}\PY{o}{,}\PY{n}{P}\PY{o}{,}\PY{n}{v}\PY{o}{)}~\PY{o}{=\PYZgt{}}~\PY{n}{W}\PY{o}{.\PYZca{}}\PY{n}{c}~\PY{n}{v}\nl
      ~~~~~~~~\PY{o}{|}~\PY{o}{(}\PY{n}{N}\PY{o}{,}\PY{n}{N}\PY{o}{,}\PY{n}{v}\PY{o}{)}~\PY{o}{=\PYZgt{}}~\PY{n}{W}\PY{o}{.\PYZca{}}\PY{n}{c}~\PY{n}{v}\nl
      ~~~~~~~~\PY{o}{|}~\PY{o}{(}\PY{n}{N}\PY{o}{,}\PY{n}{P}\PY{o}{,}\PY{n}{v}\PY{o}{)}~\PY{o}{=\PYZgt{}}~\PY{n}{W}\PY{o}{.\PYZca{}}\PY{n}{c}~\PY{n}{v}\nl
      ~~~~~~~~\PY{o}{|}~\PY{o}{(}\PY{n}{P}\PY{o}{,}\PY{n}{N}\PY{o}{,}\PY{n}{v}\PY{o}{)}~\PY{o}{=\PYZgt{}}~\PY{o}{(}\PY{n}{Fset}~\PY{n}{E}\PY{o}{.*}~\PY{n}{W}\PY{o}{)}~\PY{n}{v}\nl
      ~~~~~~~~\PY{k+kr}{end}\PY{o}{].}
    \end{input}
  \end{sentence}
}

\newcommand{\SnippetDseparated}{
  \sep
  \begin{sentence}
    \begin{input}
      ~~\PY{k+kn}{Definition}~\PY{n+nf}{D\PYZus{}separated}~\PY{o}{(}\PY{n+nv}{W}\PY{o}{:}~\PY{n+nb}{set}~\PY{n}{T}\PY{o}{)}~\PY{o}{(}\PY{n+nv}{E}\PY{o}{:}~\PY{n}{relation}~\PY{n}{T}\PY{o}{)}~\PY{o}{(}\PY{n+nv}{x}~\PY{n+nv}{y}\PY{o}{:}~\PY{n}{T}\PY{o}{)}~\PY{o}{:=}~\nl
      ~~~~\PY{o}{\PYZti{}(}\PY{k+kr}{\(\exists\)}~\PY{o}{(}\PY{n+nv}{p}\PY{o}{:}~\PY{n}{seq}~\PY{o}{(}\PY{n}{T}\PY{o}{*}\PY{n}{T}\PY{o}{*}\PY{n}{O}\PY{o}{)),}~\PY{n}{Active\PYZus{}path}~\PY{n}{W}~\PY{n}{E}~\PY{n}{p}~\PY{n}{x}~\PY{n}{y}\PY{o}{).}
    \end{input}
  \end{sentence}
}
\newcommand{\SnippetAeop}{
  \sep
  \begin{sentence}
    \begin{input}
      ~~\PY{k+kn}{Definition}~\PY{n+nf}{Active\PYZus{}path}\nl
      ~~~~\PY{o}{(}\PY{n+nv}{W}\PY{o}{:}~\PY{n+nb}{set}~\PY{n}{T}\PY{o}{)}~\PY{o}{(}\PY{n+nv}{E}\PY{o}{:}~\PY{n}{relation}~\PY{n}{T}\PY{o}{)}~\PY{o}{(}\PY{n+nv}{p}\PY{o}{:}~\PY{n}{seq}~\PY{o}{(}\PY{n}{T}\PY{o}{*}\PY{n}{T}\PY{o}{*}\PY{n}{O}\PY{o}{))}~\PY{o}{(}\PY{n+nv}{x}~\PY{n+nv}{y}\PY{o}{:}~\PY{n}{T}\PY{o}{)}~\PY{o}{:=}\nl
      ~~~~\PY{k+kr}{match}~\PY{n}{p}~\PY{k+kr}{with}~\nl
      ~~~~\PY{o}{|}~\PY{o}{[::]}~\PY{o}{=\PYZgt{}}~\PY{n}{x}~\PY{o}{=}~\PY{n}{y}~\nl
      ~~~~\PY{o}{|}~\PY{o}{[::}\PY{n}{eo1}\PY{o}{]}~\PY{o}{=\PYZgt{}}~\PY{n}{eo1}\PY{o}{.}\PY{l+m+mi}{1}\PY{o}{.}\PY{l+m+mi}{1}~\PY{o}{=}~\PY{n}{x}~\PY{o}{\(\,\wedge\,\)}~~\PY{n}{eo1}\PY{o}{.}\PY{l+m+mi}{1}\PY{o}{.}\PY{l+m+mi}{2}~\PY{o}{=}~\PY{n}{y}~\PY{o}{\(\,\wedge\,\)}~~\PY{n}{Oedge}~\PY{n}{E}~\PY{n}{eo1}~\nl
      ~~~~\PY{o}{|}~\PY{n}{eo1}~\PY{o}{::}~\PY{o}{[::}~\PY{n}{eo2}~\PY{o}{\PYZam{}}~\PY{n}{p}\PY{o}{]}\nl
      ~~~~~~\PY{o}{=\PYZgt{}}~\PY{n}{eo1}\PY{o}{.}\PY{l+m+mi}{1}\PY{o}{.}\PY{l+m+mi}{1}~\PY{o}{=}~\PY{n}{x}~\PY{o}{\(\,\wedge\,\)}~\PY{o}{(}\PY{n+nb}{last}~\PY{n}{eo2}~\PY{n}{p}\PY{o}{).}\PY{l+m+mi}{1}\PY{o}{.}\PY{l+m+mi}{2}~\PY{o}{=}~\PY{n}{y}~\nl
      ~~~~~~~~\PY{o}{\(\,\wedge\,\)}~\PY{n}{allL}~\PY{o}{(}\PY{n}{ActiveOe}~\PY{n}{W}~\PY{n}{E}\PY{o}{)}~\PY{o}{(}\PY{n}{belast}~\PY{n}{eo2}~\PY{n}{p}\PY{o}{)}~\PY{n}{eo1}~\PY{o}{(}\PY{n+nb}{last}~\PY{n}{eo2}~\PY{n}{p}\PY{o}{)}\nl
      ~~~~\PY{k+kr}{end}\PY{o}{.}
    \end{input}
  \end{sentence}
}
\newcommand{\Snippetallnotation}{
  \sep
  \begin{sentence}
    \begin{input}
      ~~\PY{k+kn}{Notation}~\PY{l+s+s2}{\PYZdq{}p~\([{\in}]\)~X\PYZdq{}}~\PY{o}{:=}~\PY{o}{(}\PY{k+kp}{all}~\PY{o}{(}\PY{k+kr}{fun}~\PY{n+nv}{x}~\PY{o}{=\PYZgt{}}~\PY{n}{x}~\(\in\)~\PY{n}{X}\PY{o}{)}~\PY{n}{p}\PY{o}{).}~
    \end{input}
  \end{sentence}
}
\newcommand{\SnippetChrel}{
  \sep
  \begin{sentence}
    \begin{input}
      ~~\PY{k+kn}{Definition}~\PY{n+nf}{Chrel}~\PY{o}{\PYZob{}}\PY{n+nv}{T}\PY{o}{:}\PY{k+kt}{Type}\PY{o}{\PYZcb{}}~\PY{o}{:=[}\PY{n+nb}{set}~\PY{n}{s}\PY{o}{:}~\PY{o}{(}\PY{n}{T}\PY{o}{*}\PY{n}{T}\PY{o}{)*(}\PY{n}{T}\PY{o}{*}\PY{n}{T}\PY{o}{)|}~\PY{o}{(}\PY{n}{s}\PY{o}{.}\PY{l+m+mi}{1}\PY{o}{).}\PY{l+m+mi}{2}~\PY{o}{=}~\PY{o}{(}\PY{n}{s}\PY{o}{.}\PY{l+m+mi}{2}\PY{o}{).}\PY{l+m+mi}{1}\PY{o}{].}
    \end{input}
  \end{sentence}
}
\newcommand{\SnippetDU}{
  \sep
  \begin{sentence}
    \begin{input}
      ~~\PY{k+kn}{Definition}~\PY{n+nf}{D\PYZus{}U}~\PY{o}{(}\PY{n+nv}{R}~\PY{n+nv}{E}\PY{o}{:}~\PY{n}{relation}~\PY{n}{T}\PY{o}{)}~\PY{o}{:=}~\PY{o}{[}\PY{n+nb}{set}~\PY{n}{stto}~\PY{o}{:}~\PY{n}{seq}~\PY{o}{(}\PY{n}{T}\PY{o}{*}\PY{n}{T}\PY{o}{*}\PY{n}{O}\PY{o}{)}~\PY{o}{|}\PY{n}{size}\PY{o}{(}\PY{n}{stto}\PY{o}{)\PYZgt{}}\PY{l+m+mi}{0}~\nl
      ~~~~~\PY{o}{\(\,\wedge\,\)}~\PY{n}{R}~\PY{o}{(}\PY{n}{Eope}~\PY{n}{stto}\PY{o}{)}~\PY{o}{\(\,\wedge\,\)}~\PY{n}{stto}~\([{\in}]\)~\PY{o}{(}\PY{n}{Oedge}~\PY{n}{E}\PY{o}{)}~\PY{o}{\(\,\wedge\,\)}~\PY{n}{stto}~\PY{o}{[}\PY{n}{Suc}\(\in\)\PY{o}{]}~\PY{n}{ChrelO}\PY{o}{].}
    \end{input}
  \end{sentence}
}
\newcommand{\SnippetDUa}{
  \sep
  \begin{sentence}
    \begin{input}
      ~~\PY{k+kn}{Definition}~\PY{n+nf}{D\PYZus{}U\PYZus{}a}~\PY{o}{(}\PY{n+nv}{R}~\PY{n+nv}{E}\PY{o}{:}~\PY{n}{relation}~\PY{n}{T}\PY{o}{)}~\PY{o}{(}\PY{n+nv}{W}\PY{o}{:}~\PY{n+nb}{set}~\PY{n}{T}\PY{o}{)}~\PY{o}{(}\PY{n+nv}{x}~\PY{n+nv}{y}\PY{o}{:}\PY{n}{T}\PY{o}{):=}\nl
      ~~~~\PY{o}{(}\PY{n}{D\PYZus{}U}~\PY{n}{R}~\PY{n}{E}\PY{o}{)}~\PY{o}{\(\cap\)}~\PY{o}{[}\PY{n+nb}{set}~\PY{n}{stto}~\PY{o}{|}~\PY{n}{stto}~\PY{o}{[}\PY{n}{Suc}\(\in\)\PY{o}{]}~\PY{o}{(}\PY{n}{A\PYZus{}tr}~\PY{n}{W}~\PY{n}{E}\PY{o}{)].}
    \end{input}
  \end{sentence}
}
\newcommand{\SnippetDUaone}{
  \sep
  \begin{sentence}
    \begin{input}
      ~~\PY{k+kn}{Definition}~\PY{n+nf}{D\PYZus{}U\PYZus{}a1}~\PY{o}{(}\PY{n+nv}{E}\PY{o}{:}~\PY{n}{relation}~\PY{n}{T}\PY{o}{)}~\PY{o}{(}\PY{n+nv}{W}\PY{o}{:}~\PY{n+nb}{set}~\PY{n}{T}\PY{o}{)}~\PY{o}{(}\PY{n+nv}{x}~\PY{n+nv}{y}\PY{o}{:}\PY{n}{T}\PY{o}{):=}\nl
      ~~~~\PY{o}{[}\PY{n+nb}{set}~\PY{n}{stto}~\PY{o}{|}\PY{n}{size}\PY{o}{(}\PY{n}{stto}\PY{o}{)\PYZgt{}}\PY{l+m+mi}{0}~\PY{o}{\(\,\wedge\,\)}~\PY{o}{(}\PY{n}{Eope}~\PY{n}{stto}\PY{o}{)=(}\PY{n}{x}\PY{o}{,}\PY{n}{y}\PY{o}{)}~\PY{o}{\(\,\wedge\,\)}~\PY{n}{stto}~\([{\in}]\)~\PY{o}{(}\PY{n}{Oedge}~\PY{n}{E}\PY{o}{)}~\nl
      ~~~~~\PY{o}{\(\,\wedge\,\)}~\PY{n}{stto}~\PY{o}{[}\PY{n}{Suc}\(\in\)\PY{o}{]}~\PY{n}{ChrelO}~\PY{o}{\(\,\wedge\,\)}~\PY{n}{stto}~\PY{o}{[}\PY{n}{Suc}\(\in\)\PY{o}{]}~\PY{o}{(}\PY{n}{A\PYZus{}tr}~\PY{n}{W}~\PY{n}{E}\PY{o}{)].}
    \end{input}
  \end{sentence}
}

\newcommand{\SnippetOedge}{
  \sep
  \begin{sentence}
    \begin{input}
      ~~\PY{k+kn}{Definition}~\PY{n+nf}{Oedge}~\PY{o}{(}\PY{n+nv}{E}\PY{o}{:}~\PY{n}{relation}~\PY{n}{T}\PY{o}{):}~\PY{n+nb}{set}~\PY{o}{(}\PY{n}{T}\PY{o}{*}\PY{n}{T}\PY{o}{*}\PY{n}{O}\PY{o}{)}~\PY{o}{:=}\nl
      ~~~~\PY{k+kr}{fun}~\PY{o}{(}\PY{n+nv}{oe}\PY{o}{:}~\PY{n}{T}\PY{o}{*}\PY{n}{T}\PY{o}{*}\PY{n}{O}\PY{o}{)}~\PY{o}{=\PYZgt{}}~\PY{k+kr}{match}~\PY{n}{oe}~\PY{k+kr}{with}~\PY{o}{|}~\PY{o}{(}\PY{n}{e}\PY{o}{,}\PY{n}{P}\PY{o}{)}~\PY{o}{=\PYZgt{}}~\PY{n}{E}~\PY{n}{e}~\PY{o}{|}~\PY{o}{(}\PY{n}{e}\PY{o}{,}\PY{n}{N}\PY{o}{)}~\PY{o}{=\PYZgt{}}~\PY{n}{E}\PY{o}{.\PYZhy{}}\PY{l+m+mi}{1}~\PY{n}{e}~\PY{k+kr}{end}\PY{o}{.}
    \end{input}
  \end{sentence}
}
\newcommand{\SnippetTCP}{
  \sep
  \begin{sentence}
    \begin{input}
      ~~\PY{k+kn}{Lemma}~\PY{n+nf}{TCP}\PY{o}{:}~\PY{n}{E}\PY{o}{.+}~\PY{o}{=}~\PY{o}{[}\PY{n+nb}{set}~\PY{n}{vp}\PY{o}{|}~\PY{k+kr}{\(\exists\)}~\PY{n+nv}{p}\PY{o}{,}~\PY{n}{size}\PY{o}{(}\PY{n}{p}\PY{o}{)}~\PY{o}{\PYZgt{}}~\PY{l+m+mi}{1}~\PY{o}{\(\,\wedge\,\)}~\PY{n}{Pe}~\PY{n}{ptv}~\PY{n}{p}~\PY{o}{=}~\PY{n}{vp}~\PY{o}{\(\,\wedge\,\)}~\PY{n}{p}~\PY{o}{[}\PY{n}{L}\(\in\)\PY{o}{]}~\PY{n}{E}\PY{o}{].}
    \end{input}
  \end{sentence}
}

\newcommand{\Snippetdsepnota}{
  \sep
  \begin{sentence}
    \begin{input}
      ~~\PY{k+kn}{Notation}~\PY{l+s+s2}{\PYZdq{}(~x~[\(\bot\)d]~y~|~W~)\PYZdq{}}~\PY{o}{:=}~\PY{o}{(}\PY{n}{D\PYZus{}separated}~\PY{n}{W}~\PY{n}{E}~\PY{n}{x}~\PY{n}{y}\PY{o}{).}
    \end{input}
  \end{sentence}
}

\newcommand{\SnippetLiftO}{
  \sep
  \begin{sentence}
    \begin{input}
      ~~\PY{k+kn}{Definition}~\PY{n+nf}{LiftO}~\PY{o}{(}\PY{n+nv}{st}\PY{o}{:}~\PY{n}{seq}~\PY{n}{T}\PY{o}{)}~\PY{o}{(}\PY{n+nv}{so}\PY{o}{:}~\PY{n}{seq}~\PY{n}{O}\PY{o}{)}~\PY{o}{:=}~\PY{n}{pair}~\PY{o}{(}\PY{n}{Lift}~\PY{n}{st}\PY{o}{)}~\PY{n}{so}\PY{o}{.}
    \end{input}
  \end{sentence}
}
\newcommand{\SnippetSn}{
  \sep
  \begin{sentence}
    \begin{input}
      ~~\PY{k+kn}{Definition}~\PY{n+nf}{Sn}~\PY{o}{(}\PY{n+nv}{n}\PY{o}{:}~\PY{n}{nat}\PY{o}{)}~\PY{o}{(}\PY{n+nv}{D}\PY{o}{:}~\PY{n+nb}{set}~\PY{n}{T}\PY{o}{):=}~\PY{o}{[}\PY{n+nb}{set}~\PY{n}{st}\PY{o}{|}~\PY{n}{st}~\([{\in}]\)~\PY{n}{D}\PY{o}{\(\,\wedge\,\)}\PY{n}{size}\PY{o}{(}\PY{n}{st}\PY{o}{)=}\PY{n}{n}\PY{o}{].}
    \end{input}
  \end{sentence}
}
\newcommand{\SnippetLift}{
  \sep
  \begin{sentence}
    \begin{input}
      ~~\PY{k+kn}{Fixpoint}~\PY{n+nf}{Lift}~\PY{o}{(}\PY{n+nv}{st}\PY{o}{:}~\PY{n}{seq}~\PY{n}{T}\PY{o}{):}~\PY{n}{seq}~\PY{o}{(}\PY{n}{T}\PY{o}{*}\PY{n}{T}\PY{o}{)}~\PY{o}{:=}~\nl
      ~~~~\PY{k+kr}{match}~\PY{n}{st}~\PY{k+kr}{with}~\nl
      ~~~~\PY{o}{|}~\PY{n}{x}~\PY{o}{::}~\PY{o}{[::}~\PY{n}{y}~\PY{o}{\PYZam{}}~\PY{n}{st}\PY{o}{]}~\PY{k+kr}{as}~\PY{n}{st1}~\PY{o}{=\PYZgt{}}~\PY{o}{(}\PY{n}{x}\PY{o}{,}\PY{n}{y}\PY{o}{)::(}\PY{n}{Lift}~\PY{n}{st1}\PY{o}{)}\nl
      ~~~~\PY{o}{|}~\PY{n}{\PYZus{}}~\PY{o}{=\PYZgt{}}~\PY{n}{Nil}~\PY{o}{(}\PY{n}{T}\PY{o}{*}\PY{n}{T}\PY{o}{)}\nl
      ~~~~\PY{k+kr}{end}\PY{o}{.}
    \end{input}
  \end{sentence}
}
\newcommand{\SnippetOO}{
  \sep
  \begin{sentence}
    \begin{input}
      ~~\PY{k+kn}{Inductive}~\PY{n+nf}{O}~\PY{o}{:=}~\PY{o}{|}~\PY{n}{P}~\PY{o}{|}~\PY{n}{N}\PY{o}{.}
    \end{input}
  \end{sentence}
}
\newcommand{\Snippetpair}{
  \sep
  \begin{sentence}
    \begin{input}
      ~~\PY{k+kn}{Fixpoint}~\PY{n+nf}{pair}~\PY{o}{(}\PY{n+nv}{stt}\PY{o}{:}~\PY{n}{seq}~\PY{o}{(}\PY{n}{T}\PY{o}{*}\PY{n}{T}\PY{o}{))}~\PY{o}{(}\PY{n+nv}{so}\PY{o}{:}~\PY{n}{seq}~\PY{n}{O}\PY{o}{)}~\PY{o}{:=}~\nl
      ~~~~\PY{k+kr}{match}~\PY{n}{stt}\PY{o}{,}~\PY{n}{so}~\PY{k+kr}{with}~\nl
      ~~~~\PY{o}{|}~\PY{o}{(}\PY{n}{pt}\PY{o}{)::}\PY{n}{stt}\PY{o}{,}~\PY{n}{o}\PY{o}{::}\PY{n}{so}~\PY{o}{=\PYZgt{}}~\PY{o}{(}\PY{n}{pt}\PY{o}{,}\PY{n}{o}\PY{o}{)::(}\PY{n}{pair}~\PY{n}{stt}~\PY{n}{so}\PY{o}{)}\nl
      ~~~~\PY{o}{|}~\PY{o}{(}\PY{n}{pt}\PY{o}{)::}\PY{n}{stt}\PY{o}{,}~\PY{o}{[::]}~\PY{o}{=\PYZgt{}}~~\PY{o}{(}\PY{n}{pt}\PY{o}{,}\PY{n}{P}\PY{o}{)::(}\PY{n}{pair}~\PY{n}{stt}~\PY{o}{[::])}\nl
      ~~~~\PY{o}{|}~\PY{n}{\PYZus{}}~\PY{o}{,}~\PY{n}{\PYZus{}}~\PY{o}{=\PYZgt{}}~\PY{n}{Nil}~\PY{o}{(}\PY{n}{T}\PY{o}{*}\PY{n}{T}\PY{o}{*}\PY{n}{O}\PY{o}{)}\nl
      ~~~~\PY{k+kr}{end}\PY{o}{.}
    \end{input}
  \end{sentence}
}

\newcommand{\Snippetrelation}{
  \sep
  \begin{sentence}
    \begin{input}
      ~~\PY{k+kn}{Definition}~\PY{n+nf}{relation}~\PY{o}{(}\PY{n+nv}{T}\PY{o}{:}~\PY{k+kt}{Type}\PY{o}{)}~\PY{o}{:=}~\PY{n+nb}{set}~\PY{o}{(}\PY{n}{T}~\PY{o}{*}~\PY{n}{T}\PY{o}{).}
    \end{input}
  \end{sentence}
}
\newcommand{\SnippetRPathequiv}{
  \sep
  \begin{sentence}
    \begin{input}
      ~~\PY{k+kn}{Lemma}~\PY{n+nf}{RPath\PYZus{}equiv}\PY{o}{:}~\PY{k+kr}{\(\forall\)}~\PY{o}{(}\PY{n+nv}{st}\PY{o}{:}~\PY{n}{seq}~\PY{n}{T}\PY{o}{),}~\PY{n}{st}~\PY{o}{[}\PY{n}{L}\(\in\)\PY{o}{]}~\PY{n}{R}~\PY{o}{\(\leftrightarrow\)}~\PY{n}{st}~\PY{o}{[}\PY{n}{Suc}\(\in\)\PY{o}{]}~\PY{n}{R}\PY{o}{.}
    \end{input}
  \end{sentence}
}
\newcommand{\SnippetRPath}{
  \sep
  \begin{sentence}
    \begin{input}
      ~~\PY{k+kn}{Notation}~\PY{l+s+s2}{\PYZdq{}s~[Suc{\(\in\)}]~R\PYZdq{}}~\PY{o}{:=}~\PY{o}{(}\PY{n}{RPath}~\PY{n}{R}~\PY{n}{s}\PY{o}{).}
    \end{input}
  \end{sentence}
}
\newcommand{\SnippetLiftSuc}{
  \sep
  \begin{sentence}
    \begin{input}
      ~~\PY{k+kn}{Lemma}~\PY{n+nf}{Lift\PYZus{}Suc}\PY{o}{:}~\PY{k+kr}{\(\forall\)}~\PY{o}{(}\PY{n+nv}{st}\PY{o}{:}\PY{n}{seq}~\PY{n}{T}\PY{o}{),}~\PY{o}{(}\PY{n}{Lift}~\PY{n}{st}\PY{o}{)}~\PY{o}{[}\PY{n}{Suc}\(\in\)\PY{o}{]}~\PY{n}{Chrel}\PY{o}{.}~
    \end{input}
  \end{sentence}
}
\newcommand{\SnippetRpathLone}{
  \sep
  \begin{sentence}
    \begin{input}
      ~~\PY{k+kn}{Lemma}~\PY{n+nf}{Rpath\PYZus{}L1}\PY{o}{:}~\PY{k+kr}{\(\forall\)}~\PY{o}{(}\PY{n+nv}{st}\PY{o}{:}~\PY{n}{seq}~\PY{n}{T}\PY{o}{),}~\PY{n}{st}~\([{\in}]\)~\PY{n}{X}~\PY{o}{\(\rightarrow\)}~\PY{n}{st}~\PY{o}{[}\PY{n}{L}\(\in\)\PY{o}{]}~\PY{o}{(}\PY{n}{X}~\PY{o}{`*`}~\PY{n}{X}\PY{o}{).}~
    \end{input}
  \end{sentence}
}
\newcommand{\SnippetDI}{
  \sep
  \begin{sentence}
    \begin{input}
      ~~\PY{k+kn}{Definition}~\PY{n+nf}{D}~\PY{o}{\PYZob{}}\PY{n+nv}{T}\PY{o}{:}~\PY{k+kt}{Type}\PY{o}{\PYZcb{}:=}~\PY{o}{[}\PY{n+nb}{set}~\PY{n}{st}\PY{o}{:}\PY{n}{seq}~\PY{n}{T}\PY{o}{|}~\PY{n}{size}\PY{o}{(}\PY{n}{st}\PY{o}{)}~\PY{o}{\PYZgt{}}~\PY{l+m+mi}{1}\PY{o}{].}\nl
    \end{input}
  \end{sentence}
  \sep
  \begin{sentence}
    \begin{input}
      ~~\PY{k+kn}{Definition}~\PY{n+nf}{I}~\PY{o}{\PYZob{}}\PY{n+nv}{T}\PY{o}{:}~\PY{k+kt}{Type}\PY{o}{\PYZcb{}:=}~\PY{o}{[}\PY{n+nb}{set}~\PY{n}{spt}\PY{o}{:}\PY{n}{seq}~\PY{o}{(}\PY{n}{T}\PY{o}{*}\PY{n}{T}\PY{o}{)|}~\PY{n}{size}\PY{o}{(}\PY{n}{spt}\PY{o}{)}~\PY{o}{\PYZgt{}}~\PY{l+m+mi}{0}~\PY{o}{\(\,\wedge\,\)}~\PY{n}{spt}~\PY{o}{[}\PY{n}{Suc}\(\in\)\PY{o}{]}~\PY{n}{Chrel}\PY{o}{].}
    \end{input}
  \end{sentence}
}
\newcommand{\SnippetLiftinj}{
  \sep
  \begin{sentence}
    \begin{input}
      ~~\PY{k+kn}{Lemma}~\PY{n+nf}{Lift\PYZus{}inj}\PY{o}{:}~\PY{k+kr}{\(\forall\)}~\PY{o}{(}\PY{n+nv}{st}~\PY{n+nv}{st\PYZsq{}}\PY{o}{:}~\PY{n}{seq}~\PY{n}{T}\PY{o}{),}~\PY{n}{st}~\(\in\)~\PY{n}{D}~\PY{o}{\(\rightarrow\)}~\PY{n}{Lift}~\PY{n}{st}~\PY{o}{=}~\PY{n}{Lift}~\PY{n}{st\PYZsq{}}~\PY{o}{\(\rightarrow\)}~\PY{n}{st}~\PY{o}{=}~\PY{n}{st\PYZsq{}}\PY{o}{.}
    \end{input}
  \end{sentence}
}
\newcommand{\SnippetLiftsurj}{
  \sep
  \begin{sentence}
    \begin{input}
      ~~\PY{k+kn}{Lemma}~\PY{n+nf}{Lift\PYZus{}surj}\PY{o}{:}~\PY{k+kr}{\(\forall\)}~\PY{o}{(}\PY{n+nv}{spt}\PY{o}{:}~\PY{n}{seq}~\PY{o}{(}\PY{n}{T}\PY{o}{*}\PY{n}{T}\PY{o}{)),}~\PY{n}{spt}~\(\in\)~\PY{n}{I}~\PY{o}{\(\rightarrow\)}~\PY{k+kr}{\(\exists\)}~\PY{n+nv}{st}\PY{o}{,}~\PY{n}{st}\(\in\)~\PY{n}{D}~\PY{o}{\(\,\wedge\,\)}~\PY{n}{Lift}~\PY{n}{st}\PY{o}{=}\PY{n}{spt}\PY{o}{.}~
    \end{input}
  \end{sentence}
}
\newcommand{\SnippetEpathgt}{
  \sep
  \begin{sentence}
    \begin{input}
      ~~\PY{k+kn}{Definition}~\PY{n+nf}{P\PYZus{}gt}~\PY{o}{(}\PY{n+nv}{n}\PY{o}{:}~\PY{n}{nat}\PY{o}{)}~\PY{o}{(}\PY{n+nv}{E}\PY{o}{:}~\PY{n}{relation}~\PY{n}{T}\PY{o}{)}~~\PY{o}{:=}~\nl
      ~~~~\PY{o}{[}\PY{n+nb}{set}~\PY{n}{spt}~\PY{o}{|}~\PY{n}{size}\PY{o}{(}\PY{n}{spt}\PY{o}{)}~\PY{o}{\PYZgt{}}~\PY{n}{n}~\PY{o}{\(\,\wedge\,\)}~\PY{n}{spt}~\([{\in}]\)~\PY{n}{E}~\PY{o}{\(\,\wedge\,\)}~\PY{n}{spt}~\PY{o}{[}\PY{n}{Suc}\(\in\)\PY{o}{]}~\PY{n}{Chrel}\PY{o}{].}
    \end{input}
  \end{sentence}
}
\newcommand{\SnippetDP}{
  \sep
  \begin{sentence}
    \begin{input}
      ~~\PY{k+kn}{Definition}~\PY{n+nf}{D\PYZus{}P}~\PY{o}{(}\PY{n+nv}{R}~\PY{n+nv}{E}\PY{o}{:}~\PY{n}{relation}~\PY{n}{T}\PY{o}{):=}~\nl
      ~~~~\PY{o}{[}\PY{n+nb}{set}~\PY{n}{spt}\PY{o}{|}~\PY{n}{spt}~\(\in\)~\PY{n}{I}~\PY{o}{\(\,\wedge\,\)}~\PY{n}{R}~\PY{o}{(}\PY{n}{Epe}~\PY{n}{spt}\PY{o}{)}~\PY{o}{\(\,\wedge\,\)}~\PY{n}{spt}~\([{\in}]\)~\PY{n}{E}~\PY{o}{].}
    \end{input}
  \end{sentence}
}
\newcommand{\SnippetDV}{
  \sep
  \begin{sentence}
    \begin{input}
      ~~\PY{k+kn}{Definition}~\PY{n+nf}{D\PYZus{}V}~\PY{o}{(}\PY{n+nv}{R}~\PY{n+nv}{E}\PY{o}{:}~\PY{n}{relation}~\PY{n}{T}\PY{o}{):=}\nl
      ~~~~\PY{o}{[}\PY{n+nb}{set}~\PY{n}{st}\PY{o}{|}~\PY{n}{st}~\(\in\)~\PY{n}{D}~\PY{o}{\(\,\wedge\,\)}~\PY{n}{R}~\PY{o}{(}\PY{n}{Pe}~\PY{n}{st}\PY{o}{)}~\PY{o}{\(\,\wedge\,\)}~\PY{n}{st}~\PY{o}{[}\PY{n}{Suc}\(\in\)\PY{o}{]}~\PY{n}{E}\PY{o}{].}
    \end{input}
  \end{sentence}
}
\newcommand{\SnippetDPDV}{
  \sep
  \begin{sentence}
    \begin{input}
      ~~\PY{k+kn}{Lemma}~\PY{n+nf}{DP\PYZus{}DV}\PY{o}{:}~\PY{k+kr}{\(\forall\)}~\PY{o}{(}\PY{n+nv}{R}~\PY{n+nv}{E}\PY{o}{:}~\PY{n}{relation}~\PY{n}{T}\PY{o}{),}~\PY{n}{image}~\PY{o}{(}\PY{n}{D\PYZus{}V}~\PY{n}{R}~\PY{n}{E}\PY{o}{)}~\PY{o}{(@}\PY{n}{Lift}~\PY{n}{T}\PY{o}{)}~\PY{o}{=}~\PY{o}{(}\PY{n}{D\PYZus{}P}~\PY{n}{R}~\PY{n}{E}\PY{o}{).}
    \end{input}
  \end{sentence}
}
\newcommand{\SnippetPe}{
  \sep
  \begin{sentence}
    \begin{input}
      ~~\PY{k+kn}{Definition}~\PY{n+nf}{Pe}~\PY{o}{(}\PY{n+nv}{st}\PY{o}{:}~\PY{n}{seq}~\PY{n}{T}\PY{o}{)}~\PY{o}{:=}~\PY{o}{(}\PY{n}{head}~\PY{n}{ptv}\PY{o}{.}\PY{l+m+mi}{1}~\PY{n}{st}\PY{o}{,}~\PY{n+nb}{last}~\PY{n}{ptv}\PY{o}{.}\PY{l+m+mi}{1}~\PY{n}{st}\PY{o}{).}
    \end{input}
  \end{sentence}
}
\newcommand{\SnippetEpe}{
  \sep
  \begin{sentence}
    \begin{input}
      ~~\PY{k+kn}{Definition}~\PY{n+nf}{Epe}~\PY{o}{(}\PY{n+nv}{spt}\PY{o}{:}~\PY{n}{seq}~\PY{o}{(}\PY{n}{T}\PY{o}{*}\PY{n}{T}\PY{o}{))}~\PY{o}{:=}~\PY{o}{((}\PY{n}{head}~\PY{n}{ptv}~\PY{n}{spt}\PY{o}{).}\PY{l+m+mi}{1}\PY{o}{,}~\PY{o}{(}\PY{n+nb}{last}~\PY{n}{ptv}~\PY{n}{spt}\PY{o}{).}\PY{l+m+mi}{2}\PY{o}{).}
    \end{input}
  \end{sentence}
}
\newcommand{\SnippetUgt}{
  \sep
  \begin{sentence}
    \begin{input}
      ~~\PY{k+kn}{Definition}~\PY{n+nf}{U\PYZus{}gt}~\PY{o}{(}\PY{n+nv}{n}\PY{o}{:}~\PY{n}{nat}\PY{o}{)}~\PY{o}{(}\PY{n+nv}{E}\PY{o}{:}~\PY{n}{relation}~\PY{n}{T}\PY{o}{):=}\nl
      ~~~~\PY{o}{[}\PY{n+nb}{set}~\PY{n}{sto}~\PY{o}{|}~\PY{n}{size}\PY{o}{(}\PY{n}{sto}\PY{o}{)}~\PY{o}{\PYZgt{}}~\PY{n}{n}~\PY{o}{\(\,\wedge\,\)}~\PY{n}{sto}~\([{\in}]\)~\PY{o}{(}\PY{n}{Oedge}~\PY{n}{E}\PY{o}{)}~\PY{o}{\(\,\wedge\,\)}~\PY{o}{(}\PY{n}{Lift}~\PY{n}{sto}\PY{o}{)}~\([{\in}]\)~\PY{n}{ChrelO}\PY{o}{].}
    \end{input}
  \end{sentence}
}
\newcommand{\SnippetChrelO}{
  \sep
  \begin{sentence}
    \begin{input}
      ~~\PY{k+kn}{Definition}~\PY{n+nf}{ChrelO}~\PY{o}{:=}~\PY{o}{[}\PY{n+nb}{set}~\PY{n}{ppa}\PY{o}{:}~\PY{o}{(}\PY{n}{T}\PY{o}{*}\PY{n}{T}\PY{o}{*}\PY{n}{O}\PY{o}{)*(}\PY{n}{T}\PY{o}{*}\PY{n}{T}\PY{o}{*}\PY{n}{O}\PY{o}{)}~\PY{o}{|}~\PY{o}{(}\PY{n}{ppa}\PY{o}{.}\PY{l+m+mi}{1}\PY{o}{.}\PY{l+m+mi}{1}\PY{o}{).}\PY{l+m+mi}{2}~\PY{o}{=}~\PY{o}{(}\PY{n}{ppa}\PY{o}{.}\PY{l+m+mi}{2}\PY{o}{.}\PY{l+m+mi}{1}\PY{o}{).}\PY{l+m+mi}{1}\PY{o}{].}
    \end{input}
  \end{sentence}
}
\newcommand{\SnippetEpeLift}{
  \sep
  \begin{sentence}
    \begin{input}
      ~~\PY{k+kn}{Lemma}~\PY{n+nf}{Epe\PYZus{}Lift}\PY{o}{:}~\PY{k+kr}{\(\forall\)}~\PY{o}{(}\PY{n+nv}{st}\PY{o}{:}\PY{n}{seq}~\PY{n}{T}\PY{o}{),}~\PY{n}{st}~\(\in\)~\PY{n}{D}~\PY{o}{\(\rightarrow\)}~\PY{n}{Epe}~\PY{o}{(}\PY{n}{Lift}~\PY{n}{st}\PY{o}{)}~\PY{o}{=}~\PY{n}{Pe}~\PY{n}{st}\PY{o}{.}
    \end{input}
  \end{sentence}
}
\newcommand{\SnippetPeUnLift}{
  \sep
  \begin{sentence}
    \begin{input}
      ~~\PY{k+kn}{Lemma}~\PY{n+nf}{Pe\PYZus{}UnLift}\PY{o}{:}~\PY{k+kr}{\(\forall\)}~\PY{o}{(}\PY{n+nv}{spt}\PY{o}{:}~\PY{n}{seq}~\PY{o}{(}\PY{n}{T}\PY{o}{*}\PY{n}{T}\PY{o}{)),}~\PY{n}{spt}~\(\in\)~\PY{n}{I}\PY{o}{\(\rightarrow\)}\PY{n}{Pe}~\PY{o}{(}\PY{n}{UnLift}~\PY{n}{spt}~\PY{n}{ptv}\PY{o}{.}\PY{l+m+mi}{1}\PY{o}{)=}\PY{n}{Epe}~\PY{n}{spt}\PY{o}{.}
    \end{input}
  \end{sentence}
}
\newcommand{\SnippetEopeLiftO}{
  \sep
  \begin{sentence}
    \begin{input}
      ~~\PY{k+kn}{Lemma}~\PY{n+nf}{Eope\PYZus{}LiftO}\PY{o}{:}~\PY{k+kr}{\(\forall\)}~\PY{o}{(}\PY{n+nv}{st}\PY{o}{:}\PY{n}{seq}~\PY{n}{T}\PY{o}{)}~\PY{o}{(}\PY{n+nv}{so}\PY{o}{:}\PY{n}{seq}~\PY{n}{O}\PY{o}{),}\nl
      ~~~~~~\PY{n}{size}\PY{o}{(}\PY{n}{st}\PY{o}{)}~\PY{o}{\PYZgt{}}~\PY{l+m+mi}{1}~\PY{o}{\(\rightarrow\)}~\PY{n}{size}~\PY{o}{(}\PY{n}{so}\PY{o}{)}~\PY{o}{=}~\PY{n}{size}~\PY{n}{st}~\PY{o}{\PYZhy{}}\PY{l+m+mi}{1}~\PY{o}{\(\rightarrow\)}~\PY{n}{Eope}~\PY{o}{(}\PY{n}{LiftO}~\PY{n}{st}~\PY{n}{so}\PY{o}{)}~\PY{o}{=}~\PY{n}{Pe}~\PY{n}{ptv}~\PY{n}{st}\PY{o}{.}
    \end{input}
  \end{sentence}
}
\newcommand{\SnippetEope}{
  \sep
  \begin{sentence}
    \begin{input}
      ~~\PY{k+kn}{Definition}~\PY{n+nf}{Eope}~\PY{o}{(}\PY{n+nv}{stto}~\PY{o}{:}~\PY{n}{seq}\PY{o}{(}\PY{n}{T}\PY{o}{*}\PY{n}{T}\PY{o}{*}\PY{n}{O}\PY{o}{))}~\PY{o}{:}~\PY{n}{T}\PY{o}{*}\PY{n}{T}~\PY{o}{:=}\nl
      ~~~~\PY{o}{((}\PY{n}{head}~\PY{o}{(}\PY{n}{ptv}\PY{o}{,}\PY{n}{P}\PY{o}{)}~\PY{n}{stto}\PY{o}{).}\PY{l+m+mi}{1}\PY{o}{.}\PY{l+m+mi}{1}\PY{o}{,}~\PY{o}{(}\PY{n+nb}{last}~\PY{o}{(}\PY{n}{ptv}\PY{o}{,}\PY{n}{P}\PY{o}{)}~\PY{n}{stto}\PY{o}{).}\PY{l+m+mi}{1}\PY{o}{.}\PY{l+m+mi}{2}\PY{o}{).}
    \end{input}
  \end{sentence}
}
\newcommand{\SnippetPeUnLiftO}{
  \sep
  \begin{sentence}
    \begin{input}
      ~~\PY{k+kn}{Lemma}~\PY{n+nf}{Pe\PYZus{}UnLiftO}\PY{o}{:}~\PY{k+kr}{\(\forall\)}~\PY{o}{(}\PY{n+nv}{stto}\PY{o}{:}~\PY{n}{seq}~\PY{o}{(}\PY{n}{T}\PY{o}{*}\PY{n}{T}\PY{o}{*}\PY{n}{O}\PY{o}{)),}~\nl
      ~~~~~~\PY{n}{size}\PY{o}{(}\PY{n}{stto}\PY{o}{)}~\PY{o}{\PYZgt{}}~\PY{l+m+mi}{0}~\PY{o}{\(\rightarrow\)}~\PY{n}{stto}~\PY{o}{[}\PY{n}{Suc}\(\in\)\PY{o}{]}~\PY{n}{ChrelO}~\PY{o}{\(\rightarrow\)}~\nl
      ~~~~~~\PY{o}{(}\PY{n}{Pe}~\PY{n}{ptv}~\PY{o}{(}\PY{n}{UnLiftO}~\PY{n}{stto}~\PY{n}{ptv}\PY{o}{.}\PY{l+m+mi}{1}\PY{o}{).}\PY{l+m+mi}{1}\PY{o}{)}~\PY{o}{=}~~\PY{n}{Eope}~\PY{n}{stto}\PY{o}{.}
    \end{input}
  \end{sentence}
}
\newcommand{\SnippetActiveeq}{
  \sep
  \begin{sentence}
    \begin{input}
      ~~\PY{k+kn}{Lemma}~\PY{n+nf}{Active\PYZus{}eq}\PY{o}{:}~\PY{k+kr}{\(\forall\)}~\PY{o}{(}\PY{n+nv}{E}\PY{o}{:}~\PY{n}{relation}~\PY{n}{T}\PY{o}{)}~\PY{o}{(}\PY{n+nv}{W}\PY{o}{:}~\PY{n+nb}{set}~\PY{n}{T}\PY{o}{)}~\PY{o}{(}\PY{n+nv}{x}~\PY{n+nv}{y}\PY{o}{:}\PY{n}{T}\PY{o}{)}~\PY{n+nv}{stto}\PY{o}{,}\nl
      ~~~~~~\PY{o}{((}\PY{n}{x}\PY{o}{=}\PY{n}{y}~\PY{o}{\(\,\wedge\,\)}~\PY{n}{stto}~\PY{o}{=}~\PY{o}{[::])}~\PY{o}{\(\,\vee\,\)}~~\PY{n}{stto}~\(\in\)~\PY{o}{(}\PY{n}{D\PYZus{}U\PYZus{}a1}~\PY{n}{E}~\PY{n}{W}~\PY{n}{x}~\PY{n}{y}\PY{o}{))}\nl
      ~~~~~~\PY{o}{\(\leftrightarrow\)}~\PY{n}{Active\PYZus{}path}~\PY{n}{W}~\PY{n}{E}~\PY{n}{stto}~\PY{n}{x}~\PY{n}{y}\PY{o}{.}
    \end{input}
  \end{sentence}
}
\newcommand{\SnippetallL}{
  \sep
  \begin{sentence}
    \begin{input}
      ~~\PY{k+kn}{Definition}~\PY{n+nf}{allL}~\PY{o}{(}\PY{n+nv}{R}\PY{o}{:}~\PY{n}{relation}~\PY{n}{T}\PY{o}{)}~\PY{n+nv}{st}~\PY{n+nv}{x}~\PY{n+nv}{y}~\PY{o}{:=}~\PY{o}{(}\PY{n}{x}\PY{o}{::(}\PY{n}{rcons}~\PY{n}{st}~\PY{n}{y}\PY{o}{))}~\PY{o}{[}\PY{n}{L}\(\in\)\PY{o}{]}~\PY{n}{R}\PY{o}{.}
    \end{input}
  \end{sentence}
}
\newcommand{\SnippetAtr}{
  \sep
  \begin{sentence}
    \begin{input}
      ~~\PY{k+kn}{Definition}~\PY{n+nf}{A\PYZus{}tr}~\PY{o}{(}\PY{n+nv}{W}\PY{o}{:}~\PY{n+nb}{set}~\PY{n}{T}\PY{o}{)}~\PY{o}{(}\PY{n+nv}{E}\PY{o}{:}~\PY{n}{relation}~\PY{n}{T}\PY{o}{)}~\PY{o}{:=}~\PY{n}{ChrelO}~\PY{o}{\(\cap\)}~\nl
      ~~~~\PY{o}{[}\PY{n+nb}{set}~\PY{n}{oe}~\PY{o}{:}~\PY{o}{(}\PY{n}{T}\PY{o}{*}\PY{n}{T}\PY{o}{*}\PY{n}{O}\PY{o}{)}~\PY{o}{*}~\PY{o}{(}\PY{n}{T}\PY{o}{*}\PY{n}{T}\PY{o}{*}\PY{n}{O}\PY{o}{)|}~\PY{k+kr}{match}~\PY{o}{(}\PY{n}{oe}\PY{o}{.}\PY{l+m+mi}{1}\PY{o}{.}\PY{l+m+mi}{2}\PY{o}{,}\PY{n}{oe}\PY{o}{.}\PY{l+m+mi}{2}\PY{o}{.}\PY{l+m+mi}{2}\PY{o}{,}~\PY{n}{oe}\PY{o}{.}\PY{l+m+mi}{1}\PY{o}{.}\PY{l+m+mi}{1}\PY{o}{.}\PY{l+m+mi}{2}\PY{o}{)}~\PY{k+kr}{with}~\nl
      ~~~~~~\PY{o}{|}~\PY{o}{(}\PY{n}{P}\PY{o}{,}\PY{n}{P}\PY{o}{,}\PY{n}{v}\PY{o}{)}~\PY{o}{=\PYZgt{}}~\PY{n}{W}\PY{o}{.\PYZca{}}\PY{n}{c}~\PY{n}{v}~\PY{o}{|}~\PY{o}{(}\PY{n}{N}\PY{o}{,}\PY{n}{N}\PY{o}{,}\PY{n}{v}\PY{o}{)}~\PY{o}{=\PYZgt{}}~\PY{n}{W}\PY{o}{.\PYZca{}}\PY{n}{c}~\PY{n}{v}~\PY{o}{|}~\PY{o}{(}\PY{n}{N}\PY{o}{,}\PY{n}{P}\PY{o}{,}\PY{n}{v}\PY{o}{)}~\PY{o}{=\PYZgt{}}~\PY{n}{W}\PY{o}{.\PYZca{}}\PY{n}{c}~\PY{n}{v}\nl
      ~~~~~~\PY{o}{|}~\PY{o}{(}\PY{n}{P}\PY{o}{,}\PY{n}{N}\PY{o}{,}\PY{n}{v}\PY{o}{)}~\PY{o}{=\PYZgt{}}~\PY{o}{(}\PY{n}{Fset}~\PY{n}{E}\PY{o}{.*}~\PY{n}{W}\PY{o}{)}~\PY{n}{v}~\PY{k+kr}{end}\PY{o}{].}
    \end{input}
  \end{sentence}
}
\newcommand{\SnippetAw}{
  \sep
  \begin{sentence}
    \begin{input}
      ~~\PY{k+kn}{Definition}~\PY{n+nf}{Em}~\PY{o}{:=}~\PY{n}{E}\PY{o}{.\PYZhy{}}\PY{l+m+mi}{1}\PY{o}{.}\nl
    \end{input}
  \end{sentence}
  \sep
  \begin{sentence}
    \begin{input}
      ~~\PY{k+kn}{Definition}~\PY{n+nf}{Ew}~\PY{o}{:=}~\PY{n}{\(\Delta\)\PYZus{}}\PY{o}{(}\PY{n}{W}\PY{o}{.\PYZca{}}\PY{n}{c}\PY{o}{)}\PY{o}{;}\PY{n}{E}\PY{o}{.}\nl
    \end{input}
  \end{sentence}
  \sep
  \begin{sentence}
    \begin{input}
      ~~\PY{k+kn}{Definition}~\PY{n+nf}{Bw}~\PY{o}{:=}~\PY{n}{E}\PY{o}{;}\PY{n}{Ew}\PY{o}{.*}~\PY{o}{.}\nl
    \end{input}
  \end{sentence}
  \sep
  \begin{sentence}
    \begin{input}
      ~~\PY{k+kn}{Definition}~\PY{n+nf}{Emw}~\PY{o}{:=}~\PY{n}{Ew}\PY{o}{.\PYZhy{}}\PY{l+m+mi}{1}\PY{o}{.}~\nl
    \end{input}
  \end{sentence}
  \sep
  \begin{sentence}
    \begin{input}
      ~~\PY{k+kn}{Definition}~\PY{n+nf}{Bmw}~\PY{o}{:=}~\PY{n}{Bw}\PY{o}{.\PYZhy{}}\PY{l+m+mi}{1}\PY{o}{.}\nl
    \end{input}
  \end{sentence}
  \sep
  \begin{sentence}
    \begin{input}
      ~~\PY{k+kn}{Definition}~\PY{n+nf}{Kw}~\PY{o}{:=}~\PY{o}{(}\PY{n}{Bmw}\PY{o}{;}\PY{n}{\(\Delta\)\PYZus{}}\PY{o}{(}\PY{n}{W}\PY{o}{.\PYZca{}}\PY{n}{c}\PY{o}{)}\PY{o}{;}\PY{n}{Bw}\PY{o}{).}\nl
    \end{input}
  \end{sentence}
  \sep
  \begin{txt}
    ~~\nl
  \end{txt}
  \sep
  \begin{sentence}
    \begin{input}
      ~~\PY{k+kn}{Definition}~\PY{n+nf}{DKD}~\PY{o}{:=}~\PY{o}{(}~\PY{n}{\(\Delta\)\PYZus{}}\PY{o}{(}\PY{n}{W}\PY{o}{)}\PY{o}{;}\PY{n}{Kw}\PY{o}{;}~\PY{n}{\(\Delta\)\PYZus{}}\PY{o}{(}\PY{n}{W}\PY{o}{)).}\nl
    \end{input}
  \end{sentence}
  \sep
  \begin{sentence}
    \begin{input}
      ~~\PY{k+kn}{Definition}~\PY{n+nf}{Cw}~\PY{o}{:=}~\PY{o}{((}\PY{n}{DKD}\PY{o}{).+)}~\PY{o}{\(\cup\)}~\PY{n}{\(\Delta\)\PYZus{}}\PY{o}{(}\PY{n}{W}\PY{o}{).}\nl
    \end{input}
  \end{sentence}
  \sep
  \begin{sentence}
    \begin{input}
      ~~\PY{k+kn}{Definition}~\PY{n+nf}{Dw}~\PY{o}{:=}~\PY{o}{(}\PY{n}{Bw}~\PY{o}{\(\cup\)}~\PY{n}{Kw}\PY{o}{)}\PY{o}{;}\PY{o}{(}\PY{n}{Cw}\PY{o}{;}\PY{o}{(}\PY{n}{Bmw}~\PY{o}{\(\cup\)}~\PY{n}{Kw}\PY{o}{)).}\nl
    \end{input}
  \end{sentence}
  \sep
  \begin{sentence}
    \begin{input}
      ~~\PY{k+kn}{Definition}~\PY{n+nf}{Aw}~\PY{o}{:=}~\PY{o}{\PYZsq{}}\PY{n}{\(\Delta\)}~\PY{o}{\(\cup\)}~\PY{n}{Bw}~\PY{o}{\(\cup\)}~\PY{n}{Bmw}~\PY{o}{\(\cup\)}~\PY{n}{Kw}~\PY{o}{\(\cup\)}~\PY{n}{Dw}\PY{o}{.}
    \end{input}
  \end{sentence}
}
\newcommand{\Snippettheoremfive}{
  \sep
  \begin{sentence}
    \begin{input}
      ~~\PY{k+kn}{Theorem}~\PY{n+nf}{Th5}\PY{o}{:}~\PY{k+kr}{\(\forall\)}~\PY{o}{(}\PY{n+nv}{x}~\PY{n+nv}{y}\PY{o}{:}~\PY{n}{T}\PY{o}{),}~\PY{o}{(}~\PY{n}{x}~\PY{o}{[\(\bot\)}\PY{n}{d}\PY{o}{]}~\PY{n}{y}~\PY{o}{|}~\PY{n}{W}~\PY{o}{)}~\PY{o}{\(\leftrightarrow\)}~\PY{o}{\PYZti{}}~\PY{n}{Aw}~\PY{o}{(}\PY{n}{x}\PY{o}{,}\PY{n}{y}\PY{o}{).}
    \end{input}
  \end{sentence}
}

\usepackage{framed}
\usepackage{graphicx}
\graphicspath{{FIGURES/}}

\newcommand{\myparagraph}[1]{{\medskip\noindent\bf #1}}

\newcommand{\pinX}{\PY{n}{p}~\PY{n}{$[\in]$}~\PY{n}{X}}
\newcommand{\Liftp}{\PY{n}{Lift}~\PY{n}{p}}
\newcommand{\pSucR}{\PY{n}{p}~\PY{o}{$[\text{Suc}{\in}]$}~\PY{n}{R}}
\newcommand{\Liftpin}{\PY{n}{p}~\PY{o}{$[\text{L}{\in}]$}~\PY{n}{R}}

\newlength{\leftbarwidth}
\newlength{\leftbarsep}
\colorlet{leftbarcolor}{red}

\renewenvironment{leftbar}{%
    \MakeFramed {\advance \hsize -\width \FrameRestore }%
}{%
    \endMakeFramed
}

\usepackage{amsmath}

\title{Conditional Separation as a Binary Relation. \\ A Coq Assisted Proof}

\author{%
  Jean-Philippe Chancelier$^\dagger$,
  Michel De Lara\footnote{CERMICS, Ecole des Ponts, Marne-la-Vall\'ee, France},
  Benjamin Heymann\footnote{Criteo AI Lab, Paris, France}
}

\date{\today}

\begin{document}

\maketitle

\begin{abstract}
  The concept of d-separation holds a pivotal role in causality theory, serving
  as a fundamental tool for deriving conditional independence properties from
  causal graphs. Pearl defined the d-separation of two subsets conditionally on
  a third one. In this study, we present a novel perspective by showing
  i) how the d-separation can be extended beyond acyclic graphs,
  possibly infinite, and ii) how it can be
  expressed and characterized as a binary relation between vertices. Compared
  to the typical perspectives in causality theory, our equivalence opens the
  door to more compact and computational proofing techniques, because the
  language of binary relations is well adapted to equational reasoning.
Additionally, and of independent interest, the proofs
  of the results presented in this paper are checked with the Coq proof
  assistant.
\end{abstract}



\section{Introduction} 

In an era increasingly driven by data-informed decision-making, the significance
of causal inference has grown substantially across applied sciences, statistics,
and machine learning.  Pioneering this field, Pearl's seminal work
\citep{pearl1995causal,pearl2018book} leverages graphical models
\citep{cowell2006probabilistic} to introduce the do-calculus and the concept of
d-separation on directed acyclic graphs (DAGs). This concept plays a pivotal role in causality theory by providing a tool for
deducing conditional independence properties from causal graphs.

This study introduces a novel perspective by handling graphs
as binary relations --- hence what we call graph is a directed simple graph
permitting loops in graph theory, and we allow for infinite such graphs ---
and, from there, move in two successive directions.
First, we extend the d-separation beyond acyclic graphs,
to general, possibly infinite, graphs.
Second, we characterize the d-separation
property as a binary relation among the vertices of the graph.
Compared to the
typical perspectives in causality theory, our equivalence opens the door to more
compact and computational proofing techniques, because the language of binary
relations is well adapted to equational reasoning.  Additionally, and of
independent interest, the proofs of the results presented in this paper are
checked with the Coq proof assistant.  Last but not least, the
characterization presented in this work serves as a building block for two
concurrently developed works.\footnote{The mathematical side of the present
  paper was written in~\cite{Chancelier-De-Lara-Heymann-2021} in parallel to two
  other papers
  \citep{De-Lara-Chancelier-Heymann-2021,Heymann-De-Lara-Chancelier-2021}, all
  of which aimed at providing another perspective on conditional independence
  (and do-calculus). The first paper~\cite{Chancelier-De-Lara-Heymann-2021} was
  a prerequisite for \citep{De-Lara-Chancelier-Heymann-2021} and
  both~\cite{Chancelier-De-Lara-Heymann-2021,De-Lara-Chancelier-Heymann-2021}
  were a prerequisite for~\citep{Heymann-De-Lara-Chancelier-2021}.  In order to
  facilitate the reading
  of~\cite{Chancelier-De-Lara-Heymann-2021,De-Lara-Chancelier-Heymann-2021},
  which were quite long and technical, we have implemented Coq proofs for them.
  The aim of the present paper is to provide a version of the mathematical
  results of the preprint~\cite{Chancelier-De-Lara-Heymann-2021} complemented
  with the description of the Coq formalization used for the proofs.  The aim is
  thus twofold, as it gives the proof of yet unpublished results together with
  their Coq assisted proof.}

As far as we know, this is the first attempt to formalize d-separation in a
proof assistant.  However, some works can be found in the literature on
probabilistic conditional independence and proof assistants.  For example,
in~\cite{Affeldt-et-al:2020}, the authors introduce a formalism for reasoning
with conditional probabilities and joint distributions in
Coq. In~\cite{Yamaguchi-et-al:2016}, the authors introduce a formalism for
reasoning on probabilistic conditional independence (PCI) in Coq based on
universal algebraic structure suitable for studying PCI relations called cains
(derived from \emph{ca}usal \emph{in}ference) and developed
in~\cite{Wang:2010}. What we present in this work could be a starting point to
make a link between~\cite{Wang:2010} and Pearl's d-separation in Coq.  The Coq
code developed by J.P. Chancelier for proving the results exposed in this paper
is publicly available on GitHub\footnote{at URL
  \texttt{https://github.com/jpc-cermics/relations.git}} and counts around
$7000$ lines of code using
Mathcomp/SSReflect~\cite{MathComp:2022,Gonthier-Mahboubi-Tassi:2016}.

The paper is organized as follows. In
Sect.~\ref{Conditional_active_relation_induced_by_endpoints_of_non-blocking_paths},
we revisit graphs as binary relations and define \undirectedEdgePaths;
then, we present our extended definitions of active \undirectedEdgePaths\ and of d-separation.
In
Sect.~\ref{Characterization__of_the_conditional_directional_separation_binary_relation},
we state and sketch the proof of our main result, the characterization of the
d-separation relation as the complementary of the conditional active relation.
The main body of the proof is to be found in the Appendices, which follow its sketch.
In Appendix~\ref{Comments_on_Coq_in_the_appendices_proofs}, we comment on how Coq is used
in parallel to mathematical proofs.
In Appendix~\ref{appendix_The_star_conditional_active_relation}, we show that
the conditional active relation can be replaced by 
the star conditional active relation in the statement of our main result.
In Appendix~\ref{Proof_of_Theorem_implies}, we show that
the star conditional active relation is included in
the complementary of the d-separation relation.
In Appendix~\ref{Proof_of_Theorem_isimplied}, we show the reverse inclusion.

\section{A formal Pearl's d-separation definition}
\label{Conditional_active_relation_induced_by_endpoints_of_non-blocking_paths}

In~\S\ref{Binary_relations_and_graphs}, we deal with graphs but using the
concepts of binary relations.
In~\S\ref{UndirectedEdgePaths_in_a_graph}, we formally define what we call
\undirectedEdgePaths\ in a graph and discuss the Coq implementation used to
formalize \undirectedEdgePaths.  Thus equipped,
in~\S\ref{Definition_of_active_undirectedEdgePaths_and_of_conditional_directional_separation},
we formally adapt Pearl's definition of active (and blocked)
\undirectedEdgePaths\ in a graph, from which we deduce the (conditional)
d-separation binary relation.

\subsection{Binary relations and graphs}
\label{Binary_relations_and_graphs}

We employ the vocabulary and concepts both of binary relations and of graph
theory.
We denote by~$\NN$ the set of natural numbers (including zero),
and \( \NN^* =\NN\setminus\na{0} \).
We use the notation $\ic{r,s}=$ $\{$ $r$, $r+1$, $\ldots$, $s-1$, $s\}$ for two natural numbers $r \leq s$.

In~\S\ref{defs:binary-rel}, we provide background on binary relations.
In~\S\ref{sec:sequences}, we list tools that will be useful
to navigate between vertices, edges and pair of edges representations in a
graph,
as defined in~\S\ref{Graphs_as_binary_relations}.

\subsubsection{Background on binary relations}
\label{defs:binary-rel}

Let $\AGENT$ be a nonempty set (finite or not).  We recall that a \emph{(binary)
  relation}~$\relation$ on~$\AGENT$ is a subset
$\relation \subset \AGENT\times\AGENT $ and that
\( \bgent\, \relation\, \cgent \) means \( \np{\bgent,\cgent} \in \relation \).
For any subset \( \Bgent \subset \AGENT \), the \emph{(sub)diagonal relation} is
\( \Delta_{\Bgent} = \bset{ \np{\bgent,\cgent} \in \AGENT\times\AGENT }%
{ \bgent=\cgent \in \Bgent } \) and the \emph{diagonal relation} is
\( \Delta=\Delta_{\AGENT} \). A relation is \emph{reflexive} if \(\Delta \subset \relation\).
A \emph{foreset} of a relation~$\relation$ is any set of the form
\( \relation \, \cgent = \defset{ \bgent \in \AGENT }{ \bgent\, \relation \,
  \cgent } \), where \( \cgent \in \AGENT \), or, by extension, of the form
\( \relation \, \Cgent = \defset{ \bgent \in \AGENT }{ \exists \cgent \in \Cgent \eqsepv
  \bgent\, \relation \, \cgent } \), where \( \Cgent \subset \AGENT \).
An \emph{afterset} of a relation~$\relation$ is
any set of the form \( \bgent \, \relation = 
\defset{ \cgent \in  \AGENT }{ \bgent\, \relation \, \cgent } \),
where \( \bgent \in \AGENT \), 
or, by extension, of the form \( \Bgent \, \relation = 
\defset{ \cgent \in  \AGENT }{ \exists \bgent \in \Bgent \eqsepv \bgent\,
  \relation \, \cgent } \), where \( \Bgent \subset \AGENT \).
The \emph{opposite} or \emph{complementary~$\Complementary{\relation}$} of a binary
relation~$\relation$ is the relation~$\Complementary{\relation}=\AGENT\times\AGENT\setminus\relation$,
that is, defined by \( \bgent\, \relation^{\mathsf{c}} \, \cgent \iff 
\neg \np{ \bgent\, \relation \, \cgent } \).
The \emph{converse~$\Converse{\relation}$} of a binary relation~$\relation$ is
defined by \( \bgent\, \Converse{\relation} \, \cgent \iff \cgent\, \relation \, \bgent
\) (and $\relation$ is  \emph{symmetric} if \( \Converse{\relation}=\relation \)).
The \emph{composition}
$\relation\relation'$ of two
binary relations~$\relation, \relation'$ on~$\AGENT$ is defined by
\( \bgent (\relation\relation') \cgent \iff
\exists \delta \in  \AGENT \), \( \bgent\, \relation \, \delta \)
and \( \delta\, \relation' \, \cgent \);
then, by induction we define\footnote{%
  In what follows, when we consider a binary relation
  as a subset $\relation \subset \AGENT\times\AGENT $, we will use the notation
  \( \SetProd{\relation}{n} \subset \SetProd{\AGENT\times\AGENT}{n}\), where $n$ is
  a positive integer, to
  denote a product subset of the product set~\( \AGENT^{2n} \),
  thus making the distinction with the binary relation
  $\relation^{n} \subset \AGENT\times\AGENT $ obtained by $n$~compositions.}
\( \relation^{n+1}=\relation\relation^{n} \) for \( n \in \NN^* \). 
The \emph{transitive closure} of a binary relation~$\relation$ is
\( \TransitiveClosure{\relation} = \cup_{k=1}^{\infty} \relation^{k} \)
(and $\relation$ is  \emph{transitive} if \( \TransitiveClosure{\relation}=\relation \))
and the \emph{reflexive and transitive closure} is 
\( \TransitiveReflexiveClosure{\relation}= \TransitiveClosure{\relation} \cup
\Delta = \cup_{k=0}^{\infty} \relation^{k} \) with the convention $\relation^0=\Delta$. 
A \emph{partial equivalence relation} is a symmetric and transitive binary
relation (generally denoted by~$\sim$ or~$\equiv$).
An \emph{equivalence relation} is a reflexive, symmetric and transitive binary
relation.
\bigskip

\begin{leftbar}
  Binary relations are implemented as sets on a product space using the classical sets implemented in
  \texttt{classical\_sets.v} from the Coq mathcomp library~\cite{MathComp:2022}
  using SSReflect tactics~\cite{Gonthier-Mahboubi-Tassi:2016}
\end{leftbar}
\begin{alectryon}
  {\small \Snippetrelation{.\footnote{At the end of a Coq statement, ended by a
        dot belonging to the Vernacular (the language of Coq commands)
        we add a dot or a comma which serve as text punctuation.}}}
\end{alectryon}
\begin{leftbar}
  \noindent As described in more details below, we have developed a library for relations taking into account all the
  definitions recalled at the beginning of~\S\ref{defs:binary-rel}.
\end{leftbar}

\subsubsection{Sequences, sets and binary relations}
\label{sec:sequences}
\newcommand{\fullset}{T}

Being an active extended oriented path, as defined later, involves mixed
properties of vertices, oriented edges and successive pairs of oriented edges
path. This is why we need to develop tools that permit to navigate between
vertices, edges and pair of edges representations. This part is devoted to list
these tools and some of their properties.  To formalize graph paths, we use
sequences (as defined in mathcomp \texttt{seq.v}) combined with set
formalization (defined in mathcomp \texttt{classical\_sets.v}).

\myparagraph{$\bullet$ \protect\pinX.} We consider a set $\fullset$ and, for any
subset ${\Set} \subset \fullset$ and $n\in \NN$, we denote by
$\Sequence_n(\Set)=\SetProd{\Set}{n}$ the set of sequences of length~$n$ of
elements of the set~$\Set$ ($\Sequence_0(\Set)$ being the singleton set with the
empty sequence) and by $\Sequence_{\ge n}(\Set)$ the set of finite sequences of
length greater than or equal to $n$ of elements of the set $\Set$, that is, the
disjoint union\footnote{%
  The symbol~$\sqcup$ stands for a disjoint union.}
$\sqcup_{k\ge n} \Sequence_k(\Set)$. The largest set $\Sequence_{\ge 0}(\Set)$ will be
denoted by $\Sequence(\Set)$:
\begin{equation}
  \Sequence_{\ge n}(\Set)=\bigsqcup_{k\ge n} \Sequence_k(\Set) \mtext{ and }
  \Sequence(\Set)=  \Sequence_{\ge 0}(\Set)
  \eqfinp
  \label{eq:set_of_sequences}
\end{equation}
\medskip

\begin{leftbar} The sets $\Sequence_{\ge n}(\fullset)$ and
  $\Sequence_{n}(\fullset)$ are formalized in Coq in the mathcomp library as
  sequences of elements of type \texttt{T} and the restriction to elements in a
  subset \texttt{(\PY{n+nv}{D}: set T)} is obtained using the function
  \texttt{all} (in mathcomp library~\texttt{seq.v}).  As an example, the set
  $\Sequence_{n}(D)$ is implemented as follows
  \begin{alectryon}
    {\small \Snippetallnotation}\\
    {\small \SnippetSn}.
  \end{alectryon}
\end{leftbar}

\myparagraph{$\bullet$ (\protect\Liftp) and (\protect\Liftpin).}  Then, we define a
lift operator
\[ \Lift: \Sequence(\fullset) \to \Sequence(\fullset{\times}\fullset)
  \eqfinv\] such that
\begin{itemize}
\item 
  for all $n \ge 2$, the restriction of the operator $\Lift$ to the set
  $\Sequence_{n}(\fullset)$ coincides with the following mapping
  \(\Lift_n: \Sequence_{n}(\fullset)\to \Sequence_{n{-}1}(\fullset{\times}\fullset)\),
  given by
  \begin{align}
    \forall (\vertex_1,\dots, \vertex_n)\in \SetProd{\fullset}{n}
    \eqsepv
    \Lift_n(\vertex_1,\dots, \vertex_n) = \bp{\np{\vertex_1, \vertex_2},\np{\vertex_2, \vertex_3},\dots,
    \np{\vertex_{n-1},\vertex_n}}
    \eqfinv
  \end{align}
  transforming a sequence of elements of $\fullset$ of length $n$
  into a sequence of oriented pairs in $\fullset{\times}\fullset$ of length $n{-}1$,
\item 
  the restriction of the operator $\Lift$ on
  $\Sequence_{0}(\fullset) \cup\Sequence_{1}(\fullset)$ is the constant mapping
  giving the empty list on $\fullset{\times}\fullset$.
\end{itemize}

\begin{leftbar} The lift operator $\Lift$ is implemented in Coq as a recursive
  mapping denoted \texttt{Lift}:
  \begin{alectryon}  {\small \SnippetLift.}\end{alectryon}
  We note that, thanks to Coq polymorphism, the lift mapping is parameterized by
  a type and thus can be used to lift a sequence of vertices into a sequences of
  edges, but also to lift a sequence of edges (in $\fullset{\times}\fullset$) into a
  sequence of ordered pairs of edges (in $(\fullset{\times}\fullset)^2$).
\end{leftbar}

The notation $\protect\Liftpin$ is used to denote the
expression {\PY{n}{(Lift p)}~\PY{o}{$[{\in}]$}~\PY{n}{R}}.

\myparagraph{$\bullet$ \protect\pSucR.} We must be able to check that successive
elements of a sequence whose elements are in $\fullset$ belong to a given subset
$R$ of $\fullset{\times}\fullset$, that is, satisfy a relation $R$ on $\fullset$.

\begin{leftbar} This is easily implemented with the help of an inductive
  predicate (\texttt{RPath} in \texttt{seq1.v})
  \begin{alectryon}{\small \SnippetRPath,}\end{alectryon}
  \noindent that we do not detail here as we prove that it can be equivalently implemented with the Lift mapping
  (which enables more computational proofs) as we have
  \begin{alectryon}{\small \SnippetRPathequiv.}\end{alectryon}
\end{leftbar}

As a first example, consider the (chain) relation $C_{\textsc{H}}$ -- denoted by
\texttt{Chrel} in Coq -- defined by
$((v_1, v_2)\, C_{\textsc{H}}\, (v_3, v_4) \iff v_2 = v_3)$, on the product set
$T{\times}T$
\begin{alectryon}{\small \SnippetChrel.}\end{alectryon}
Now, the fact that lifted sequences are well chained sequences can be stated as proving the following Coq Lemma
\begin{alectryon}{\small \SnippetLiftSuc.}\end{alectryon}

As a second example, if the elements of a sequence belong to a set $X$, then the
elements of the lifted sequence belong to the product relation $X{\times}X$ as proved
in the following lemma
\begin{alectryon}{\small \SnippetRpathLone.}\end{alectryon}

\myparagraph{$\bullet$ \texttt{Lift} bijection.}
The lift operation, when restricted to the subset \texttt{D} defined below, is
bijective onto its image \texttt{I}
\begin{leftbar}
  \begin{alectryon}
    {\small \SnippetDI\\\SnippetLiftinj\\\SnippetLiftsurj.}
  \end{alectryon}
\end{leftbar}
\noindent Moreover, the inverse of \texttt{Lift} is explicitely obtained by a recursive mapping \texttt{UnLift} (not detailed here).

\subsubsection{Graphs as binary relations}
\label{Graphs_as_binary_relations}

Let $\VERTEX$ be a nonempty set (finite or not), whose elements are called
\emph{vertices}.  Let \( \EDGE \subset \VERTEX\times\VERTEX \) be a relation
on~$\VERTEX$, whose elements are ordered pairs (that is, couples) of vertices
called \emph{edges}.  The first element of an edge is the \emph{tail of the
  edge}, whereas the second one is the \emph{head of the edge}.  Both tail and
head are called \emph{endpoints} of the edge, and we say that the edge connects
its endpoints.  We define a \emph{loop} 
as an element of \( \Delta \cap \EDGE \), that is, a loop is an edge that connects a
vertex to itself.

A \emph{graph}, as we use it throughout this paper, is a
couple~$(\VERTEX,\EDGE)$.  This definition is basic, and we now stress
proximities and differences with classic notions in graph theory.  As we define
a graph, it may hold a finite or infinite number of vertices; there is at most
one edge that has a couple of ordered vertices as single endpoints, hence a
graph (in our sense) is not a multigraph (in graph theory); loops are not
excluded (since we do not impose $\Delta \cap \EDGE=\emptyset$).  Hence, what we call a graph
would be called a directed simple graph permitting loops in graph theory.

To define blocked and active \undirectedEdgePaths\ -- an essential notion in
causal inference -- relative to the graph~\( \npOrientedGraph \), we need to fix
additional vocabulary and notation.  In the graph~$(\VERTEX,\EDGE)$, the
\emph{undirected edges} are the elements of $\EDGE \cap \Converse{\EDGE}$ --- that
is, edges with both $(\cgent,\bgent)\in \EDGE$ and $(\bgent, \cgent)\in \EDGE$
(hence, including loops).  Then, the graph~$(\VERTEX,\EDGE)$ is said to be
\emph{undirected} if all edges are undirected edges, or, equivalently, if
\( \EDGE=\EDGE \cap \Converse{\EDGE} \) or if \( \Converse{\EDGE} = \EDGE \).  The
\emph{undirected extension} of a graph~$(\VERTEX,\EDGE)$ is the
graph~$(\VERTEX,\EDGE \cup \Converse{\EDGE})$.

In the graph~$(\VERTEX,\EDGE)$, the \emph{directed edges} are the elements of
$\EDGE \cap \npComplementary{ \Converse{\EDGE} }$ --- that is, edges with
$(\cgent,\bgent)\in \EDGE$ such that $(\bgent, \cgent)\not\in \EDGE$ (recall that we
do not assume that $\EDGE \cap \Converse{\EDGE}=\emptyset$).  Then, the
graph~$(\VERTEX,\EDGE)$ is said to be \emph{directed} if all edges are directed
edges, or, equivalently, if \( \EDGE \cap \Converse{\EDGE}=\emptyset \), that is, when no
two edges have the same endpoints.  \medskip

\begin{leftbar} 
  A graph $(\VERTEX,\EDGE)$ is given in Coq by an oriented pair composed of a
  type \texttt{(T: Type)} and a relation on \texttt{T}, that is
  $(\EDGE:\texttt{relation T})$ (which is equivalent to a set declaration
  $(\EDGE:\texttt{set T*T})$).  Thus, the (classical) set definition of mathcomp
  analysis \texttt{classical\_sets.v} is used to formalize a graph.  We have not
  used the Coq package \texttt{graph-theory} to formalize graph as we did not
  want to stick to finite graphs.
\end{leftbar}

\subsection{\UndirectedEdgePaths\ in a graph}
\label{UndirectedEdgePaths_in_a_graph}

In graph theory, one finds the notions of path, chain and walk.  To avoid
ambiguities, we formally define
in~\S\ref{Edge_paths_in_a_(directed_simple)_graph_(permitting_loops)} an
\EdgePath\ in a graph --- in our sense, that is, a (directed simple) graph
(permitting loops) --- as the classical notion of path in a graph~\citep{Diestel}.
Then,
in~\S\ref{Extended-oriented_paths_in_a_(directed_simple)_graph_(permitting_loops)},
we define an \undirectedEdgePath\ in a graph as what corresponds to a chain path
in~\citep{Lauritzen-et-al-1990}.
We consider a graph~$(\VERTEX,\EDGE)$ as defined
in~\S\ref{Graphs_as_binary_relations}, that is, a (directed simple) graph
(permitting loops).

\subsubsection{Edge paths in a graph} 
\label{Edge_paths_in_a_(directed_simple)_graph_(permitting_loops)}

After defining \EdgePaths\ and their endpoints, we introduce deployments in
\EdgePaths\ 
(see the summary Table~\ref{tab:Notions_for_EdgePath_s_in_a_graph}).

\myparagraph{$\bullet$ Definition of \EdgePath s.}
We define the set of
\emph{\EdgePaths\ of length~$n$} ($n\ge 1$), relative to the
graph~$\npOrientedGraph$, by
\begin{subequations}
  \label{def:epaths_defs}
  \begin{align}
    \PATH_n \np{\OrientedGraph} 
    &=
      \Bset{ \bseqa{ \np{\tail{\vertex_i},\head{\vertex_i}} }{i\in\ic{1,n}} \in \SetProd{\EDGE}{n} }
      { \head{\vertex_i}=\tail{\vertex_{i+1}} \text{ for } i\in\ic{1,n{-}1}}
      \eqfinp
      \label{def:epaths_gtzero_}
  \end{align}
  For $n\ge1$, we define the set of 
  \emph{\EdgePaths\ of length greater than~$n$},
  relative to the graph~$\npOrientedGraph$,
  by 
  \begin{align}
    \PATH_{> n}\np{\OrientedGraph}
    &=
      \sqcup_{n'> n} \PATH_{n'}\np{\OrientedGraph}
      \eqfinv
      \intertext{and finally the set of \emph{\EdgePaths}, relative to the graph~$\npOrientedGraph$, by}
      \PATH\np{\OrientedGraph}
    &=
      \PATH_{> 0}\np{\OrientedGraph}
    \eqfinp
    \label{def:epaths_gtzero}
  \end{align}
\end{subequations}
  
\begin{leftbar} Using the tools introduced in \S\ref{sec:sequences} and
  Equations~\eqref{def:epaths_defs} we obtain the following formalization of
  $\PATH_{> n}$
  \begin{alectryon}{\small \SnippetEpathgt,}\end{alectryon}
  \noindent where \texttt{E} is the edge relation $\EDGE$ and where \texttt{Chrel}
  was defined in~\S\ref{sec:sequences}.
\end{leftbar}

We denote by $\cardinal{\path}$ the length of an \EdgePath~$\path \in
\PATH\np{\OrientedGraph}$ (computed in Coq by the mapping \texttt{size}).
An \emph{\EdgeSubPath} of the \EdgePath~$\path$ is an \EdgePath\ obtained by a
subsequence of consecutive indices.

\myparagraph{$\bullet$ Definition of endpoints of \EdgePath s.} The first
element~$\tail{\vertex_1}$ of an \EdgePath\
\( \path=\bseqa{ \np{\tail{\vertex_i},\head{\vertex_i}} }{i\in\ic{1,n}} \) is the
\emph{tail of the \EdgePath}, whereas the last one~$\head{\vertex_n}$ is the
\emph{head of the \EdgePath}.  Both tail (obtained with function \texttt{head}
in Coq) and head (obtained with function \texttt{last} in Coq) are called
\emph{endpoints} of the \EdgePath.

\begin{subequations}
  We define the \emph{projection mapping
    \( \Projection^n: \PATH_n \np{\OrientedGraph} \to\VERTEX\times\VERTEX \) on the
    tail and head endpoints of an \EdgePath\ of length~$n$} by
  \begin{equation}
    \forall
    \path=(\tail{\vertex_i},\head{\vertex_i})_{i\in\ic{1,n}}\in\PATH_n\np{\OrientedGraph}
    \eqsepv 
    \Projection^n\np{\path}=
    \Projection^n\np{(\tail{\vertex_i},\head{\vertex_i})_{i\in\ic{1,n}}} 
    = \np{\tail{\vertex_1},\head{\vertex_{n}}} \in \VERTEX\times\VERTEX 
    \eqfinp
  \end{equation}
  We define the \emph{projection mapping
    \( \Projection: \PATH \np{\OrientedGraph} \to \VERTEX\times\VERTEX \) on the tail
    and head endpoints of an \EdgePath} by
  \begin{equation}
    \forall \path\in \PATH\np{\OrientedGraph}
    \eqsepv
    \Projection\np{\path} = \Projection^{\cardinal{\path}}\np{\path} \in \VERTEX\times\VERTEX 
    \eqfinp
    \label{eq:Projection_PATH_EdgePath}
  \end{equation}
  We also distinguish \emph{the tail and the head endpoints projection mappings
    of an \EdgePath} by (see Figure~\ref{fig:tailhead_EdgePath})
  \begin{equation}
    \Projection=\np{\Tail{\Projection},\Head{\Projection}}
    \mtext{ where }
    \Tail{\Projection}: \PATH \np{\OrientedGraph} \to \VERTEX
    \mtext{ and }
    \Head{\Projection}: \PATH \np{\OrientedGraph} \to \VERTEX
    \eqfinp 
    \label{eq:Projection_PATH_EdgePath_first_and_last}
  \end{equation}  
    \label{eq:Projection_PATH_EdgePath_all}
\end{subequations}

\begin{figure}[hbtp]
  \begin{center}
    \fbox{\includegraphics[width=0.5\textwidth]{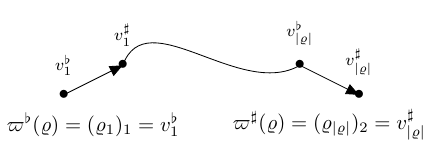}}
    \caption{Tail and head endpoints of an \EdgePath\ projection
      mappings~\eqref{eq:Projection_PATH_EdgePath_all} for $\path \in \PATH\np{\OrientedGraph}$ 
      \label{fig:tailhead_EdgePath}}
  \end{center}
\end{figure}

\begin{leftbar} The endpoints are obtained in Coq by the mapping \texttt{Pe}
  (meaning path endpoints) for a sequence of vertices, and by the mapping
  \texttt{Epe} (meaning extended path endpoints) for a sequence of edges
  \begin{alectryon}{\small \SnippetPe\SnippetEpe.}\end{alectryon}
  \noindent We prove in the next two lemmata that \texttt{Pe} and \texttt{Epe} behave properly with respect to the \texttt{Lift} bijection between \texttt{D} and \texttt{I}
  \begin{alectryon}{\small \SnippetEpeLift\SnippetPeUnLift.}\end{alectryon}
\end{leftbar}

As a first result linking (edge) paths and relations, we prove that
\begin{alectryon}  {\small \SnippetTCP}\end{alectryon}
which asserts that two nodes
\texttt{(v1,v2)} are in relation through the transitive closure of a
relation~\texttt{E}, that is \texttt{(v1,v2) $\in$ E.+} if and only if there
exists an edge path with endpoints \texttt{(v1,v2)} in the graph~$(\VERTEX,\EDGE)$ represented by
\texttt{(T: Type),(E: relation T)}.

\myparagraph{$\bullet$ Concatenation of \EdgePath s.} Concatenation of sequences 
denoted by the infix operator $\ltimes$ (and denoted by \texttt{++} in Coq) is easily defined and is associative.
When considering \EdgePaths, concatenation of 
\( \pathbis \in \PATH\np{\OrientedGraph}\) and
\( \pathter \in \PATH\np{\OrientedGraph}\) gives a sequence $(\pathbis \ltimes \pathter) \in \Sequence(\VERTEX{\times}\VERTEX)$ which belongs to
$\PATH\np{\OrientedGraph}$ under the additional assumption that 
\( \Head{\Projection}\np{\pathbis}=\Tail{\Projection}\np{\pathter} \) (see
Equation~\eqref{eq:Projection_PATH_EdgePath_first_and_last}), that is, 
\begin{equation}
  \forall\pathbis \in \PATH\np{\OrientedGraph}\eqsepv
  \forall \pathter \in \PATH\np{\OrientedGraph} \eqsepv
  \Head{\Projection}\np{\pathbis}=\Tail{\Projection}\np{\pathter}
  \implies (\pathbis \ltimes \pathter) \in \PATH\np{\OrientedGraph}
  \eqfinp
  \label{eq:concatenation_Path}
\end{equation}

\myparagraph{$\bullet$ Deployment in \EdgePaths.}
With any binary relation
\( \relation \subset \VERTEX{\times}\VERTEX \), we associate the
subset~\( \DeploymentInPaths{\relation}{\OrientedGraph} \)
of~\( \PATH\np{\OrientedGraph} \), that we call the \emph{deployment in
  \EdgePaths}, defined by
\begin{equation}
  \forall \relation \subset \VERTEX\times\VERTEX 
  \eqsepv 
  \DeploymentInPaths{\relation}{\OrientedGraph}
  = \Converse{\Projection}(\relation) \subset \PATH\np{\OrientedGraph} 
  \eqfinv
  \label{eq:DeploymentInPaths}
\end{equation}
where the projection~\( \Projection \) has been defined
in~\eqref{eq:Projection_PATH_EdgePath_all}.  The deployment in \EdgePath
s~\( \DeploymentInPaths{\relation}{\OrientedGraph} \) is made of the \EdgePaths\
whose endpoints satisfy the binary relation~\( \relation \).

\begin{leftbar}
  \begin{alectryon}  {\small \SnippetDP}\end{alectryon}
  It is to be noted that the deployment in \EdgePaths\ may be obtained as the
  image by the Lift mapping of a subset of sequences of vertices as follows
  \begin{alectryon}{\small \SnippetDV\SnippetDPDV.}\end{alectryon}
  Note that, when an edge path of length greater than zero is given as a lifted
  sequence of elements of~$T$ (as for example in
  \texttt{\PY{n}{Lift}~\PY{o}{(}\PY{n}{x}\PY{o}{::(}\PY{n}{rcons}~\PY{n}{p}~\PY{n}{y}\PY{o}{))}~\PY{o}{=}~\PY{n}{st}\PY{o}{.}}),
  the endpoints \texttt{(x,y)} and the intermediate nodes \texttt{p} of the
  \EdgePath\ \texttt{st} are immediately obtained.
\end{leftbar}

\begin{table}
  \centering
  \begin{tabular}{||c|c|c||}
    \hline\hline
    Name
    & Expression
    & Equation 
    \\ \hline\hline
    set of (edge) paths
    &
      \( \PATH\np{\OrientedGraph} \)
    & Equation~\eqref{def:epaths_gtzero}
    \\  \hline         
    tail and head endpoints
    & 
      \( \Projection=\np{\Tail{\Projection},\Head{\Projection}} \)
    &
      Equation~\eqref{eq:Projection_PATH_EdgePath}
    \\
    projection mappings
    &
      \(  \Tail{\Projection}: \PATH \np{\OrientedGraph} \to \VERTEX \)
    &
      Equation~\eqref{eq:Projection_PATH_EdgePath_first_and_last}
    \\
    & \( \Head{\Projection}: \PATH \np{\OrientedGraph} \to \VERTEX \)
    & Equation~\eqref{eq:Projection_PATH_EdgePath_first_and_last}
    \\  \hline
    deployment in \EdgePaths\
    & \( \DeploymentInPaths{\cdot}{\OrientedGraph} = \Converse{\Projection} \)
    & Equation~\eqref{eq:DeploymentInPaths}
    \\ \hline \hline         
  \end{tabular}
  \caption{Notions for \EdgePaths\ in a graph (\S\ref{Edge_paths_in_a_(directed_simple)_graph_(permitting_loops)})}
  \label{tab:Notions_for_EdgePath_s_in_a_graph}
\end{table}

\subsubsection{\UndirectedEdgePaths\ in a graph}
\label{Extended-oriented_paths_in_a_(directed_simple)_graph_(permitting_loops)}

To define \emph{\undirectedEdgePaths}, we consider a set
\begin{equation}
  {\Orientation}=\na{-1,+1}
  \eqfinv
\end{equation}
with two elements, and implemented in Coq as an inductive type taking two values
\texttt{N} (for~$-1$) and \texttt{P} (for~$1$)
\begin{alectryon}{\small \SnippetOO,}\end{alectryon}
\noindent which will serve as an orientation specification of an edge. 
We also introduce the set $\EDGEo \subset \VERTEX\times\VERTEX{\times}\Orientation$ defined by 
\begin{align}
  (\vertex,\vertexbis,o) \in \EDGEo
  \iff (\vertex, \vertexbis) \in \EDGE^{(o)}
  \eqfinv
  \label{def:Eo_relation}
\end{align}
where \( \EDGE^{(+1)}= \EDGE \) and \( \EDGE^{(-1)}= \Converse{\EDGE} \).

After defining \undirectedEdgePaths\ and their endpoints, we introduce
deployments in \undirectedEdgePaths\ (see the summary
Table~\ref{tab:Notions_for_undirectedEdgePath_s_in_a_graph}).

\myparagraph{$\bullet$ Definition of \undirectedEdgePaths.}
We define the set of \emph{\undirectedEdgePaths\ of length~$n$} ($n\ge 1$),
relative to the graph~$\npOrientedGraph$, by 
\begin{subequations}
  \label{eq:eopaths_defs}
  \begin{align}
  \UPATH_n \np{\OrientedGraph} 
  &=
    \Bset{ \bseqa{ \np{\tail{\vertex_i},\head{\vertex_i}, o_i} }{i\in\ic{1,n}} 
    \in
    \SetProd{\EDGEo}{n}}
    {\head{\vertex_i}=\tail{\vertex_{i+1}} \text{ for } i\in\ic{1,n{-}1}}
    \eqfinv
    \label{def:eopaths_gtzero_}
  \end{align}
  For $n\ge1$, we define the set of 
  \emph{\undirectedEdgePath s\footnote{%
  It is to be noted that an \undirectedEdgePath\ is \emph{not a path in the graph
    $\npOrientedGraph$}, neither in the undirected graph $\npUnorientedGraph$.
  However, considering a couple $(\path,\orient)\in \UPATH_n\np{\OrientedGraph}$, we obtain that
  $\path \in \PATH_n\np{\UnorientedGraph}$,
  that is, $\path$ is an (edge) path in the unoriented graph
  \( \npUnorientedGraph \).
  We thus obtain a natural surjection
  $(\path,\orient)\mapsto \path$
    from $\UPATH\np{\OrientedGraph}$ to $\PATH\np{\UnorientedGraph}$.
    This canonical surjection is not necessary injective because
    a path in $\PATH\np{\UnorientedGraph}$ that has an edge in $\EDGE \cap \Converse{\EDGE}$
    is the image of two distinct \undirectedEdgePaths.
    The surjection $(\path,\orient)\mapsto \path$ 
    is a bijection in the special case when 
    the graph~\( \npOrientedGraph \) is directed,
    that is, when $\EDGE \cap \Converse{\EDGE}=\emptyset$, that is, 
    when no two edges have the same endpoints. 
  }
 of length greater than~$n$},
  relative to the graph~$\npOrientedGraph$,
  by 
  \begin{align}
    \UPATH_{> n}\np{\OrientedGraph}
    &=
      \sqcup_{n'> n} \UPATH_{n'}\np{\OrientedGraph}
      \eqfinv
      \intertext{and finally the set of \emph{\undirectedEdgePaths}, relative to the graph~$\npOrientedGraph$, by}
      \UPATH\np{\OrientedGraph}
    &=
      \UPATH_{> 0}\np{\OrientedGraph}
    \eqfinp
    \label{def:eopaths_gtzero}
  \end{align}
  \end{subequations}
  \begin{leftbar}
    Using the tools introduced in \S\ref{sec:sequences} and
    Equations~\eqref{eq:eopaths_defs}, we obtain the following Coq formalization of
    $\UPATH_{> n}\np{\OrientedGraph}$
    \begin{alectryon}
      {\small \SnippetUgt\SnippetOedge\SnippetChrelO,}\end{alectryon}
    \noindent where \texttt{Oede E} is used to formalize the $\EDGEo$~subset
    in~\eqref{def:Eo_relation}, and where \texttt{ChrelO} is the chain relation on 
    $\VERTEX\times\VERTEX{\times}\Orientation$.
  \end{leftbar}

  An \undirectedEdgePath\  can be decomposed as an oriented pair composed of an \EdgePath\ 
  and a sequence of orientations.
  \begin{subequations}
    For that purpose, we introduce the mapping
    \begin{align}
      \pi : \Sequence(\VERTEX{\times}\VERTEX{\times}\Orientation)
      &\to 
        \Sequence(\VERTEX\times\VERTEX)\times \Sequence(\Orientation)
        \eqfinv
        \intertext{where $\Sequence(\cdot)$ was defined in~\eqref{eq:set_of_sequences}, 
        given by}
        \forall \UndirectedPath \in
        \Sequence(\VERTEX{\times}\VERTEX{\times}\Orientation)
        \eqsepv 
        \pi\np{\UndirectedPath}
      &= 
        \bp{ \PiSVV(\UndirectedPath),
        \PiSO(\UndirectedPath)}
        \in \Sequence(\VERTEX\times\VERTEX)\times \Sequence(\Orientation)
        \eqfinv
      \\
      \textrm{where } 
      \PiSVV(\UndirectedPath)
      &=
        \nseqa{ \PiVV(\UndirectedPath_i)}{i \in\ic{1,\cardinal{\UndirectedPath}}} 
        \in   \Sequence(\VERTEX\times\VERTEX)
        \eqfinv
        \label{eq:def_pi_VERTEXtimesVERTEX}
      \\
      \textrm{and } 
      \PiSO(\UndirectedPath)
      &=
        \nseqa{ \PiO (\UndirectedPath_i)}{i \in\ic{1,\cardinal{\UndirectedPath}}}
        \in \Sequence(\Orientation)
        \eqfinv
    \end{align}
    where $\PiVV$ (resp. $\PiO$) is the projection
    from the set $\VERTEX{\times}\VERTEX{\times}\Orientation$ onto the set $\VERTEX{\times}\VERTEX$ 
    (resp. $\Orientation$).
    When $\UndirectedPath \in \UPATH\np{\OrientedGraph}$, we obtain that
    $\path = \PiSVV(\UndirectedPath)$ is an \EdgePath\, that is
    $\path \in  \PATH\np{\OrientedGraph}$.
    \label{eq:def_pi}
  \end{subequations}

  Reciprocally, given an \EdgePath\ $\path \in \PATH\np{\OrientedGraph}$ and
  a sequence~$o\in \Orientation^{\cardinal{\path}}$ of orientations of the same size, we denote by
  $\UndirectedPath=\Pair{\path,o}$ the \undirectedEdgePath\,
  $\UndirectedPath \in \UPATH\np{\OrientedGraph}$, defined by
  \begin{equation}
    \Pairname: (\path, o) \in \text{Im} \pi  \mapsto \UndirectedPath \text{ with }
    \UndirectedPath_i = (\path_i, o_i)\eqsepv \forall i\in \ic{1, \cardinal{\path}}.
    \label{def:pair}
  \end{equation}
  The mapping $\Pairname$ is the inverse of the mapping~$\pi$, defined in
  Equation~\eqref{eq:def_pi}, \emph{but on the range of the mapping~$\pi$} (so that
  using~$\Pairname$ is a slight abuse of notation).
  
  An \emph{\undirectedEdgeSubPath} of the \undirectedEdgePath~$\UndirectedPath
  \in\UPATH\np{\OrientedGraph}$ is an \undirectedEdgePath\ obtained by a
  subsequence of consecutive indices. 
  
  \myparagraph{$\bullet$ Definition of endpoints of \undirectedEdgePaths.}
  The endpoints of an \undirectedEdgePath\
  \( \UndirectedPath \) are defined as the endpoints of the \EdgePath\
  $\PiSVV (\UndirectedPath)$ as defined in~\eqref{eq:def_pi_VERTEXtimesVERTEX}.
  
  We define the \emph{projection mapping
  \( \Projection_{\UPATH}: \UPATH \np{\OrientedGraph} \to\VERTEX\times\VERTEX \) 
  on the tail and head endpoints of \undirectedEdgePaths} by
  \begin{subequations}
    \begin{equation}
      \forall \UndirectedPath \in \UPATH\np{\OrientedGraph}
      \eqsepv
      \Projection_{\UPATH}\np{\UndirectedPath} = \Projection \bp{ \PiSVV(\UndirectedPath)} \in \VERTEX\times\VERTEX 
      \eqfinv
      \label{eq:Projection_UPATH}
    \end{equation}
    where the projection mapping
    \( \Projection: \PATH \np{\OrientedGraph} \to \VERTEX\times\VERTEX \) 
    on the tail and head endpoints of an \EdgePath\ has been introduced in~\eqref{eq:Projection_PATH_EdgePath_all}.
    We also distinguish \emph{the tail and the head endpoints projection mappings
      on \undirectedEdgePaths} by
    (see Figure~\ref{fig:tailhead_undirectedEdgePath})
    \begin{equation}
      \Projection_{\UPATH}=\np{\Tail{\Projection}_{\UPATH},\Head{\Projection}_{\UPATH}}
      \mtext{ where }
      \Tail{\Projection}_{\UPATH}: \UPATH \np{\OrientedGraph} \to \VERTEX
      \mtext{ and }
      \Head{\Projection}_{\UPATH}: \UPATH \np{\OrientedGraph} \to \VERTEX
      \eqfinp 
      \label{eq:Projection_UPATH_first_and_last}
    \end{equation}  
    \label{eq:Projection_UPATH_all}
\end{subequations}

\begin{figure}[hbtp]
  \begin{center}
    \fbox{\includegraphics[width=0.5\textwidth]{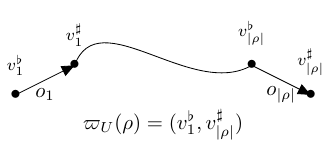}}
    \caption{
      Projection mappings~\eqref{eq:Projection_UPATH_all} on the tail and the head endpoints of an
      \undirectedEdgePath\ $\UndirectedPath \in \UPATH\np{\OrientedGraph}$ 
      \label{fig:tailhead_undirectedEdgePath}}
  \end{center}
\end{figure}
  
\begin{leftbar} The Coq definition of $\Projection_{\UPATH}$ follows
  \begin{alectryon}
    {\small \SnippetEope,}
  \end{alectryon}
\end{leftbar}

\myparagraph{$\bullet$ Concatenation of \undirectedEdgePaths.}
As already noted, concatenation of sequences 
denoted by the infix operator $\ltimes$ is easily defined and is associative.
When considering \undirectedEdgePaths, concatenation of
\( \UndirectedPathbis \in \UPATH\np{\OrientedGraph}\) and
\( \UndirectedPathter \in \UPATH\np{\OrientedGraph}\) gives a sequence $(\UndirectedPathbis \ltimes \UndirectedPathter) \in\Sequence(\VERTEX{\times}\VERTEX{\times}\Orientation)$
which belongs to $\UPATH\np{\OrientedGraph}$ under the additional assumption that
\( \Head{\Projection}_{\UPATH}\np{\UndirectedPathbis}=\Tail{\Projection}_{\UPATH}\np{\UndirectedPathter} \) (see
Equation~\eqref{eq:Projection_UPATH_first_and_last}), that is,
\begin{equation}
  \forall\UndirectedPathbis \in \PATH\np{\OrientedGraph}\eqsepv
  \forall \UndirectedPathter \in \PATH\np{\OrientedGraph} \eqsepv
  \Head{\Projection}\np{\UndirectedPathbis}=\Tail{\Projection}\np{\UndirectedPathter}
  \implies (\UndirectedPathbis \ltimes \UndirectedPathter) \in \UPATH\np{\OrientedGraph}
  \eqfinp
  \label{eq:concatenation_UPath}
\end{equation}
Moreover, forall \( \UndirectedPathbis \in \UPATH\np{\OrientedGraph}\) and
\( \UndirectedPathter \in \UPATH\np{\OrientedGraph}\), we have that
\begin{equation}
  \UndirectedPathbis \ltimes\UndirectedPathter
  = \Pairname \bp{ \np{ \PiSVV (\pathbis)} \ltimes \np{\PiSVV(\pathter)},
    \np{\PiSO(\orientbis) \ltimes \PiSO(\orientter)}}
  \in \UPATH\np{\OrientedGraph} 
  \eqfinp
  \label{eq:concatenation_UndirectedPath}
\end{equation}

\myparagraph{$\bullet$ Deployment in \undirectedEdgePaths.} With any binary
relation \( \relation \subset \VERTEX\times\VERTEX \), we associate the
subset~\( \DeploymentInUPaths{\relation}{\OrientedGraph} \)
of~\( \UPATH\np{\OrientedGraph} \) in~\eqref{eq:eopaths_defs} that we call the
\emph{deployment in \undirectedEdgePaths}, defined by
\begin{equation}
  \forall \relation \subset \VERTEX\times\VERTEX 
  \eqsepv 
  \DeploymentInUPaths{\relation}{\OrientedGraph}
  = \Converse{\Projection_{\UPATH}}(\relation) \subset \UPATH\np{\OrientedGraph} 
  \eqfinv
  \label{eq:DeploymentInUPaths}
\end{equation}
where the projection~\( \Projection_{\UPATH} \) has been defined
in~\eqref{eq:Projection_UPATH_all}. 

The deployment \(\DeploymentInUPaths{\relation}{\OrientedGraph} \)
is made of the \undirectedEdgePaths\ whose endpoints satisfy the 
binary relation~\( \relation \).

\begin{leftbar} It is formalized in Coq as follows
  \begin{alectryon}  {\small \SnippetDU} \end{alectryon}
  \noindent Moreover, As for \EdgePath, it is to be noted that the \undirectedEdgePaths\
  may be obtained as the image of a product of sequences of vertices and sequences of orientation.
  This is done by the \texttt{LiftO} mapping combining lift and pairing
  \begin{alectryon}{\small \Snippetpair\SnippetLiftO}\end{alectryon}
\end{leftbar}

\begin{leftbar} Finally, we prove that \texttt{LiftO} is bijective on restricted
  domain and image, with inverse \texttt{UnLiftO}, and that the bijection
  properly commutes with \texttt{Eope}
  \begin{alectryon}{\small \SnippetEopeLiftO\SnippetPeUnLiftO.}\end{alectryon}
\end{leftbar}

\begin{table}
  \centering
  \begin{tabular}{||c|c|c||}
    \hline\hline
    Name
    & Expression
    & Equation 
    \\ \hline\hline
    set of extended-oriented paths
    &
      \( \UPATH\np{\OrientedGraph} \)
    & Equation~\eqref{def:eopaths_gtzero}
    \\  \hline         
    tail and head endpoints
    & 
      \( \Projection_{\UPATH}: \UPATH \np{\OrientedGraph} \to\VERTEX\times\VERTEX \)
    &
      Equation~\eqref{eq:Projection_UPATH}
    \\
    projection mappings
    &
      \(  \Tail{\Projection}_{\UPATH}: \UPATH \np{\OrientedGraph} \to \VERTEX \)
    &
      Equation~\eqref{eq:Projection_UPATH_first_and_last}
    \\
    & \( \Head{\Projection}_{\UPATH}: \UPATH \np{\OrientedGraph} \to \VERTEX \)
    & Equation~\eqref{eq:Projection_UPATH_first_and_last}
    \\  \hline
    deployment in \undirectedEdgePaths
    & 
      \( \DeploymentInUPaths{\cdot}{\OrientedGraph} = \Converse{\Projection_{\UPATH}} \)
    &
      Equation~\eqref{eq:DeploymentInUPaths}
    \\ \hline \hline         
  \end{tabular}
  \caption{Notions for \undirectedEdgePaths\ in a graph (\S\ref{Extended-oriented_paths_in_a_(directed_simple)_graph_(permitting_loops)})}
  \label{tab:Notions_for_undirectedEdgePath_s_in_a_graph}
\end{table}

\subsection{Active \undirectedEdgePaths\ and d-separation}
\label{Definition_of_active_undirectedEdgePaths_and_of_conditional_directional_separation}
  
Let \( \npOrientedGraph \) be a graph --- as defined
in~\S\ref{Graphs_as_binary_relations}, that is, a (directed simple) graph
(permitting loops) --- and $\AgentSubsetW\subset\VERTEX$ be a subset of vertices.

In~\S\ref{Definition_of_active_undirectedEdgePaths} we formally adapt Pearl's
definition of active (and blocked) \undirectedEdgePaths\ in a graph, from which
we deduce the (conditional) d-separation binary relation
in~\S\ref{Definition_of_conditional_directional_separation}.

\subsubsection{Definition of active \undirectedEdgePaths}
\label{Definition_of_active_undirectedEdgePaths}
We take inspiration from~\citep{PEARL1986357} to define the notion of blocked
paths on a graph, not necessarily finite nor acyclic.  For this purpose, we
first define {active} \undirectedEdgePaths\ relative to the
graph~$\npOrientedGraph$ in Definition~\ref{de:ActivePaths}.  Then, we obtain
the definition of {blocked} \undirectedEdgePaths\ relative to the
graph~$\npOrientedGraph$, as defined by~\citep{PEARL1986357}, by switching to
the complementary set.

We start by introducing a binary relation, $\ActiveTr$ (active triplet), on the
set $\VERTEX{\times}\VERTEX{\times}\Orientation$, which is parameterized by the set of
edges $\EDGE$ of a graph \( \npOrientedGraph \) and by a subset
$\AgentSubsetW\subset\VERTEX$.

\begin{definition}
  The \emph{active triplet binary relation $\ActiveTr$} on the set
  $\VERTEX{\times}\VERTEX{\times}\Orientation$ is defined as follows
  \begin{subnumcases}%
    {
      \np{\tail{\vertex_1},\head{\vertex_{1}},\orient_1}
      \,\ActiveTr\,
      \np{\tail{\vertex_2},\head{\vertex_{2}},\orient_2}
      \iff
      \label{it:ActivePaths}
    }  
    \orient_1 = +1\eqsepv \orient_{2}= +1 \text{ and }
    \head{\vertex_1} = \tail{\vertex_{2}} \in \Complementary{\AgentSubsetW}\eqsepv
    \label{it:ActivePaths_case1}
    \\
    \orient_1 = -1\eqsepv \orient_{2}= -1 \text{ and }
    \head{\vertex_1} = \tail{\vertex_{2}} \in \Complementary{\AgentSubsetW}\eqsepv
    \label{it:ActivePaths_case2}
    \\
    \orient_1 = -1\eqsepv \orient_{2}= +1 \text{ and }
    \head{\vertex_1} = \tail{\vertex_{2}} \in \Complementary{\AgentSubsetW}\eqsepv
    \label{it:ActivePaths_case3}
    \\
    \orient_1 = +1\eqsepv \orient_{2}= -1 \text{ and }
    \head{\vertex_1} = \tail{\vertex_{2}} \in
    \TransitiveReflexiveClosure{\EDGE}\AgentSubsetW 
    \eqsepv
    \label{it:ActivePaths_case4}
  \end{subnumcases}
  where \( \TransitiveReflexiveClosure{\EDGE}= \TransitiveClosure{\EDGE} \cup \Delta \)
  is the reflexive and transitive closure of the relation~$\EDGE$.
\end{definition}

\begin{definition} (active \undirectedEdgePaths\ \( \ActiveUndirectedPaths{\OrientedGraph}\))
  \label{de:ActivePaths}
  We say that an \undirectedEdgePath\
  \(\UndirectedPath \in \UPATH\np{\OrientedGraph} \) in~\eqref{eq:eopaths_defs}
  relative to the graph~\(\npOrientedGraph\), is an \emph{active
    \undirectedEdgePath} (\wrt\footnote{\wrt\ stands for ``with respect to''.} the
  subset~$\AgentSubsetW$) if the successive elements (triplets) of
  \(\UndirectedPath\) satisfy the binary relation \(\ActiveTr\) defined in
  Equation~\eqref{it:ActivePaths},
  that is, $\UndirectedPath_i\, \ActiveTr\, \UndirectedPath_{i+1}$, for all
  $i \in \ic{1, \cardinal{\UndirectedPath}{-}1}$.
  Notice that any \undirectedEdgePath\ of length~1 is
  active by definition.

  We denote by \( \ActiveUndirectedPaths{\OrientedGraph} \subset \UPATH\np{\OrientedGraph}\)
  the subset of all active \undirectedEdgePaths\ (\wrt\ the subset~$\AgentSubsetW$). 
  We say that an \undirectedEdgePath\ is \emph{blocked} if it is not active
  and we denote by \( \BlockedUndirectedPaths{\OrientedGraph} = 
  \bpComplementary{\ActiveUndirectedPaths{\OrientedGraph}}\)   
  the subset of all blocked \undirectedEdgePaths\ (\wrt\ the subset~$\AgentSubsetW$).
\end{definition} 
\begin{leftbar}
  The Coq formalization of the binary relation \(\ActiveTr\) easily follows 
  \begin{alectryon}  {\small \SnippetAtr.}\end{alectryon}
  The binary relation $\ActiveTr$ contains the (oriented chain) relation --
  denoted by \texttt{ChrelO} in Coq -- and defined by avec
  $((v_1^\flat, v_1^\sharp,o_1)\, C_{\textsc{H}}^{\Orientation}\, (v_2^\flat, v_2^\sharp , o_2) \iff
  v_1^\sharp = v_2^\flat)$.  The forward set
  $\TransitiveReflexiveClosure{\EDGE}\AgentSubsetW$ is implemented by
  \texttt{Fset} in Coq.
\end{leftbar}

\begin{leftbar} Now, the Coq formalization of
  $\ActiveDeploymentInUPaths{\relation}{\OrientedGraph}$ is obtained as an
  intersection of two sets
\begin{alectryon}  {\small \SnippetDUa.} \end{alectryon}
\noindent When the relation $\relation$ is a singleton
$\relation=\na{\np{x,y}}$, the set
$\ActiveDeploymentInUPaths{\relation}{\OrientedGraph}$ boils down to the
following equivalent definition
\begin{alectryon}  {\small \SnippetDUaone.}\end{alectryon}
\end{leftbar}

\subsubsection{Definition of  conditional directional separation (d-separation)}
\label{Definition_of_conditional_directional_separation}
We introduce in Definition~\ref{de:vertices-d-separated} a new binary relation
between vertices: we say that two vertices are (conditionally) directionally
separated if and only if the two vertices are different and all the
\undirectedEdgePaths, having them as endpoints, are blocked (equivalently they
are different and there does not exist an active \undirectedEdgePaths\ having
them as endpoints). This definition mimics Pearl's
d-separation~\citep{PEARL1986357}, but with two differences: the graph is not
supposed to be acyclic, and the separation is between vertices and not between
disjoint subsets.

\begin{definition}
  \label{de:vertices-d-separated}
  Let \( \npOrientedGraph \) be a graph, 
  and $\AgentSubsetW\subset\VERTEX$ be a subset of vertices.
  %
  We denote
  \begin{equation}
    \bgent \ConditionalDirectionalSeparation \cgent \mid \AgentSubsetW
    \iff
    (\bgent \not= \cgent) \wedge 
    \Bp{\DeploymentInUPaths{\na{\np{\bgent,\cgent}}}{\graph}
      \subset \BlockedUndirectedPaths{\OrientedGraph}}
    \qquad \bp{ \forall \bgent,\cgent \in \AGENT }
    \eqfinv
    \label{eq:vertices-d-separated}
  \end{equation}
  and we say that the vertices~$\bgent$ and $\cgent$ are 
  (conditionally) \emph{directionally separated}  (\wrt\ the subset~$\AgentSubsetW$).
\end{definition}
\begin{leftbar}
  The Coq implementation of d-separation is given by
  \begin{alectryon}{\small \Snippetdsepnota\SnippetDseparated}\end{alectryon}
  \begin{alectryon}{\small \SnippetAeop\SnippetActiveOe\SnippetallL.}\end{alectryon}
\end{leftbar}
\begin{leftbar}
  The retained formulation is indeed equivalent to Definition~\ref{de:vertices-d-separated},
  as proved in Coq lemma~\texttt{Active\_eq} given below
  \begin{alectryon} {\small \SnippetActiveeq.}\end{alectryon}
\end{leftbar}

\section{Characterization  of $d$-separation by means of binary relations}
\label{Characterization__of_the_conditional_directional_separation_binary_relation}

Our main result is the characterization of the conditional directional
separation relation --- the extension
\( \bp{ \;\protect\ConditionalDirectionalSeparation \; \mid \AgentSubsetW } \), or
shortly \( \ConditionalDirectionalSeparation \), of the d-separation introduced
in Definition~\ref{de:vertices-d-separated} --- as the complementary of the
conditional active relation --- Equation~\eqref{eq:conditional_active_relation} in
Definition~\ref{de:all_the_relations} below.

For this purpose, we introduce the following binary relations on the vertices of
a graph --- as defined in~\S\ref{Graphs_as_binary_relations}, that is, a (directed
simple) graph (permitting loops).

\begin{definition}
  \label{de:all_the_relations}
  Let \( \npOrientedGraph \) be a graph, 
  and $\AgentSubsetW\subset\VERTEX$ be a subset of vertices.
  We define the \emph{conditional parental relation}~\( \ParentalPrecedence \)
  as
  \begin{subequations}
    \begin{align}
      \ParentalPrecedence 
      &= \Delta_{\Complementary{\AgentSubsetW}}\Precedence
        \mtext{ \qquad that is, }
        \bgent\ParentalPrecedence\cgent \iff
        \bgent\in\Complementary{\AgentSubsetW} \mtext{ and }
        \bgent\Precedence\cgent \qquad \bp{\forall \bgent,\cgent \in \AGENT }
        \eqfinv
        \label{eq:conditional_parental_relation}
        \intertext{the \emph{conditional ascendent relation} $\ConditionalAscendent$ as }
        \ConditionalAscendent
      &=
        \ConditionalDown = \Precedence \TransitiveReflexiveClosureParentalPrecedence
        \label{eq:conditional_ascendent_relation}
        \eqfinv
        \intertext{which relates a descendent with
        an ascendent by means of elements in~$\Complementary{\AgentSubsetW}$. 
        We define their converses~\( \ConverseParentalPrecedence \) and
        \( \ConverseConditionalAscendent \) as }
        \ConverseParentalPrecedence
      &= \npConverse{\ParentalPrecedence}
        = \Converse{\Precedence} \Delta_{\Complementary{\AgentSubsetW}}
        \eqfinv
        \label{eq:converse_conditional_parental_relation}
      \\
      \ConverseConditionalAscendent
      &= \Converse{\bp{\ConditionalAscendent}}
        = \ConditionalUp
        = \TransitiveReflexiveClosureConverseParentalPrecedence \Converse{\Precedence}
        \eqfinp
        \label{eq:converse_conditional_ascendent_relation}
        \intertext{With these elementary binary relations,
        we define the \emph{conditional common cause relation}~$\ConditionalCommonCause$ 
        as the symmetric relation}
        \ConditionalCommonCause 
      &=
        \ConverseConditionalAscendent \Delta_{\Complementary\AgentSubsetW}
        \ConditionalAscendent
        = \TransitiveClosureConverseParentalPrecedence \TransitiveClosureParentalPrecedence
        \label{eq:common_cause}
        \eqfinv
        \intertext{the \emph{conditional cousinhood relation}~$\Cousinhood$
        as the partial equivalence relation} 
        \Cousinhood
      &=
        \bpTransitiveClosure{\Delta_{\AgentSubsetW} \ConditionalCommonCause
        \Delta_{\AgentSubsetW} }
        \cup
        \Delta_{\AgentSubsetW}
        \eqfinv
        \label{eq:Cousinhood}
        \intertext{
        and the \emph{conditional active relation}~$\ConditionalActive$ 
        as the symmetric relation}
        \ConditionalActive
      &= \Delta \cup 
        \ConditionalAscendent \cup \ConverseConditionalAscendent \cup \ConditionalCommonCause
        \cup
        \bp{\ConditionalAscendent \cup \ConditionalCommonCause}
        \Cousinhood
        \bp{\ConverseConditionalAscendent \cup \ConverseConditionalCommonCause}
        \eqfinp
        \label{eq:conditional_active_relation}
    \end{align}
  \end{subequations}
\end{definition}

\begin{leftbar} The Coq implementation is straightforward using the binary
  relation library \texttt{rel.v} that we have developed.
  \begin{alectryon}  {\small \SnippetAw.}\end{alectryon}
\end{leftbar}

We now state the main result of this paper.

\begin{theorem} (Coq \PY{k+kn}{Theorem}~~\PY{n+nf}{Th\ref{th:ConditionalDirectionalSeparation_IFF_relation}})
  \label{th:ConditionalDirectionalSeparation_IFF_relation}
  Let \( \npOrientedGraph \) be a graph, 
  and $\AgentSubsetW\subset\VERTEX$ be a subset of vertices.
  The conditional directional separation
  relation~$\ConditionalDirectionalSeparation$
  (Definition~\ref{de:vertices-d-separated}) is the
  complementary~$\bpComplementary{ \ConditionalActive }$ of the conditional
  active relation~$\ConditionalActive$
  (Equation~\eqref{eq:conditional_active_relation} in
  Definition~\ref{de:all_the_relations}):
  \begin{equation}
    \bp{ \;\protect\ConditionalDirectionalSeparation \; \mid \AgentSubsetW }
    =
    \bpComplementary{ \ConditionalActive }
    \mtext{ or, equivalently, }
           \bgent\ConditionalDirectionalSeparation\cgent
      \; \mid \AgentSubsetW  \iff
      \neg \np{ \bgent\ConditionalActive\cgent }
      \qquad \bp{ \forall \bgent,\cgent \in \AGENT }
    \eqfinp 
  \end{equation}
  In other words, we have that 
  \begin{equation}
    \bset{ \np{\bgent,\cgent} \in \AGENT\times\AGENT }{
      \np{\bgent \not=\cgent} \wedge 
      \DeploymentInUPaths{\na{\np{\bgent,\cgent}}}{\graph} \subset
      \BlockedUndirectedPaths{\graph} }
    = \npComplementary{\ConditionalActive}
    \eqfinp
  \end{equation}
\end{theorem}


\begin{figure}[hbtp]
  \begin{center}
    \mbox{\includegraphics[width=1.0\textwidth]{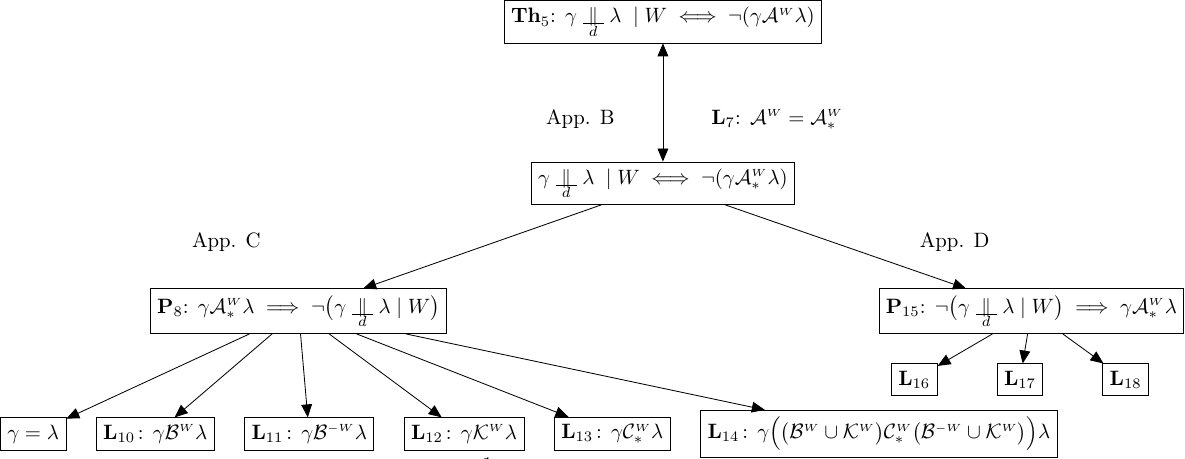}}
    \caption{
      Sketch of proof of
      Theorem~\ref{th:ConditionalDirectionalSeparation_IFF_relation},
      mentioning the corresponding Appendices and lemmata
      \label{fig:sketch_of_proof}}
  \end{center}
\end{figure}

\begin{proof}
  The proof of Theorem~\ref{th:ConditionalDirectionalSeparation_IFF_relation}
  is broken in three steps and summarized in Figure~\ref{fig:sketch_of_proof}.
  \medskip
  
  First, we will prove in postponed Lemma~\ref{le:cond-act-equals-cond-act-star}
  (in~\S\ref{appendix_Definitions_binary_relations}) that
  $\ConditionalActive =\ConditionalActiveStar$, where the binary relation~$\ConditionalActiveStar$ is defined by
  Equation~\eqref{eq:ConditionalActivePlus}.  \medskip
  
  Second, we will prove in postponed
  Proposition~\ref{pr:ConditionalDirectionalSeparation_seilpmi_relation}
  (in~Appendix~\ref{Proof_of_Theorem_implies}) that,
  for any vertices \( \bgent,\cgent \in \AGENT \), we have the implication
  \( \bgent \ConditionalActiveStar \cgent \implies \neg \bp{\bgent
    \ConditionalDirectionalSeparation \cgent \mid \AgentSubsetW} \) or,
  equivalently (see~\eqref{eq:vertices-d-separated} in
  Definition~\ref{de:vertices-d-separated}), the implication
  \[ \bgent {\ConditionalActiveStar} \cgent \implies
    \np{\bgent=\cgent} \vee \Bp{ \DeploymentInUPaths{\na{\np{\bgent,\cgent}}}{\graph} 
    \cap \ActiveUndirectedPaths{\graph} 
    \neq\emptyset }
    \eqfinp
  \]
  We simply give a sketch of proof here as details are to be found
  in Proposition~\ref{pr:ConditionalDirectionalSeparation_seilpmi_relation}
  accompanied by postponed lemmata given in
  Appendix~\ref{Proof_of_Theorem_implies}.
  The binary relation~\(  \ConditionalActiveStar \) defined
  in~\eqref{eq:conditional_active_relationStar}
  is given by the union of five relations. Then, 
  the proof of Proposition~\ref{pr:ConditionalDirectionalSeparation_seilpmi_relation}
  examines the five cases and exhibits an active path
  (one in~\( \ActiveUndirectedPaths{\graph} \), see Definition~\ref{de:ActivePaths})
  that joins the vertices~\( \bgent \) and \( \cgent \) in the five cases
  when \( \bgent \ConditionalActiveStar \cgent \).
  \medskip
  
  Third, we will prove in Proposition~\ref{th:main_isimplied}
  (in~Appendix~\ref{Proof_of_Theorem_isimplied}) that, for any vertices
  \( \bgent,\cgent \in \AGENT \), we have the implication
  $\neg \bp{\bgent \ConditionalDirectionalSeparation \cgent \mid \AgentSubsetW}
  \implies \bgent \ConditionalActiveStar \cgent $ or, equivalently
  (see~\eqref{eq:vertices-d-separated} in
  Definition~\ref{de:vertices-d-separated}), that
  \[ \np{\bgent=\cgent} \vee
   \Bp{ \DeploymentInUPaths{\na{\np{\bgent,\cgent}}}{\graph} 
    \cap \ActiveUndirectedPaths{\graph} 
    \neq\emptyset }
    \implies
    \bgent {\ConditionalActiveStar} \cgent
    \eqfinp
  \]
  We give again a sketch of proof. The case $\bgent = \cgent$ is immediate.
  Thus, we assume that there exists an \undirectedEdgePath\
  \( \UndirectedPath \in \UPATH\np{\graph} \) joining the vertices~$\bgent$ and
  $\cgent$ and such that \( \UndirectedPath \) is active.  If the path length
  of~\( \UndirectedPath \) is equal to one, the proof easily follows.  If the
  path length of~\( \UndirectedPath \) is $\ge 2$ a proof by induction on the path
  length is obtained using Lemma~\ref{lem:induction-lemma}.  During the
  induction step, four cases are to be discussed, following the fact that the
  active triplet binary relation $\ActiveTr$ is governed by four cases.  The
  scheme of the induction step is summarized in Figure~\ref{fig:isimplied}.
  
  This ends the proof.
\end{proof}
\begin{leftbar}
  We end this section by giving the Coq statement of
  Theorem~\ref{th:ConditionalDirectionalSeparation_IFF_relation} together with
  the corresponding Coq proof dependency graph in Figure~\ref{Fig:th5dpd}.
  \begin{alectryon}
    {\small \Snippettheoremfive} 
  \end{alectryon}
\end{leftbar}

\begin{figure}[h]
  \begin{center}
    \includegraphics[width=0.9\textwidth]{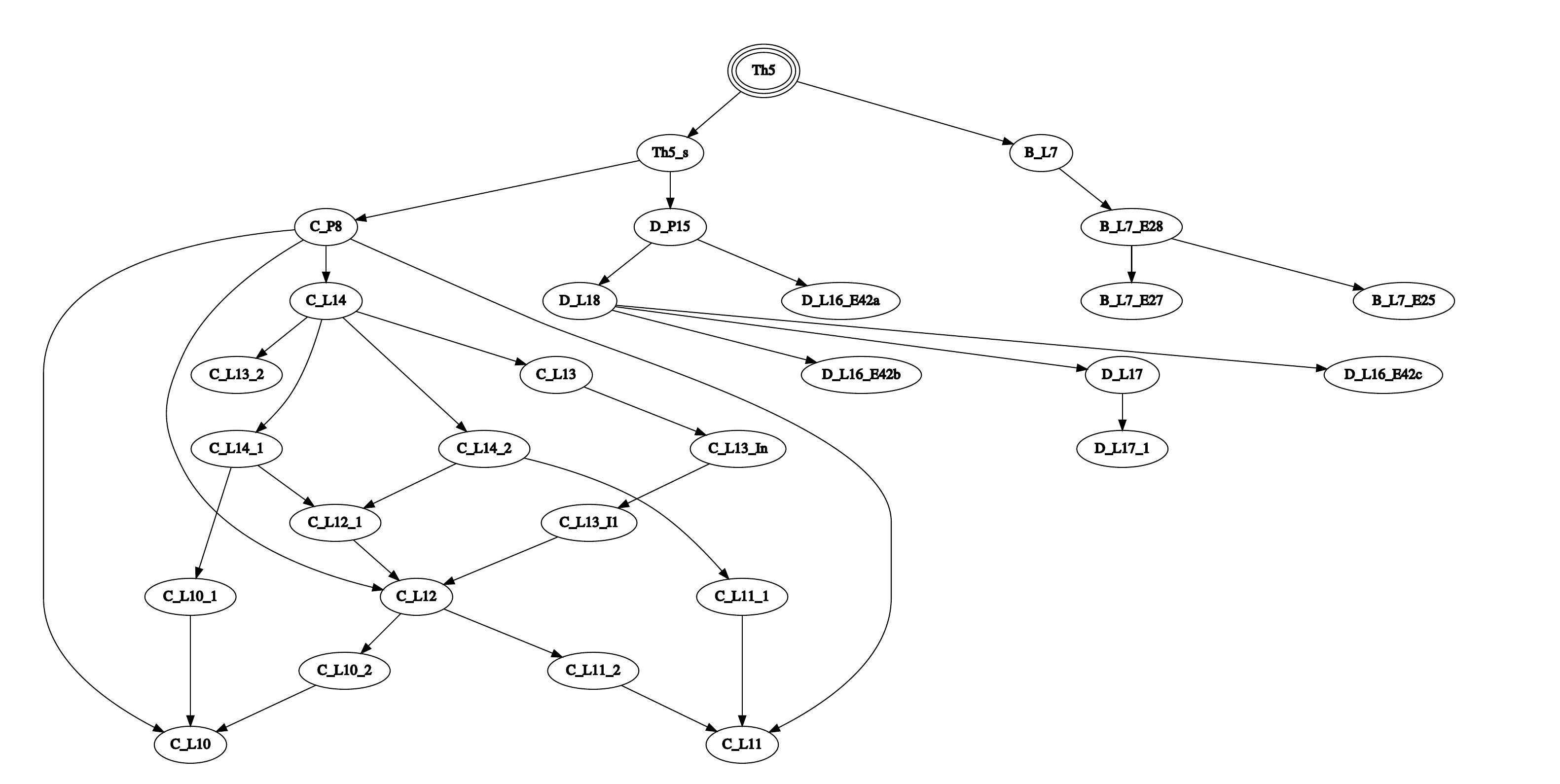}
    \caption{\label{Fig:th5dpd}
      Coq-produced dependency graph for Coq
      Theorem~\ref{th:ConditionalDirectionalSeparation_IFF_relation}
      (\texttt{B\_L7} is Lemma~\ref{le:cond-act-equals-cond-act-star},
      \texttt{D\_P15} is Proposition~\ref{th:main_isimplied}
      and \texttt{C\_P8} is Proposition~\ref{pr:ConditionalDirectionalSeparation_seilpmi_relation})
    }
  \end{center}
\end{figure}

\section{Conclusion}

Together with its two companion
papers~\citep{De-Lara-Chancelier-Heymann-2021,Heymann-De-Lara-Chancelier-2021},
this paper is a contribution to providing another perspective on conditional
independence and do-calculus.  In this paper, we have considered directed graphs
(DGs) not necessarily acyclic, and we have shown how the d-separation can be
extended beyond acyclic graphs and can be expressed and characterized as a
binary relation between vertices.
The results in this paper are instrumental in proving those
in~\citep{De-Lara-Chancelier-Heymann-2021} on topological conditional separation
(t-separation), hence in the use of t-separation to establish conditional
independence in~\citep{Heymann-De-Lara-Chancelier-2021}.

Moreover, there are other perspectives.  First, such developpement of a theory
based on binary relations is interesting in itself as it makes it amenable to
computer aided proof.  Second, there are other notions of separation (between
subsets) in graph theory that can also be expressed by means of binary relations
(between vertices).  We illustrate this with an example.  Let
\( \npOrientedGraph \) be a graph, $\AgentSubsetW\subset\VERTEX$ be a subset of
vertices, and $\bgent$, $\cgent\in \AGENT$ be vertices.  Then, the three following
statements are equivalent: any path from~$\bgent$ to~$\cgent$ passes
through~$\AgentSubsetW$; there does not exist a path from~$\bgent$ to~$\cgent$
which passes through~$\Complementary{\AgentSubsetW}$;
\( \neg \np{ \bgent \npTransitiveClosure{ \EDGE \Delta_{\Complementary{\AgentSubsetW}}
    \EDGE} \cgent } \).

\appendix

\section{Comments on Coq in the appendices proofs}
\label{Comments_on_Coq_in_the_appendices_proofs}

 The Coq proof of Theorem~\ref{th:ConditionalDirectionalSeparation_IFF_relation} closely follows the
 mathematical proof detailed in the three following Appendices
 (see also Figures~\ref{fig:sketch_of_proof} and~\ref{Fig:th5dpd}) and, for each (mathematical) lemma,
 we will give its Coq name.
 
 The Coq proof is obtained with the help of a novel library, developed by
 J.P. Chancelier and which provides tools
 for reasoning on binary relations and on active \undirectedEdgePaths.
 This library~\cite{Chancelier-Coq:2024} is publicly available on GitHub at URL
 \begin{center}
   \texttt{https://github.com/jpc-cermics/relations.git}
 \end{center}
  
 \subsection{Binary relations}

 All the mathematical objects described in \S\ref{defs:binary-rel}, together with
 associated lemmata for manipulating them, are implemented in a Coq library.
 As relations are coded as sets, the library we have developed (mainly contained in file \texttt{rel.v})
 is based on the \emph{mathcomp} implementation of sets
 \texttt{classical\_sets.v}. It also contains a SSReflect reimplementation of
 transitive (reflexive) closures of relation (\texttt{ssrel.v}) that are found in the Coq standard
 library.  Moreover --- as proofs on relations proceed by rewriting rules using intensively associative and commutative
 properties of union and intersection of relations --- we use the
 \texttt{AAC\_tactics} to ease equality proofs between long expressions with
 relations containing unions, intersections, compositions and diagonal
 relations. For this to be possible, some specific properties of relations with
 respect to \texttt{AAC} are to be listed and proved in Coq. This is done in
 \texttt{aacset.v}.
 
 \subsection{\UndirectedEdgePaths\ in a graph}

 We have also developed a library for manipulating paths in a graph (\texttt{seq1.v}), which contains
 the objects described in \S\ref{UndirectedEdgePaths_in_a_graph}.
 Even if the the concepts of \undirectedEdgePaths\ and active \undirectedEdgePaths\
 are quite specific to d-separation, the library we have developed contains many tools for
 manipulating node or edge paths which are of more general interest. 
  
 \subsection{Coq proofs of Theorem~\ref{th:ConditionalDirectionalSeparation_IFF_relation}}

 Two files are specifically devoted to the proof of Theorem~\ref{th:ConditionalDirectionalSeparation_IFF_relation}.
 The file \texttt{paper\_relations.v} contains the definition of the binary relations
 contained in Definition~\ref{de:all_the_relations} and~\ref{de:all_the_relations_new} and in 
 Lemma~\ref{le:lemma8}, together with some lemmata on their respective properties.
 The file \texttt{paper\_csbr.v} contains the proofs of the propositions and lemmata which are mathematically proved in
 the following appendices.

 \section{The star conditional active relation~$\ConditionalActiveStar$}
\label{appendix_The_star_conditional_active_relation}

In~\S\ref{appendix_Definitions_binary_relations}, we introduce new binary
relations to define the star conditional active
relation~$\ConditionalActiveStar$.
In~\S\ref{Proof_that_ConditionalActive=ConditionalActiveStar}, we prove that
$\ConditionalActive =\ConditionalActiveStar$, where
the conditional active relation~$\ConditionalActive$ has been introduced
in Equation~\eqref{eq:conditional_active_relation} of
Definition~\ref{de:all_the_relations}. 

\subsection{Definition and properties of the star conditional active relation~$\ConditionalActiveStar$}
\label{appendix_Definitions_binary_relations}

We refer the reader to Definition~\ref{de:all_the_relations} for the definitions
of basic binary relations.
We add two new ones.

\begin{definition}
  \label{de:all_the_relations_new}
  Let \( \npOrientedGraph \) be a graph, 
  and $\AgentSubsetW\subset\VERTEX$ be a subset of vertices.
  We introduce the notation
  \begin{equation}
    \TransitiveReflexiveClosureOfSet{\AgentSubsetW}=
    \TransitiveReflexiveClosure{\EDGE}\AgentSubsetW
    \eqfinv
    \label{eq:TransitiveReflexiveClosure_AgentSubsetW}
  \end{equation}
  and we define the \emph{star conditional cousinhood relation}
  \begin{align}
    \CousinhoodStar
    &=
      \bpTransitiveClosure{ \Delta_{\TransitiveReflexiveClosureOfSet{\AgentSubsetW}}
      \ConditionalCommonCause \Delta_{\TransitiveReflexiveClosureOfSet{\AgentSubsetW}} }
      \cup 
      \Delta_{\TransitiveReflexiveClosureOfSet{\AgentSubsetW}}
      \eqfinv
      \label{eq:CousinhoodStar}
      \intertext{and the \emph{star conditional active relation} } 
    \ConditionalActiveStar
    &=
      \Delta \cup 
      \ConditionalAscendent \cup \ConverseConditionalAscendent \cup \ConditionalCommonCause
      \cup
      \bp{\ConditionalAscendent \cup \ConditionalCommonCause}
      \CousinhoodStar
      \bp{\ConverseConditionalAscendent \cup \ConverseConditionalCommonCause}
      \eqfinp
      \label{eq:conditional_active_relationStar}
  \end{align}
\end{definition}
Notice that the {star conditional cousinhood relation}~\( \CousinhoodStar \)
in~\eqref{eq:CousinhoodStar} is the relation~\( \Cousinhood \)
in~\eqref{eq:Cousinhood} with \( \AgentSubsetW \) replaced
by~\( \TransitiveReflexiveClosureOfSet{\AgentSubsetW} \), and the {star
  conditional active relation}~$\ConditionalActiveStar$
in~\eqref{eq:conditional_active_relationStar} is the
relation~\( \ConditionalActive \) in~\eqref{eq:conditional_active_relation} with
\( \Cousinhood \) replaced by~\( \CousinhoodStar \).

\subsection{Proof that $\ConditionalActive =\ConditionalActiveStar$}
\label{Proof_that_ConditionalActive=ConditionalActiveStar}

Recall that the conditional active relation~$\ConditionalActive$ has been introduced
in Equation~\eqref{eq:conditional_active_relation} of
Definition~\ref{de:all_the_relations}.

\begin{lemma}
  (Coq \PY{k+kn}{Lemma}~\PY{n+nf}{B\_L\ref{le:cond-act-equals-cond-act-star}})
  \label{le:cond-act-equals-cond-act-star}
  We have that
  \begin{equation}
    \ConditionalActive =\ConditionalActiveStar
    \eqfinp
  \end{equation}
\end{lemma}

\begin{proof}
  The proofs is in three steps.
  \medskip
  
  \noindent $\bullet$ (\PY{k+kn}{Lemma}~~\PY{n+nf}{B\_L\ref{le:cond-act-equals-cond-act-star}\PYZus{}E\ref{eq:R-kleene-G-equals}}\PY{o})
  We prove that
  \begin{align}
    \forall \relation \subset \AGENT\times\AGENT \eqsepv
    \forall \Bgent \subset \AGENT \eqsepv
    \TransitiveReflexiveClosure{\relation}{\Bgent}
    & =
      \npTransitiveReflexiveClosure{\Delta_{\Complementary{\Bgent}}\relation} \Bgent 
      \eqfinp
      \label{eq:R-kleene-G-equals}
  \end{align}
  For this purpose, we prove the following induction
  assumption~\( {\cal H}_n \): for any $n\ge 1$, we have that
  $\np{\cup_{k=0}^n \relation^k}\Bgent = \bp{ \cup_{k=0}^n
    \np{\Delta_{\Complementary{\Bgent}} \relation}^k }\Bgent$, where we recall the
  convention $\relation^0=\Delta$.

  As a preliminary result, for any \( \Bgent, \Cgent\subset \AGENT \), from the
  sequence of equalities
  \( \np{\Delta \Bgent} \cup \np{\relation \Cgent} = \Bgent \cup \bp{ \np{\relation \Cgent}
    \backslash \Bgent}
  = \Bgent \cup  \np{\Delta_{\Complementary{\Bgent}} \relation \Cgent}
  =\np{\Delta \Bgent} \cup  \np{\Delta_{\Complementary{\Bgent}} \relation
    \Cgent} \), we deduce that
  \begin{equation}
    \np{\Delta \Bgent} \cup \np{\relation \Cgent}=
    \np{\Delta \Bgent} \cup  \np{\Delta_{\Complementary{\Bgent}} \relation \Cgent}
    \eqfinp    
    \label{eq:preliminary_result}
  \end{equation}
  Thus, with \( \Cgent=\Bgent \), we obtain that $\np{\cup_{k=0}^1 \relation^k}\Bgent
  =  \bp{ \cup_{k=0}^1 \np{\Delta_{\Complementary{\Bgent}} \relation}^k
  }\Bgent$, that is, assumption~\( {\cal H}_1 \) holds true. 
  
  Now, we suppose that, for a given $n\ge 1$, the induction assumption~\( {\cal H}_n
  \) holds true. Then, we have that 
  \begin{align*}
    \np{\cup_{k=0}^{n+1} \relation^k}\Bgent
    &=
      \np{\Delta \Bgent} \cup
      \Bp{\relation \bp{\np{\cup_{k=0}^n \relation^k}}\Bgent}
      \tag{using the convention $\relation^0=\Delta$}
    \\
    &=
      \np{\Delta \Bgent} \cup
      \Bp{\Delta_{\Complementary{\Bgent}} \relation \bp{\np{\cup_{k=0}^n \relation^k}}\Bgent}
      \intertext{using the preliminary result~\eqref{eq:preliminary_result}
      but with the binary relation \( \relation \bp{\np{\cup_{k=0}^n
      \relation^k}} \)}
    &=
      \np{\Delta \Bgent} \cup
      \bgp{\Delta_{\Complementary{\Bgent}} \relation
      \bp{ \cup_{k=0}^n \np{\Delta_{\Complementary{\Bgent}} \relation}^k }\Bgent}
      \tag{using the induction assumption~\( {\cal H}_n \)}
    \\
    &=
      \np{\Delta \Bgent} \cup \Bp{
      \bp{ \cup_{k=1}^{n+1} \np{\Delta_{\Complementary{\Bgent}} \relation}^k }\Bgent}
    \\
    &=
      \bp{ \cup_{k=0}^{n+1} \np{\Delta_{\Complementary{\Bgent}} \relation}^k }\Bgent
      \eqfinp
      \tag{as \( \np{\Delta_{\Complementary{\Bgent}} \relation}^0 =\Delta\)}
  \end{align*}
  Thus, we have proven the induction assumption~\( {\cal H}_{n+1} \).
  \medskip
  
  Now, let us suppose that
  $\bgent \in \TransitiveReflexiveClosure{\relation}\Bgent$. Then, there exists a
  positive integer $n\ge 1$ such that
  $\bgent \in \np{\cup_{k=0}^n \relation^k}\Bgent$; using the just proven
  property~\( {\cal H}_n \), we get that
  $\bgent \in \bp{ \cup_{k=0}^{n} \np{\Delta_{\Complementary{\Bgent}} \relation}^k
  }\Bgent$ and, therefore,
  $\bgent \in
  \npTransitiveReflexiveClosure{\Delta_{\Complementary{\Bgent}}\relation}\Bgent$.
  Thus, we have shown that
  $ \TransitiveReflexiveClosure{\relation}\Bgent \subset
  \npTransitiveReflexiveClosure{\Delta_{\Complementary{\Bgent}}\relation}\Bgent$.
  The converse inclusion is easier to prove as
  \( \Delta_{\Complementary{\Bgent}}\relation \subset\relation \).  Finally, we have shown
  the equality
  $ \TransitiveReflexiveClosure{\relation}\Bgent =
  \npTransitiveReflexiveClosure{\Delta_{\Complementary{\Bgent}}\relation}\Bgent$,
  which is~\eqref{eq:R-kleene-G-equals}.  \medskip
  
  \noindent $\bullet$ (\PY{k+kn}{Lemma}~~\PY{n+nf}{B\_L\ref{le:cond-act-equals-cond-act-star}\PYZus{}E\ref{eq:delta-R-G-subset-R-delta-G-Rinv}}\PY{o})
  The following inclusion is easy to prove:
  \begin{align}
    \forall \relation \subset \AGENT\times\AGENT \eqsepv
    \forall \Bgent \subset \AGENT \eqsepv
    \Delta_{\relation\Bgent} 
    &\subset 
      \relation\Delta_{\Bgent}\Converse{\relation}
      \eqfinp
      \label{eq:delta-R-G-subset-R-delta-G-Rinv}
  \end{align}
  
  \noindent $\bullet$  (\PY{k+kn}{Lemma}~~\PY{n+nf}{B\_L\ref{le:cond-act-equals-cond-act-star}\PYZus{}E\ref{eq:star-no-star}}\PY{o})
  We prove that
  \begin{align}
    \TransitiveReflexiveClosureParentalPrecedence  \Delta_{\AgentSubsetW}
    \TransitiveReflexiveClosureConverseParentalPrecedence
    & =
      \TransitiveReflexiveClosureParentalPrecedence \Delta_{\np{\TransitiveReflexiveClosure{\Precedence}\AgentSubsetW}}
      \TransitiveReflexiveClosureConverseParentalPrecedence\eqfinp
      \label{eq:star-no-star}
  \end{align}
  Using Equation~\eqref{eq:R-kleene-G-equals} with $\relation{=}\Precedence$ and
  $\Bgent{=}\AgentSubsetW$ gives
  $\TransitiveReflexiveClosure{\Precedence}\AgentSubsetW
  =\npTransitiveReflexiveClosure{\Delta_{\Complementary{\AgentSubsetW}}\Precedence}
  \AgentSubsetW = \TransitiveReflexiveClosureParentalPrecedence \AgentSubsetW$.
  Combined with the Inclusion~\eqref{eq:delta-R-G-subset-R-delta-G-Rinv}, we get
  \( \Delta_{\TransitiveReflexiveClosure{\Precedence}\AgentSubsetW} 
  = \Delta_{\TransitiveReflexiveClosureParentalPrecedence \AgentSubsetW} 
  \subset \TransitiveReflexiveClosureParentalPrecedence \Delta_{\AgentSubsetW}
  \TransitiveReflexiveClosureConverseParentalPrecedence
  \). 
  Thus, we obtain that 
  \begin{align*}
    \TransitiveReflexiveClosureParentalPrecedence  \Delta_{\np{\TransitiveReflexiveClosure{\Precedence}\AgentSubsetW}}
    \TransitiveReflexiveClosureConverseParentalPrecedence
    & \subset\TransitiveReflexiveClosureParentalPrecedence
      \np{\TransitiveReflexiveClosureParentalPrecedence \Delta_{\AgentSubsetW} \TransitiveReflexiveClosureConverseParentalPrecedence} 
      \TransitiveReflexiveClosureConverseParentalPrecedence
      =
      {\TransitiveReflexiveClosureParentalPrecedence \Delta_{\AgentSubsetW} \TransitiveReflexiveClosureConverseParentalPrecedence}\eqfinp
  \end{align*}
  Thus, we have obtained the inclusion
  \(
  \TransitiveReflexiveClosureParentalPrecedence  \Delta_{\AgentSubsetW}
  \TransitiveReflexiveClosureConverseParentalPrecedence
  \supset
  \TransitiveReflexiveClosureParentalPrecedence \Delta_{\np{\TransitiveReflexiveClosure{\Precedence}\AgentSubsetW}}
  \TransitiveReflexiveClosureConverseParentalPrecedence
  \).
  The reverse inclusion  follows from the fact that
  \(\AgentSubsetW \subset  \TransitiveReflexiveClosureOfSet{\AgentSubsetW}=
  \TransitiveReflexiveClosure{\EDGE}\AgentSubsetW \)
  by~\eqref{eq:TransitiveReflexiveClosure_AgentSubsetW},   
  which gives
  \(
  \TransitiveReflexiveClosureParentalPrecedence  \Delta_{\AgentSubsetW}\TransitiveReflexiveClosureConverseParentalPrecedence
  \subset 
  \TransitiveReflexiveClosureParentalPrecedence \Delta_{\np{\TransitiveReflexiveClosure{\Precedence}\AgentSubsetW}}\TransitiveReflexiveClosureConverseParentalPrecedence
  \).

  \medskip
  \noindent $\bullet$  (\PY{k+kn}{Lemma}~\PY{n+nf}{B\_L\ref{le:cond-act-equals-cond-act-star}})
  Finally, we prove that
  $\ConditionalActive =\ConditionalActiveStar$. For that purpose, it suffices
  to show that replacing the subexpressions $\Delta_{\AgentSubsetW}$ by
  $\Delta_{\np{\TransitiveReflexiveClosure{\Precedence}\AgentSubsetW}}$ in the
  expression~\eqref{eq:conditional_active_relation} of $\ConditionalActive$ does not change the relation. Using the
  definition of $\ConditionalActive$ in
  Equation~\eqref{eq:conditional_active_relation}, we obtain that
  $\Delta_{\AgentSubsetW}$ appears only in subexpressions of the form
  $\ConditionalAscendent \Delta_{\AgentSubsetW} \ConverseConditionalAscendent$
  or
  $\ConditionalCommonCause \Delta_{\AgentSubsetW} \ConverseConditionalAscendent$
  or $\ConditionalAscendent \Delta_{\AgentSubsetW} \ConditionalCommonCause$ or
  $\ConditionalCommonCause \Delta_{\AgentSubsetW} \ConditionalCommonCause$. Now,
  using the fact that the two relations $\ConditionalAscendent$ and
  $\ConditionalCommonCause$ always end with
  $\TransitiveReflexiveClosureParentalPrecedence$ and the two relation
  $\ConverseConditionalAscendent$ and $\ConditionalCommonCause$ always start with
  $\TransitiveReflexiveClosureConverseParentalPrecedence$ we obtain that
  $\Delta_{\AgentSubsetW}$ appears only in subexpressions of the form
  $\TransitiveReflexiveClosureParentalPrecedence \Delta_{\AgentSubsetW}
  \TransitiveReflexiveClosureConverseParentalPrecedence$.  We conclude, using
  Equation~\eqref{eq:star-no-star}, that $\Delta_{\AgentSubsetW}$ can be
  replaced by
  $\Delta_{\np{\TransitiveReflexiveClosure{\Precedence}\AgentSubsetW}}$ in
  $\ConditionalActive$ without changing the relation.
  \medskip

  This ends the proof.
\end{proof}

\section{Proof of \protect\( \bgent \ConditionalActiveStar \cgent \implies \neg
 \np{\bgent \protect\ConditionalDirectionalSeparation \cgent \mid \AgentSubsetW }
  \) }
\label{Proof_of_Theorem_implies}

The following
Proposition~\ref{pr:ConditionalDirectionalSeparation_seilpmi_relation}, that we
are going to prove, is half of the proof of
Theorem~\ref{th:ConditionalDirectionalSeparation_IFF_relation}.  It relies on
six postponed lemmata given in this Appendix~\ref{Proof_of_Theorem_implies}.

\subsection{Proposition~\ref{pr:ConditionalDirectionalSeparation_seilpmi_relation}}

\begin{proposition}  (Coq \PY{k+kn}{Proposition}~\PY{n+nf}{C\_P\ref{pr:ConditionalDirectionalSeparation_seilpmi_relation}})
  \label{pr:ConditionalDirectionalSeparation_seilpmi_relation}  
  Let \( \npOrientedGraph \) be a graph, 
  and $\AgentSubsetW\subset\VERTEX$ be a subset of vertices.
  Let \( \bgent, \cgent \in\AGENT \) be vertices.

  We have the implication
  \begin{equation}
    \bgent \ConditionalActiveStar \cgent \implies 
    \neg \bp{\bgent \ConditionalDirectionalSeparation \cgent \mid \AgentSubsetW}
    \label{eq:rel-to-d-sep}
  \end{equation}
  where $\ConditionalActiveStar$ is the star conditional active
  relation~\eqref{eq:conditional_active_relationStar}
  and $\ConditionalDirectionalSeparation$ is the
  conditional directional separation relation~\eqref{eq:vertices-d-separated}.
\end{proposition}

\begin{proof}
  Let \( \bgent, \cgent \in\AGENT \) be vertices, and assume that
  $\bgent \ConditionalActiveStar \cgent$ where $\ConditionalActiveStar$ is the conditional active
  relation~\eqref{eq:conditional_active_relationStar}.
  We start by proving that either $\bgent = \cgent$ or there exists an active
  \undirectedEdgePath\, $\UndirectedPath$, joining the vertices~$\bgent$ and $\cgent$.
  
  Now, by~\eqref{eq:conditional_active_relationStar}, giving
  $\ConditionalActiveStar$, 
  we have that 
  \[
    \bgent \bgp{\Delta \cup \ConditionalAscendent \cup
      \ConverseConditionalAscendent \cup \ConditionalCommonCause \cup
      \ConditionalActiveThreeStar} \cgent
    \eqfinp
  \]
  
  We consider the five cases, one by one. The first case is
  \( \bgent \Delta \cgent \), that is, \(\bgent = \cgent\) and the conclusion is immediate.
  Now, for each of the remaining case, we are going to show that
  there exists \(  \UndirectedPath \in \DeploymentInUPaths{\na{\np{\bgent,\cgent}}}{\graph} 
  \cap \ActiveUndirectedPaths{\graph} \) that joins $\bgent$ and $\cgent$, recalling that the
  deployment~\( \DeploymentInUPaths{\relation}{\OrientedGraph} \) in
  \undirectedEdgePaths\ of a binary relation~$\relation$ has been defined
  in~\eqref{eq:DeploymentInUPaths}.
  
  The second case is \( \bgent \ConditionalAscendent \cgent \).
  We conclude that there exists \( \UndirectedPath\in \DeploymentInUPaths{\np{\bgent,\cgent}}{\graph} 
  \cap \ActiveUndirectedPaths{\graph} \)
  thanks
  to~\eqref{eq:BlockedUndirectedPaths_elementary_Delta_ComplementaryAgentSubsetW_EDGE}
  in Lemma~\ref{lem:Bw_active}.
  
  The third case is \( \bgent \ConverseConditionalAscendent \cgent \).
  We conclude that there exists \( \UndirectedPath\in \DeploymentInUPaths{\np{\bgent,\cgent}}{\graph} 
  \cap \ActiveUndirectedPaths{\graph} \)
  thanks
  to~\eqref{eq:BlockedUndirectedPaths_elementary_Converse_EDGE_Delta_ComplementaryAgentSubsetW}
  in Lemma~\ref{lem:Bmw_active}.
  
  The fourth case is \( \bgent \ConditionalCommonCause \cgent \).
  We conclude that there exists \( \UndirectedPath\in \DeploymentInUPaths{\np{\bgent,\cgent}}{\graph} 
  \cap \ActiveUndirectedPaths{\graph} \)
  thanks
  to~\eqref{eq:BlockedUndirectedPaths_elementary_TransitiveClosure_Converse_EDGE_Delta_ComplementaryAgentSubsetW}
  in Lemma~\ref{lem:Kw_active}.

  The fifth case is $\bgent \ConditionalActiveThreeStar \cgent$.
  We conclude that there exists \( \UndirectedPath\in \DeploymentInUPaths{\np{\bgent,\cgent}}{\graph} 
  \cap \ActiveUndirectedPaths{\graph} \)
  thanks to~\eqref{eq:BlockedUndirectedPaths_elementary_TransitiveClosure_TransitiveClosure}
  in Lemma~\ref{lem:BkCwsBmKBk}.
  \medskip

  Finally, we successively have that
  \begin{align*}
    \bgent \ConditionalActiveStar \cgent
    &\implies
      (\bgent = \cgent) \vee \exists \UndirectedPath \in
      \DeploymentInUPaths{\na{\np{\bgent,\cgent}}}{\graph} \cap \ActiveUndirectedPaths{\graph}
      \nonumber
      \tag{as just proved above}
    \\
    &\implies
      \neg \Bp{(\bgent = \cgent) \vee \bcDeploymentInUPaths{\na{\np{\bgent,\cgent}}}{\graph} \subset
      \BlockedUndirectedPaths{\graph}}
      \intertext{by definition of the subset \( \BlockedUndirectedPaths{\graph}=
      \bpComplementary{\ActiveUndirectedPaths{\OrientedGraph}}\)
      of all blocked \undirectedEdgePaths\ (see Definition~\ref{de:ActivePaths})}
    &\implies
      \neg \bp{\bgent \ConditionalDirectionalSeparation \cgent \mid \AgentSubsetW}
      \eqfinp
      \tag{by~\eqref{eq:vertices-d-separated} in Definition~\ref{de:vertices-d-separated}}
  \end{align*}
  This ends the proof.
\end{proof}

\subsection{Proof of Proposition~\ref{pr:ConditionalDirectionalSeparation_seilpmi_relation}
  broken in six lemmata}

The first Lemma~\ref{lem:ActiveUndirectedPaths_subpaths_junctions_when_reconcatenating} is instrumental
for the following five lemmata. It is used to obtain active \undirectedEdgePath\ by concatenation.
Then, the following lemmata display elementary relational patterns (composition of
relations) whose deployment in \EdgePaths\ contain active
\undirectedEdgePaths. For each relation, an explicit active \undirectedEdgePath\
is built.

The binary relations used below have been introduced in
Definition~\ref{de:all_the_relations}, except for the two additional ones
defined in~\eqref{eq:CousinhoodStar} and
in~\eqref{eq:conditional_active_relationStar}.
For any positive integer \( n\geq 1\), we denote by
\( \1_{n}=\np{+1,\ldots,+1} \) (resp.  \( -\1_{n}=\np{-1,\ldots,-1} \)) the vector of
length~$n$ made of~$+1$ (resp. of~$-1$).

\subsubsection{Concatenation of \undirectedEdgePaths\ in a graph}
\label{Additional_material_on_undirectedEdgePaths_in_a_graph}

We recall that \( \UPATH\np{\OrientedGraph} \) in~\eqref{def:eopaths_gtzero} is the set of
\undirectedEdgePaths\ of positive length relative to the
graph~$\npOrientedGraph$.
Now, we develop the machinery to analyze active \undirectedEdgePaths\ by
considering decomposition into subpaths and junctions when reconcatenating.  For
this purpose, we need notation.

We denote by
\( \Tail{\ProjectionUndirectedPath}: \UPATH\np{\OrientedGraph} \to
\UPATH_{1}\np{\OrientedGraph} \) (resp.
\( \Head{\ProjectionUndirectedPath}: \UPATH\np{\OrientedGraph} \to
\UPATH_{1}\np{\OrientedGraph} \) the projection on the tail (resp. head) subpath
of an \undirectedEdgePath, defined, for
\(\UndirectedPath \in\UPATH_n\np{\OrientedGraph}\) and $n\ge 1$, by
\begin{subequations}
  \label{eq:ProjectionUndirectedPath}
  \begin{align}
    \Tail{\ProjectionUndirectedPath}\np{\UndirectedPath}
    &= 
      \Tail{\ProjectionUndirectedPath}
      \bp{ \nseqa{\np{\tail{\vertex_i},\head{\vertex_i},\orient_i}}{i\in\ic{1,n}}}
      = \ba{\np{\tail{\vertex_1},\head{\vertex_1},\orient_1}}
      \eqfinv
      \label{eq:ProjectionUndirectedPath_Tail}
    \\
    \Head{\ProjectionUndirectedPath}\np{\UndirectedPath}
    &= 
      \Head{\ProjectionUndirectedPath}
      \bp{ \nseqa{\np{\tail{\vertex_i},\head{\vertex_i},\orient_i}}{i\in\ic{1,n}}}
      = \ba{\np{\tail{\vertex_n},\head{\vertex_n},\orient_n }}
      \eqfinp
      \label{eq:ProjectionUndirectedPath_Head}
  \end{align}
\end{subequations}
\begin{figure}[hbtp]
  \begin{center}
    \fbox{\includegraphics[width=0.5\textwidth]{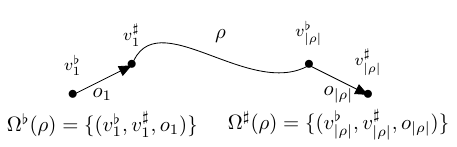}}
    \caption{
      Projection on the tail (resp. head) subpath of an \undirectedEdgePath
      \label{fig:Omega}}
  \end{center}
\end{figure}
The \undirectedEdgePath\ $\Tail{\ProjectionUndirectedPath}\np{\UndirectedPath}$
(resp.  $\Head{\ProjectionUndirectedPath}\np{\UndirectedPath}$), depicted in
Figure~\ref{fig:Omega}, is an \undirectedEdgePath\ of length one built with the
first (resp. last) extended-oriented edge of the \undirectedEdgePath\
$\np{\UndirectedPath}$.

The following
Lemma~\ref{lem:ActiveUndirectedPaths_subpaths_junctions_when_reconcatenating} is
a straightforward consequence of the definitions of active \undirectedEdgePaths\
in~\( \ActiveUndirectedPaths{\OrientedGraph} \) (see
Definition~\ref{de:ActivePaths}), of concatenation~$\ltimes$
in~\eqref{eq:concatenation_UndirectedPath} and of the projection
mappings~\(\Tail{\ProjectionUndirectedPath}\) and
\( \Head{\ProjectionUndirectedPath} \) in~\eqref{eq:ProjectionUndirectedPath}.
The proof is left to the reader.

\begin{lemma} (Coq \PY{k+kn}{Lemma}~~\PY{n+nf}{Active\_path\_cat})
  \label{lem:ActiveUndirectedPaths_subpaths_junctions_when_reconcatenating}
  Let \( \npOrientedGraph \) be a graph, 
  and $\AgentSubsetW\subset\VERTEX$ be a subset of vertices.
  Let \(\UndirectedPath \in \UPATH\np{\OrientedGraph} \)
  be an  \undirectedEdgePath\ of length~\( n \geq 2 \). We have that
  \begin{equation}
    \UndirectedPath \in \ActiveUndirectedPaths{\OrientedGraph}
    \iff
    \exists \UndirectedPathbis, \UndirectedPathter
    \in \ActiveUndirectedPaths{\OrientedGraph}
    \eqsepv 
    \UndirectedPath = \UndirectedPathbis \ltimes \UndirectedPathter
    \eqsepv
    \Head{\ProjectionUndirectedPath}\np{\UndirectedPathbis}
    \ltimes
    \Tail{\ProjectionUndirectedPath}\np{\UndirectedPathter}
    \in \ActiveUndirectedPaths{\OrientedGraph}
    \eqfinp
    \label{eq:ActiveUndirectedPaths_subpaths_junctions_when_reconcatenating}    
  \end{equation}
  \begin{leftbar}
    This Lemma~\ref{lem:ActiveUndirectedPaths_subpaths_junctions_when_reconcatenating} 
    is mathematically simple, but more involved in Coq where manipulation of active \undirectedEdgePaths\ is more tedious, as it 
    frequently requires proofs by induction on path lengths.
  \end{leftbar}
\end{lemma}

\subsubsection{Case $\bgent \ConditionalAscendent\cgent$} 
  
\begin{lemma}  (Coq \PY{k+kn}{Lemma}~\PY{n+nf}{C\_L\ref{lem:Bw_active}})
  \label{lem:Bw_active}
  Let \( \npOrientedGraph \) be a graph, 
  and $\AgentSubsetW\subset\VERTEX$ be a subset of vertices.
  For any vertices $\bgent$, $\cgent\in \VERTEX$, we have that 
  \begin{align}
    \bgent \ConditionalAscendent\cgent 
    \implies
    & 
      \mtext{there exists } \UndirectedPath \in \UPATH\np{\OrientedGraph}
      \mtext{ such that } \UndirectedPath = \Pair{\path, \1_{\cardinal{\path}} } \mtext{ and }
      \nonumber \\
    & 
      \UndirectedPath    \in 
      \DeploymentInUPaths{\np{\bgent,\cgent}}{\OrientedGraph} 
      \cap \ActiveUndirectedPaths{\OrientedGraph} 
      \eqfinp
      \label{eq:BlockedUndirectedPaths_elementary_Delta_ComplementaryAgentSubsetW_EDGE} 
  \end{align}
\end{lemma}

\begin{proof}
  We prove the implication~\eqref{eq:BlockedUndirectedPaths_elementary_Delta_ComplementaryAgentSubsetW_EDGE}.
  Let $\bgent$, $\cgent\in \VERTEX$ be such that
  \( \bgent \ConditionalAscendent\cgent \).   
  By~\eqref{eq:conditional_ascendent_relation}, we have that
  \( \ConditionalAscendent\cgent =
  \EDGE\npTransitiveReflexiveClosure{\Delta_{\Complementary{\AgentSubsetW}}\EDGE}\),
  hence that \( \bgent
  \EDGE\npTransitiveReflexiveClosure{\Delta_{\Complementary{\AgentSubsetW}}\EDGE}\cgent
  \).
    As \( \npTransitiveReflexiveClosure{\Delta_{\Complementary{\AgentSubsetW}}\EDGE}=
  \Delta \cup \cup_{n=1}^{\infty}
  \np{\Delta_{\Complementary{\AgentSubsetW}}\EDGE}^{n} \) by definition, 
  if \( \bgent \EDGE\npTransitiveReflexiveClosure{\Delta_{\Complementary{\AgentSubsetW}}\EDGE}\cgent \),
  then either 
  \( \bgent \EDGE \cgent \)
  or there exists $n\geq 1$ such that 
  \( \bgent \EDGE\np{\Delta_{\Complementary{\AgentSubsetW}}\EDGE}^{n} \cgent \).
  Thus, we consider two cases.

  If \( \bgent \EDGE \cgent \),
  the \undirectedEdgePath\ \( \ba{ \np{\bgent,\cgent,+1}} \)
  is both in \( \DeploymentInUPaths{\np{\bgent,\cgent}}{\OrientedGraph} \),
  by definition~\eqref{eq:DeploymentInUPaths},
  and belongs to~\( \ActiveUndirectedPaths{\OrientedGraph} \),
  as it is of length~1 hence is active (see Definition~\ref{de:ActivePaths}). 

  If \( \bgent\EDGE \np{\Delta_{\Complementary{\AgentSubsetW}}\EDGE}^{n} \cgent \),
  with $n\geq 1$, then 
  there exists a sequence \( \sequence{\vertex_i}{i\in\ic{0,n+1}} \) in~$\VERTEX$ such that 
  \( \vertex_{0}=\bgent \), \( \vertex_{n+1}=\cgent \), 
  and 
  \( \vertex_{i-1} \EDGE \vertex_{i} \),
  \( \vertex_{i} \in \Complementary{\AgentSubsetW} \),
  \( \vertex_{i} \EDGE \vertex_{i+1} \),
  for $i\in\ic{1,n}$.
  The following 
  \( \UndirectedPath= \bPair{\sequence{ \np{\vertex_i,\vertex_{i+1}}
    }{i\in\ic{0,n}}, \1_{n+1} } \)
  is an \undirectedEdgePath\ that 
  belongs to
  \( \DeploymentInUPaths{\np{\bgent,\cgent}}{\OrientedGraph} \),
  as \( \Projection\np{\UndirectedPath}=\np{\bgent,\cgent} \),
  and also to~\( \ActiveUndirectedPaths{\OrientedGraph} \).
  Indeed, all the \undirectedEdgeSubPaths\ 
  \( \ba{ \np{\vertex_{i-1},\vertex_{i},+1},
    \np{\vertex_{i},\vertex_{i+1},+1}} \)
  for $i\in\ic{1,n}$, satisfy Item~\ref{it:ActivePaths_case1} in
  Definition~\ref{de:ActivePaths} 
  because \( \vertex_{i} \in \Complementary{\AgentSubsetW} \).
  \medskip
  
  This ends the proof.
\end{proof}

\subsubsection{Case $\bgent \ConverseConditionalAscendent\cgent$}

\begin{lemma} (Coq \PY{k+kn}{Lemma}~\PY{n+nf}{C\_L\ref{lem:Bmw_active}})
  \label{lem:Bmw_active}
  Let \( \npOrientedGraph \) be a graph, 
  and $\AgentSubsetW\subset\VERTEX$ be a subset of vertices.
  For any vertices $\bgent$, $\cgent\in \VERTEX$, we have that
  \begin{align}
    \bgent \ConverseConditionalAscendent\cgent 
    \implies
    & 
      \mtext{there exists }\UndirectedPath \in \UPATH\np{\OrientedGraph}
      \mtext{ such that } \UndirectedPath = \Pair{\path, - \1_{\cardinal{\path}} } \mtext{ and }
      \nonumber \\
    & 
      \UndirectedPath 
      \in 
      \DeploymentInUPaths{\np{\bgent,\cgent}}{\OrientedGraph} 
      \cap \ActiveUndirectedPaths{\OrientedGraph}
      \eqfinp
      \label{eq:BlockedUndirectedPaths_elementary_Converse_EDGE_Delta_ComplementaryAgentSubsetW} 
  \end{align}
\end{lemma}

\begin{proof}
  We prove the implication~\eqref{eq:BlockedUndirectedPaths_elementary_Converse_EDGE_Delta_ComplementaryAgentSubsetW}
  in the same way as for implication~\eqref{eq:BlockedUndirectedPaths_elementary_Delta_ComplementaryAgentSubsetW_EDGE}
  in Lemma~\ref{lem:Bw_active}
  (here, all the \undirectedEdgeSubPaths\ satisfy Item~\ref{it:ActivePaths_case2} in
  Definition~\ref{de:ActivePaths}).
\end{proof}

\subsubsection{Case $\bgent\ConditionalCommonCause\cgent$}

\begin{lemma} (Coq \PY{k+kn}{Lemma}~~\PY{n+nf}{C\_L\ref{lem:Kw_active}})
  \label{lem:Kw_active}
  Let \( \npOrientedGraph \) be a graph, 
  and $\AgentSubsetW\subset\VERTEX$ be a subset of vertices.
  For any vertices $\bgent$, $\cgent\in \VERTEX$, we have that 
  \begin{align}
    \bgent\ConditionalCommonCause\cgent
    \implies
    & 
      \mtext{there exists }\UndirectedPath=\Pair{\path,\orient} \in
      \UPATH\np{\OrientedGraph},
      \mtext{with } \cardinal{\UndirectedPath} \geq 2
      \nonumber \\
    & \mtext{and } \orient=\np{-1,\ldots,+1} 
      \in \na{-1}\times\na{-1,+1}^{\cardinal{\UndirectedPath}-2}\times\na{+1},
      \mtext{ such that}
      \nonumber \\
    & 
      \UndirectedPath=\Pair{\path,\orient}=\Pair{\path, \np{-1,\ldots,+1} } 
      \in 
      \DeploymentInUPaths{\np{\bgent,\cgent}}{\OrientedGraph}
      \cap \ActiveUndirectedPaths{\OrientedGraph} 
      \eqfinp
      \label{eq:BlockedUndirectedPaths_elementary_TransitiveClosure_Converse_EDGE_Delta_ComplementaryAgentSubsetW} 
  \end{align}
\end{lemma}

\begin{proof} 
  We
  prove the implication~\eqref{eq:BlockedUndirectedPaths_elementary_TransitiveClosure_Converse_EDGE_Delta_ComplementaryAgentSubsetW}.
  Let $\bgent$, $\cgent\in \VERTEX$ be such that
  \( \bgent\ConditionalCommonCause\cgent \).
  By definition~\eqref{eq:common_cause} of \(\ConditionalCommonCause\),
  we obtain that \( \bgent \ConverseConditionalAscendent
  \Delta_{\Complementary{\AgentSubsetW}} \ConditionalAscendent \cgent \).
  As a consequence of the definition of the composition of relations, 
  there exists \( \dgent \in \Complementary{\AgentSubsetW} \) such that 
  \( \bgent
  \ConverseConditionalAscendent
  \dgent \)
  and
  \( \dgent  \ConditionalAscendent \cgent \).
  Thus, by~\eqref{eq:BlockedUndirectedPaths_elementary_Converse_EDGE_Delta_ComplementaryAgentSubsetW},
  there exists \( \UndirectedPathbis =\Pair{\pathbis, -\1_{\cardinal{\pathbis}} } 
  \in \DeploymentInUPaths{\np{\bgent,\dgent}}{\OrientedGraph} 
  \cap \ActiveUndirectedPaths{\OrientedGraph}\) and
  by~\eqref{eq:BlockedUndirectedPaths_elementary_Delta_ComplementaryAgentSubsetW_EDGE},
  there exists \( \UndirectedPathter =\Pair{\pathter,  \1_{\cardinal{\pathter}} }
  \in \DeploymentInUPaths{\np{\dgent,\cgent}}{\OrientedGraph} 
  \cap \ActiveUndirectedPaths{\OrientedGraph} \).
  We consider the \undirectedEdgePath\ 
  \( \UndirectedPath = \UndirectedPathbis \ltimes \UndirectedPathter  \in
  \DeploymentInUPaths{\np{\bgent,\cgent}}{\OrientedGraph} \)
  obtained by concatenation as in~\eqref{eq:concatenation_UndirectedPath},
  and which is such that \( \cardinal{\UndirectedPath} \geq 2 \). 
  We claim that \( \UndirectedPath \in
  \ActiveUndirectedPaths{\OrientedGraph} \).
  Indeed, \( \UndirectedPathbis \in
  \ActiveUndirectedPaths{\OrientedGraph} \)
  and \(\UndirectedPathter \in
  \ActiveUndirectedPaths{\OrientedGraph} \)
  by assumption, so that, by
  Equation~\eqref{eq:ActiveUndirectedPaths_subpaths_junctions_when_reconcatenating} in
  Lemma~\ref{lem:ActiveUndirectedPaths_subpaths_junctions_when_reconcatenating},
  it only remains to show that 
  \[
    \ba{ \np{\vertexbis,\dgent,-1}, \np{\dgent, \vertexter,+1}}
    =
    {\Head{\ProjectionUndirectedPath}\np{\UndirectedPathbis}
      \ltimes
      \Tail{\ProjectionUndirectedPath}\np{\UndirectedPathter}}
    \in \ActiveUndirectedPaths{\OrientedGraph}
    \eqfinv
  \]
  where 
  \( \np{\vertexbis,\dgent} \) is the first edge of the \undirectedEdgePath\ 
  \( \UndirectedPathbis = \bPair{\pathbis, -\1_{\cardinal{\pathbis}} } \) and 
  \( \np{\dgent, \vertexter } \) is the last edge of the \undirectedEdgePath\ 
  \( \UndirectedPathter =\bPair{\pathter,  \1_{\cardinal{\pathter}} } \).
  Now, the above subpath satisfies Item~\ref{it:ActivePaths_case3} in
  Definition~\ref{de:ActivePaths} 
  because \( \vertex_{i} \in \Complementary{\AgentSubsetW} \).
  We deduce that \( \UndirectedPath = \UndirectedPathbis \ltimes
  \UndirectedPathter
  = \bPair{ \np{\pathbis,\pathter}, 
    \np{ -\1_{\cardinal{\pathbis}}, \1_{\cardinal{\pathter}} } } \in 
  \DeploymentInUPaths{\np{\bgent,\cgent}}{\OrientedGraph}
  \cap \ActiveUndirectedPaths{\OrientedGraph} \).
  \medskip

  This ends the proof.
\end{proof}

\subsubsection{Case $\bgent\CousinhoodStar\cgent$}

\begin{lemma}  (Coq \PY{k+kn}{lemmata}~~\PY{n+nf}{C\_L\ref{lem:Cwstar_active}\PYZus{}*}) 
  \label{lem:Cwstar_active}
  Let \( \npOrientedGraph \) be a graph, 
  and $\AgentSubsetW\subset\VERTEX$ be a subset of vertices.
  For any vertices $\bgent$, $\cgent\in \VERTEX$, we have that 
  \begin{align}
    \bgent\CousinhoodStar\cgent
    \implies
    & 
      \bgent \in \TransitiveReflexiveClosureOfSet{\AgentSubsetW}, \,
      \cgent \in \TransitiveReflexiveClosureOfSet{\AgentSubsetW} \mtext{ and }
      \bgent = \cgent \mtext{ or }
      \nonumber \\
    & 
      \mtext{there exists }\UndirectedPath=\Pair{\path,\orient} \in \UPATH\np{\OrientedGraph},
      \mtext{with } \cardinal{\UndirectedPath} \geq 2                
      \nonumber \\
    & \mtext{and } \orient=\np{-1,\ldots,+1} 
      \in \na{-1}\times\na{-1,+1}^{\cardinal{\UndirectedPath}-2}\times\na{+1},
      \mtext{ such that}
      \nonumber \\
    & 
      \UndirectedPath=\Pair{\path,\orient}=\Pair{\path, \np{-1,\ldots,+1} } 
      \in 
      \DeploymentInUPaths{\np{\bgent,\cgent}}{\OrientedGraph}
      \cap \ActiveUndirectedPaths{\OrientedGraph} 
      \eqfinp
      \label{eq:BlockedUndirectedPaths_elementary_CousinhoodStar}
  \end{align}
\end{lemma}

\begin{proof}
  We prove the
  implication~\eqref{eq:BlockedUndirectedPaths_elementary_CousinhoodStar}.  We
  suppose that \( \bgent\CousinhoodStar\cgent \).  As
  \( \CousinhoodStar = \Delta_{\TransitiveReflexiveClosureOfSet{\AgentSubsetW}} \cup
  \bpTransitiveClosure{ \Delta_{\TransitiveReflexiveClosureOfSet{\AgentSubsetW}}
    \ConditionalCommonCause \Delta_{\TransitiveReflexiveClosureOfSet{\AgentSubsetW}}
  } \) by~\eqref{eq:CousinhoodStar}, we consider three cases: either
  \( \bgent\Delta_{\TransitiveReflexiveClosureOfSet{\AgentSubsetW}}\cgent \), or
  \( \bgent
  \Delta_{\TransitiveReflexiveClosureOfSet{\AgentSubsetW}}\ConditionalCommonCause
  \Delta_{\TransitiveReflexiveClosureOfSet{\AgentSubsetW}} \cgent \) or there exists
  $n \geq 1$ such that
  \( \bgent \bp{\Delta_{\TransitiveReflexiveClosureOfSet{\AgentSubsetW}}
    \ConditionalCommonCause
    \Delta_{\TransitiveReflexiveClosureOfSet{\AgentSubsetW}}}^{n+1} \cgent \).

  Suppose that
  \( \bgent\Delta_{\TransitiveReflexiveClosureOfSet{\AgentSubsetW}}\cgent \).  Then
  \( \bgent=\cgent \), and thus we have obtain the
  implication~\eqref{eq:BlockedUndirectedPaths_elementary_CousinhoodStar}.
  
  Suppose that \( \bgent
  \Delta_{\TransitiveReflexiveClosureOfSet{\AgentSubsetW}}\ConditionalCommonCause
  \Delta_{\TransitiveReflexiveClosureOfSet{\AgentSubsetW}} \cgent \).
  Then, \( \bgent \in \TransitiveReflexiveClosureOfSet{\AgentSubsetW} \), 
  \( \cgent \in \TransitiveReflexiveClosureOfSet{\AgentSubsetW} \) and
  \( \bgent\ConditionalCommonCause \cgent \).
  Therefore, 
  by~\eqref{eq:BlockedUndirectedPaths_elementary_TransitiveClosure_Converse_EDGE_Delta_ComplementaryAgentSubsetW}
  there exists \( \UndirectedPath \in \UPATH\np{\OrientedGraph} \),
  with \( \cardinal{\UndirectedPath} \geq 2 \),
  and \( \orient 
  \in\na{-1}\times\na{-1,+1}^{\cardinal{\UndirectedPath}-2}\times\na{+1} \)
  such that \( \Pair{\path, \orient } \in \DeploymentInUPaths{\np{\bgent,\cgent}}{\OrientedGraph}
  \cap \ActiveUndirectedPaths{\OrientedGraph} \).
  
  Suppose that
  \( \bgent \bp{\Delta_{\TransitiveReflexiveClosureOfSet{\AgentSubsetW}}
    \ConditionalCommonCause
    \Delta_{\TransitiveReflexiveClosureOfSet{\AgentSubsetW}}}^{n+1} \cgent \).
  Then, there exists a sequence 
  \( \sequence{\dgent_i}{i\in\ic{0,n+1}} \) in~\( \TransitiveReflexiveClosureOfSet{\AgentSubsetW} \)
  such that \( \dgent_0=\bgent \), \( \dgent_{n+1}=\cgent \) 
  and \( \dgent_i \ConditionalCommonCause \dgent_{i+1} \) for $i\in\ic{0,n}$.
  Therefore, 
  by~\eqref{eq:BlockedUndirectedPaths_elementary_TransitiveClosure_Converse_EDGE_Delta_ComplementaryAgentSubsetW},
  there exists a sequence 
  \( \sequence{\UndirectedPath^{(i)}}{i\in\ic{0,n}} \) in~\( \UPATH\np{\OrientedGraph}\)
  of \undirectedEdgePaths\ such that
  \( \cardinal{\UndirectedPath^{(i)}} \geq 2 \) for $i\in\ic{0,n}$,
  and that 
  \( \UndirectedPath^{(i)} = \Pair{\path^{(i)}, \orient^{(i)} }
  \in \DeploymentInUPaths{\np{\dgent_i,\dgent_{i+1}}}{\OrientedGraph}
  \cap \ActiveUndirectedPaths{\OrientedGraph} \)
  for $i\in\ic{0,n}$, where \( \orient^{(i)} \) is of the form~\( \np{-1,\ldots,+1}
  \).
  We define the \undirectedEdgePath\ 
  \( \UndirectedPath = \UndirectedPath^{(0)} \ltimes \cdots \ltimes
  \UndirectedPath^{(n)}  \in
  \DeploymentInUPaths{\np{\bgent,\cgent}}{\OrientedGraph} \),
  by iterated (associative) concatenation as
  in~\eqref{eq:concatenation_UndirectedPath}, which is such that
  \( \cardinal{\UndirectedPath} \geq 2 \) and that
  \[ \UndirectedPath = \Pair{\path^{(0)} \ltimes \cdots \ltimes \path^{(n)},
      \np{ \orient^{(0)}, \ldots, \orient^{(n)} }}
    \mtext{ where }
    \orient^{(i)} \in \na{-1} \times \na{-1,+1}^{\cardinal{\path_i}-2} \times \na{+1}
    \eqsepv
    \forall i\in\ic{0,n}
    \eqfinp 
  \]
  We claim that \( \UndirectedPath \in
  \ActiveUndirectedPaths{\OrientedGraph} \).
  Indeed, \( \UndirectedPath^{(i)} \in \ActiveUndirectedPaths{\OrientedGraph} \)
  for $i\in\ic{0,n}$,
  so that, by
  Equation~\eqref{eq:ActiveUndirectedPaths_subpaths_junctions_when_reconcatenating} in
  Lemma~\ref{lem:ActiveUndirectedPaths_subpaths_junctions_when_reconcatenating},
  it only remains to show that, for $i\in\ic{0,n-1}$, 
  \begin{align}
    \ba{ \np{\tail{\vertex_{i}},\dgent_{i+1},+1},  \np{\dgent_{i+1}, \head{\vertex_{i+1}},-1}}
    &= 
      {\Head{\ProjectionUndirectedPath}\np{\UndirectedPath_{i}}
      \ltimes
      \Tail{\ProjectionUndirectedPath}\np{\UndirectedPath_{i+1}}}
      \in \ActiveUndirectedPaths{\OrientedGraph}
      \eqfinv
  \end{align}
  where 
  \( \np{\tail{\vertex}_{i},\dgent_{i+1}} \) is the last edge of the \undirectedEdgePath\ 
  \( \UndirectedPath^{(i)} = \Pair{\path^{(i)}, \orient^{(i)}} \) and 
  \( \np{\dgent_{i+1}, \head{\vertex}_{i+1}} \) is the first edge of the \undirectedEdgePath\ 
  \( \UndirectedPath^{(i+1)} =\Pair{\path^{(i+1)}, \orient^{(i+1)} } \).
  As \( \dgent_{i+1} \in \TransitiveReflexiveClosureOfSet{\AgentSubsetW} \), for $i\in\ic{0,n-1}$,
  and because of the orientation \( \np{-1,+1} \),
  all the above subpaths satisfy Item~\ref{it:ActivePaths_case4} in
  Definition~\ref{de:ActivePaths}. 
  We conclude that \( \UndirectedPath = \UndirectedPath^{(0)} \ltimes \cdots \ltimes \UndirectedPath^{(n)} \in 
  \DeploymentInUPaths{\np{\bgent,\cgent}}{\OrientedGraph}
  \cap \ActiveUndirectedPaths{\OrientedGraph} \).
  \medskip
      
  This ends the proof.
\end{proof}

\subsubsection{Case $\bgent \ConditionalActiveThreeStar \cgent$}

\begin{lemma} (Coq \PY{k+kn}{Lemma}~\PY{n+nf}{C\_L\ref{lem:BkCwsBmKBk}})
  \label{lem:BkCwsBmKBk}
  Let \( \npOrientedGraph \) be a graph, 
  and $\AgentSubsetW\subset\VERTEX$ be a subset of vertices.
  For any vertices $\bgent$, $\cgent\in \VERTEX$, we have that 
  \begin{align}
    \bgent\ConditionalActiveThreeStar\cgent
    \implies
    &
      \mtext{there exists }\UndirectedPath \in \UPATH\np{\OrientedGraph},
      \mtext{with } \cardinal{\UndirectedPath} \geq 2,
      \mtext{ such that}
      \nonumber \\
    & 
      \UndirectedPath 
      \in 
      \DeploymentInUPaths{\np{\bgent,\cgent}}{\OrientedGraph}
      \cap \ActiveUndirectedPaths{\OrientedGraph} 
      \eqfinp
      \label{eq:BlockedUndirectedPaths_elementary_TransitiveClosure_TransitiveClosure}
  \end{align}
\end{lemma}

\begin{proof}
  Suppose that $\bgent \ConditionalActiveThreeStar \cgent$.
  Therefore, there
  exist \( \dgent_1 \) and \( \dgent_2 \) in~$\VERTEX$ such that
  \[
    \bgent \bp{\ConditionalAscendent \cup \ConditionalCommonCause}\dgent_1
    \text{ and }
    \dgent_1 \CousinhoodStar \dgent_2
    \text{ and }
    \dgent_2 \bp{\ConverseConditionalAscendent \cup \ConditionalCommonCause} \cgent
    \eqfinp
  \]
  
  We are going to display an \undirectedEdgePath\
  \( \UndirectedPath \in \DeploymentInUPaths{\np{\dgent,\bgent}}{\graph} \cap
  \ActiveUndirectedPaths{\graph} \).
  \begin{itemize}
  \item
    Considering the left hand side 
    \( \bgent\bp{\ConditionalAscendent \cup \ConditionalCommonCause} \dgent_1\) and
    using Lemma~\ref{lem:Bw_active} and~\ref{lem:Kw_active}, 
    we obtain --- either by~\eqref{eq:BlockedUndirectedPaths_elementary_Delta_ComplementaryAgentSubsetW_EDGE} applied to
    \( \bgent \ConditionalAscendent \dgent_1 \), or
    by~\eqref{eq:BlockedUndirectedPaths_elementary_TransitiveClosure_Converse_EDGE_Delta_ComplementaryAgentSubsetW}
    applied to \( \bgent \ConditionalCommonCause \dgent_1\) --- 
    that there exists
    \[
      \UndirectedPath^{(1)}= \Pair{\path^{(1)}, \orient^{(1)} }
      =\Pair{\path^{(1)}, \np{\ldots,+1} } \in 
      \DeploymentInUPaths{\np{\bgent,\dgent_1}}{\graph} 
      \cap \ActiveUndirectedPaths{\graph}
      \eqfinp
    \]
    
  \item 
    In the same way, considering the right hand side
    \( \dgent_2 \bp{\ConverseConditionalAscendent \cup \ConditionalCommonCause} \cgent \)
    and using Lemma~\ref{lem:Bmw_active} and~\ref{lem:Kw_active},
    we obtain  --- either by~\eqref{eq:BlockedUndirectedPaths_elementary_Converse_EDGE_Delta_ComplementaryAgentSubsetW}
    applied to  \( \dgent_2 \ConverseConditionalAscendent \cgent \),
    or by~\eqref{eq:BlockedUndirectedPaths_elementary_TransitiveClosure_Converse_EDGE_Delta_ComplementaryAgentSubsetW}
    applied to \( \dgent_2 \ConditionalCommonCause \cgent \) --- 
    that there exists
    \[
      \UndirectedPath^{(2)} =\Pair{\path_{2}, \orient_{2} }
      =\Pair{\path^{(2)}, \np{-1,\ldots} } \in 
      \DeploymentInUPaths{\np{\dgent_2,\bgent}}{\graph} 
      \cap \ActiveUndirectedPaths{\graph}
      \eqfinp
    \]
  \item Considering the middle expression
    \( \dgent_1 \CousinhoodStar \dgent_2 \) and using
    Lemma~\ref{lem:Cwstar_active}, we obtain
    by~\eqref{eq:BlockedUndirectedPaths_elementary_CousinhoodStar}
    that \( \dgent_1 \in \TransitiveReflexiveClosureOfSet{\AgentSubsetW} \),
    \( \dgent_2 \in \TransitiveReflexiveClosureOfSet{\AgentSubsetW} \), and that
    there exists 
    \[
      \UndirectedPathbis
      =\Pair{\pathbis, \orientbis }
      =\Pair{\pathbis, \np{-1,\ldots,+1} } 
      \in        \DeploymentInUPaths{\np{\dgent_1,\dgent_2}}{\graph} 
      \cap \ActiveUndirectedPaths{\graph}
      \eqfinp
    \]
  \end{itemize}
  
  We consider the \undirectedEdgePath\ 
  \( \UndirectedPath = \UndirectedPath^{(1)} \ltimes \UndirectedPathbis \ltimes \UndirectedPath^{(2)} \)
  obtained by concatenation as in~\eqref{eq:concatenation_UndirectedPath}.
  By construction, we have that \( \UndirectedPath \in
  \DeploymentInUPaths{\np{\bgent,\cgent}}{\graph} \).
  We claim that \( \UndirectedPath \in
  \ActiveUndirectedPaths{\graph} \).
  Indeed, \( \UndirectedPath^{(1)}, \UndirectedPathbis, \UndirectedPath^{(2)}
  \in \ActiveUndirectedPaths{\graph} \) 
  by assumption,
  so that, by
  Equation~\eqref{eq:ActiveUndirectedPaths_subpaths_junctions_when_reconcatenating} in
  Lemma~\ref{lem:ActiveUndirectedPaths_subpaths_junctions_when_reconcatenating},
  it only remains to show that
  \begin{align}
    \ba{\np{\tail{\vertex_{1}},\dgent_{1},+1},  \np{\dgent_{1}, \head{\vertex}_1,-1 }}
    &= 
      {\Head{\ProjectionUndirectedPath}\np{\UndirectedPath^{(1)}}
      \ltimes
      \Tail{\ProjectionUndirectedPath}\np{\UndirectedPathbis}}
      \in \ActiveUndirectedPaths{\OrientedGraph}
      \eqfinv
      \label{eq:dgent-one-property}
      \intertext{where 
      \( \np{\tail{\vertex_{1}},\dgent_{1}} \) is the last edge of the \undirectedEdgePath\ 
      \( \UndirectedPath^{(1)}\) and 
      \( \np{\dgent_{1}, \head{\vertex}_1 } \) is the first edge of the \undirectedEdgePath\ 
      \( \UndirectedPathbis \), and that }
      \ba{  \np{\tail{\vertex}_2,\dgent_{2},+1 }, \np{\dgent_{2}, \head{\vertex_{2}},-1 } }
    &= 
      {\Head{\ProjectionUndirectedPath}\np{\UndirectedPathbis}
      \ltimes
      \Tail{\ProjectionUndirectedPath}\np{\UndirectedPath^{(2)}}}
      \in \ActiveUndirectedPaths{\OrientedGraph}
      \eqfinv
      \label{eq:dgent-two-property}
  \end{align}
  where 
  \( \np{\tail{\vertex}_2,\dgent_{2}} \) is the last edge of the \undirectedEdgePath\ 
  \( \UndirectedPathbis\) and 
  \( \np{\dgent_{2}, \head{\vertex_{2}} } \) is the first edge of the \undirectedEdgePath\ 
  \( \UndirectedPath^{(2)} \). 
  As \( \dgent_1,  \dgent_2 \in \TransitiveReflexiveClosureOfSet{\AgentSubsetW} \)
  and because of the orientation \( \np{-1,+1} \),
  the two subpaths hereabove satisfy Item~\ref{it:ActivePaths_case4} in
  Definition~\ref{de:ActivePaths}. 
  We conclude that \( \UndirectedPath \in 
  \DeploymentInUPaths{\np{\bgent,\cgent}}{\graph}
  \cap \ActiveUndirectedPaths{\graph} \).%
  \medskip 
  
  This ends the proof.
\end{proof}

\section{Proof of $\neg \np{\bgent \protect\ConditionalDirectionalSeparation \cgent \mid \AgentSubsetW}
  \implies \bgent \ConditionalActiveStar \cgent $
}
\label{Proof_of_Theorem_isimplied}

The following Proposition~\ref{th:main_isimplied}, that we are going to prove,
is the second half of the proof of
Theorem~\ref{th:ConditionalDirectionalSeparation_IFF_relation}. It relies on the three
Lemmata~\ref{le:lemma8}, \ref{le:active-length-2} and~\ref{lem:induction-lemma},
postponed at the end of this Appendix~\ref{Proof_of_Theorem_isimplied}.

\subsection{Proposition~\ref{th:main_isimplied}}

\begin{proposition} (Coq \PY{k+kn}{Proposition}~\PY{n+nf}{D\_P\ref{th:main_isimplied}})
  \label{th:main_isimplied}
  Let \( \npOrientedGraph \) be a graph, 
  and $\AgentSubsetW\subset\VERTEX$ be a subset of vertices.
  Let \( \bgent, \cgent \in\AGENT \) be vertices.
  We have the implication
  \begin{equation}
    \neg \bp{\bgent \ConditionalDirectionalSeparation \cgent \mid \AgentSubsetW}
    \implies
    \bgent \ConditionalActiveStar \cgent
    \label{eq:d-sep-to-rel}
  \end{equation}
  where $\ConditionalActiveStar$ is the star conditional active
  relation~\eqref{eq:conditional_active_relationStar}
  and $\ConditionalDirectionalSeparation$ is the
  conditional directional separation relation~\eqref{eq:vertices-d-separated}.
\end{proposition}
\begin{proof}
  Let \( \bgent,\cgent \in \AGENT \) be vertices.  We show that
  \( \neg \bp{\bgent \ConditionalDirectionalSeparation \cgent \mid \AgentSubsetW}
  \implies \bgent \ConditionalActiveStar \cgent \) or, equivalently
  (see~\eqref{eq:vertices-d-separated} in
  Definition~\ref{de:vertices-d-separated}), that
  \( \np{\bgent = \cgent} \vee \bp{\DeploymentInUPaths{\na{\np{\bgent,\cgent}}}{\graph} \cap
    \ActiveUndirectedPaths{\graph} \neq\emptyset} \implies \bgent {\ConditionalActiveStar} \cgent
  \). First, when $\bgent = \cgent$ it is clear that $\bgent {\ConditionalActiveStar} \cgent$
  as \( \Delta \subset \ConditionalActiveStar \)
  by~\eqref{eq:conditional_active_relationStar}.
  Second, we assume that there exists an \undirectedEdgePath\
  \( \UndirectedPath \in \UPATH\np{\graph} \) joining the two
  vertices~$\bgent$ and $\cgent$ and such that \( \UndirectedPath \) is active.
  We are going to prove that $\bgent \ConditionalActiveStar \cgent$.
  \begin{itemize}
  \item 
  If the path length of~\( \UndirectedPath \) is equal to~1, then,
  by definition~\eqref{def:eopaths_gtzero} of
  $\UPATH_{1}\np{\OrientedGraph}$, we necessarily have that
  either $\bgent \Precedence\cgent$, or $\bgent \Converse{\Precedence}\cgent$.
  Now, as \( \Precedence \subset \ConditionalAscendent \subset
  \ConditionalActiveStar \)
  by~\eqref{eq:conditional_active_relationStar}, 
  as \( \Converse{\Precedence} \subset \ConverseConditionalAscendent \subset
  \ConditionalActiveStar \)
  by~\eqref{eq:conditional_active_relationStar}, 
  we conclude that $\bgent \ConditionalActiveStar \cgent$.

\item 
  If the path length of~\( \UndirectedPath \) is $\ge 2$, we prove by
  induction that we have either $\bgent \ConditionalActivePlus \cgent$ or
  $\bgent \ConditionalActiveMinus \cgent$, where the two relations
  $\ConditionalActivePlus$ and $\ConditionalActiveMinus$ are defined
  in~Equations~\eqref{eq:ConditionalActivePlus} and~\eqref{eq:ConditionalActiveMinus}.
  \begin{itemize}
    \item The case where the path length of~\( \UndirectedPath \) is equal to~2
    is treated in Lemma~\ref{le:active-length-2}.
    \item The proof by induction 
      on the path length is done in Lemma~\ref{lem:induction-lemma}.
    \item We therefore conclude that $\bgent \ConditionalActiveStar \cgent$ since
      $\ConditionalActiveStar = \Delta \cup \ConditionalActivePlus \cup
      \ConditionalActiveMinus$ by~\eqref{eq:ConditionalActivePlus_cup_Minus},
      obtained in Lemma~\ref{le:lemma8}.
    \end{itemize}
  \end{itemize}
  \medskip
  
  This ends the proof.
\end{proof}

\subsection{Proof of Proposition~\ref{th:main_isimplied} broken in three lemmata}

\subsubsection{Definition and properties of $\ConditionalActivePlus$ and $\ConditionalActiveMinus$}

We introduce two binary relation $\ConditionalActivePlus$ and $\ConditionalActiveMinus$ and establish three
properties which are instrumental in the next lemmata~\ref{le:active-length-2} and~\ref{lem:induction-lemma}.

\begin{lemma}  (Coq \PY{k+kn}{Lemma}~\PY{n+nf}{D\_L\ref{le:lemma8}\_E\ref{eq:ConditionalActivePlus_cup_Minus}}
  , \PY{n+nf}{D\_L\ref{le:lemma8}\_E\ref{eq:ConditionalActivePlus_supset}}
  and  \PY{n+nf}{D\_L\ref{le:lemma8}\_E\ref{eq:ConditionalActiveMinus_supset}})
  \label{le:lemma8}
  The two following relations
  \begin{subequations}
    \begin{align}
      \ConditionalActivePlus
      &=
        \ConditionalAscendent \cup \ConditionalCommonCause
        \cup 
        \Bp{\bp{\ConditionalAscendent \cup \ConditionalCommonCause}
        \CousinhoodStar
        {\ConditionalCommonCause}}
        \eqfinv
        \label{eq:ConditionalActivePlus}
      \\
      \ConditionalActiveMinus
      &= \ConverseConditionalAscendent \cup
        \Bp{
        \bp{\ConditionalAscendent \cup \ConditionalCommonCause}
        \CousinhoodStar
        \ConverseConditionalAscendent
        }
        \eqfinv
        \label{eq:ConditionalActiveMinus}      
    \end{align}
  \end{subequations}
  satisfy the following properties
  \begin{subequations}
    \begin{align}
      \ConditionalActiveStar
      &=
        \Delta \cup \ConditionalActivePlus \cup \ConditionalActiveMinus
        \eqfinv
        \label{eq:ConditionalActivePlus_cup_Minus}      
      \\
      \ConditionalActivePlus
      & \supset 
        \ConditionalActiveMinus \Delta_{\Complementary{\AgentSubsetW}}
        {\Precedence}
        \eqfinv
        \label{eq:ConditionalActivePlus_supset}
      \\
      \ConditionalActiveMinus
      & \supset 
        \ConditionalActivePlus \Delta_{\TransitiveReflexiveClosureOfSet{\AgentSubsetW}}
        \Converse{\Precedence}
        \eqfinp 
        \label{eq:ConditionalActiveMinus_supset}        
    \end{align}
  \end{subequations}  
\end{lemma}

\begin{proof}
  \quad
  \noindent $\bullet$ (\PY{k+kn}{Lemma}~~\PY{n+nf}{D\_L\ref{le:lemma8}\PYZus{}E\ref{eq:ConditionalActivePlus_cup_Minus}}\PY{o})
  We prove~\eqref{eq:ConditionalActivePlus_cup_Minus} as follows:
  \begin{align*}
    \ConditionalActiveStar
    &=
      \Delta \cup 
      \ConditionalAscendent \cup \ConverseConditionalAscendent \cup \ConditionalCommonCause
      \cup
      \bp{\ConditionalAscendent \cup \ConditionalCommonCause}
      \CousinhoodStar
      \bp{\ConverseConditionalAscendent \cup \ConverseConditionalCommonCause}
      \tag{by definition~\eqref{eq:conditional_active_relationStar} of~$\ConditionalActiveStar$}
    \\
    &=
    \Delta \cup 
      \underbrace{  \ConditionalAscendent \cup \ConditionalCommonCause
       \cup 
        \bp{\ConditionalAscendent \cup \ConditionalCommonCause}
        \CousinhoodStar {\ConditionalCommonCause} }_{=\ConditionalActivePlus \textrm{by~\eqref{eq:ConditionalActivePlus}}}
      \cup
      \underbrace{\ConverseConditionalAscendent \cup
      \bp{\ConditionalAscendent \cup \ConditionalCommonCause}
      \CousinhoodStar  \ConverseConditionalAscendent  }_{=\ConditionalActiveMinus \textrm{by~\eqref{eq:ConditionalActiveMinus}}}
      \eqfinp               
  \end{align*}
  
  \noindent $\bullet$ (\PY{k+kn}{Lemma}~~\PY{n+nf}{D\_L\ref{le:lemma8}\PYZus{}E\ref{eq:ConditionalActivePlus_supset}}\PY{o})
  We prove~\eqref{eq:ConditionalActivePlus_supset} as follows:     
  \begin{align*}
    \ConditionalActiveMinus\Delta_{\Complementary{\AgentSubsetW}} {\Precedence}
    &=
      \bgp{ \ConverseConditionalAscendent \cup
      \Bp{
      \bp{\ConditionalAscendent \cup \ConditionalCommonCause}
      \CousinhoodStar \ConverseConditionalAscendent } }
      \Delta_{\Complementary{\AgentSubsetW}} {\Precedence}
      \tag{by definition~\eqref{eq:ConditionalActiveMinus} of~$\ConditionalActiveMinus$}
    \\
    &=
      \bgp{ \Delta \cup 
      \Bp{
      \bp{\ConditionalAscendent \cup \ConditionalCommonCause}
      \CousinhoodStar } } 
      \ConverseConditionalAscendent \Delta_{\Complementary{\AgentSubsetW}} {\Precedence}
      \tag{by factorizing $\ConverseConditionalAscendent$}
    \\
    &\subset 
      \bgp{ \Delta \cup 
      \Bp{
      \bp{\ConditionalAscendent \cup \ConditionalCommonCause}
      \CousinhoodStar } }
      \ConverseConditionalAscendent \Delta_{\Complementary{\AgentSubsetW}} 
      \ConditionalAscendent 
      \tag{as \( \Precedence \subset \ConditionalAscendent \)
      by~\eqref{eq:conditional_ascendent_relation} }
    \\
    & =
      \bgp{ \Delta \cup 
      \Bp{
      \bp{\ConditionalAscendent \cup \ConditionalCommonCause}
      \CousinhoodStar } } 
      \ConditionalCommonCause
      \tag{by definition~\eqref{eq:common_cause} of $\ConditionalCommonCause$}
    \\
    & =
      \ConditionalCommonCause \cup 
      \bp{\ConditionalAscendent \cup \ConditionalCommonCause}
      \CousinhoodStar \ConditionalCommonCause 
      \tag{by developing}
    \\        
    &\subset \ConditionalActivePlus
      \eqfinp
      \tag{by definition~\eqref{eq:ConditionalActivePlus} of~$\ConditionalActivePlus$}
  \end{align*}

  \noindent $\bullet$ (\PY{k+kn}{Lemma}~~\PY{n+nf}{D\_L\ref{le:lemma8}\PYZus{}E\ref{eq:ConditionalActiveMinus_supset}}\PY{o})
  We prove~\eqref{eq:ConditionalActiveMinus_supset} as follows:
  \begin{align*}
    \ConditionalActivePlus \Delta_{\TransitiveReflexiveClosureOfSet{\AgentSubsetW}}
    \Converse{\Precedence}
    &=
      \bgp{  \ConditionalAscendent \cup \ConditionalCommonCause
      \cup 
      \Bp{\bp{\ConditionalAscendent \cup \ConditionalCommonCause}
      \CousinhoodStar
      {\ConditionalCommonCause}} } \Delta_{\TransitiveReflexiveClosureOfSet{\AgentSubsetW}}
      \Converse{\Precedence}
      \tag{by definition~\eqref{eq:ConditionalActivePlus} of~$\ConditionalActivePlus$}
    \\
    &=
      \bp{
      \np{  \ConditionalAscendent \cup \ConditionalCommonCause}
      \Delta_{\TransitiveReflexiveClosureOfSet{\AgentSubsetW}} \Converse{\Precedence}}
      \cup
      \bgp{
      \bp{\ConditionalAscendent \cup \ConditionalCommonCause}
      \CousinhoodStar
      {\ConditionalCommonCause}
      \Delta_{\TransitiveReflexiveClosureOfSet{\AgentSubsetW}} \Converse{\Precedence}}
      \tag{by developing}
      \eqfinp 
  \end{align*}
  We treat each of the two terms in the union separately.
  We are going to show that each term is included in~$\ConditionalActiveMinus$. 
  
  For the first term, we have that
  \begin{align*}
    \np{\ConditionalAscendent \cup \ConditionalCommonCause }
    \Delta_{\TransitiveReflexiveClosureOfSet{\AgentSubsetW}}
    \Converse{\Precedence}
    &\subset
      \bp{\ConditionalAscendent \cup \ConditionalCommonCause}
      \CousinhoodStar
      \ConverseConditionalAscendent
      \intertext{as $\Delta_{\TransitiveReflexiveClosureOfSet{\AgentSubsetW}}
      \subset \CousinhoodStar$ by~\eqref{eq:CousinhoodStar},
      and as  $\Converse{\Precedence} \subset
      \ConverseConditionalAscendent$
      by~\eqref{eq:converse_conditional_ascendent_relation}}
    &\subset \ConditionalActiveMinus
      \eqfinp
      \tag{as  $\ConditionalActiveMinus
      = \ConverseConditionalAscendent \cup
      \bp{
      \np{\ConditionalAscendent \cup \ConditionalCommonCause} \CousinhoodStar
      \ConverseConditionalAscendent
      }$ by definition~\eqref{eq:ConditionalActiveMinus} }
  \end{align*}
  For the second term, we have that
  \begin{align*}
    &\bp{\ConditionalAscendent \cup \ConditionalCommonCause}
      \CousinhoodStar
      {\ConditionalCommonCause}
      \Delta_{\TransitiveReflexiveClosureOfSet{\AgentSubsetW}} \Converse{\Precedence}
    \\
    &\hspace{1cm}=
      \bp{\ConditionalAscendent \cup \ConditionalCommonCause}
      \Bp{
      \np{\Delta_{\TransitiveReflexiveClosureOfSet{\AgentSubsetW}} \ConditionalCommonCause \Delta_{\TransitiveReflexiveClosureOfSet{\AgentSubsetW}}}^{+}
      \cup
      \Delta_{\TransitiveReflexiveClosureOfSet{\AgentSubsetW}}
      }
      {\ConditionalCommonCause}
      \Delta_{\TransitiveReflexiveClosureOfSet{\AgentSubsetW}} \Converse{\Precedence}
      \tag{by definition~\eqref{eq:CousinhoodStar} of $\CousinhoodStar$}
    \\
    &\hspace{1cm}=
      \bp{\ConditionalAscendent \cup \ConditionalCommonCause}
      \bp{
      \np{\Delta_{\TransitiveReflexiveClosureOfSet{\AgentSubsetW}} \ConditionalCommonCause \Delta_{\TransitiveReflexiveClosureOfSet{\AgentSubsetW}}}^{+}
      \cup
      \Delta_{\TransitiveReflexiveClosureOfSet{\AgentSubsetW}}
      }
      \bp{\Delta_{\TransitiveReflexiveClosureOfSet{\AgentSubsetW}} {\ConditionalCommonCause}
      \Delta_{\TransitiveReflexiveClosureOfSet{\AgentSubsetW}}
      }
      \Converse{\Precedence}
      \tag{by inserting $\Delta_{\TransitiveReflexiveClosureOfSet{\AgentSubsetW}}$}
    \\
    &\hspace{1cm}\subset
      \bp{\ConditionalAscendent \cup \ConditionalCommonCause}
      \bp{
      \np{\Delta_{\TransitiveReflexiveClosureOfSet{\AgentSubsetW}} \ConditionalCommonCause \Delta_{\TransitiveReflexiveClosureOfSet{\AgentSubsetW}}}^{+}
      \cup
      \Delta_{\TransitiveReflexiveClosureOfSet{\AgentSubsetW}}
      }
      \Converse{\Precedence}
      \intertext{as $ \np{ \TransitiveClosure{\relation} \cup 
      \Delta_{\TransitiveReflexiveClosureOfSet{\AgentSubsetW}} } \relation \subset 
      \TransitiveClosure{\relation} $ for any relation~$\relation$ such that 
      $\relation\Delta_{\TransitiveReflexiveClosureOfSet{\AgentSubsetW}}=
      \Delta_{\TransitiveReflexiveClosureOfSet{\AgentSubsetW}}\relation=\relation$,
      which is the case for \( \relation=\Delta_{\TransitiveReflexiveClosureOfSet{\AgentSubsetW}} \ConditionalCommonCause
      \Delta_{\TransitiveReflexiveClosureOfSet{\AgentSubsetW}} \)}
    &\hspace{1cm}=
      \bp{\ConditionalAscendent \cup \ConditionalCommonCause}
      \CousinhoodStar
      \Converse{\Precedence}
      \tag{by definition~\eqref{eq:CousinhoodStar} of $\CousinhoodStar$}
    \\
    &\hspace{1cm}\subset
      \bp{\ConditionalAscendent \cup \ConditionalCommonCause}
      \CousinhoodStar
      \ConverseConditionalAscendent
      \tag{as $ \Converse{\Precedence}
      \subset\ConverseConditionalAscendent$
      by~\eqref{eq:converse_conditional_ascendent_relation}}
    \\
    &\hspace{1cm}\subset \ConditionalActiveMinus
      \eqfinp
      \tag{as  $\ConditionalActiveMinus
      = \ConverseConditionalAscendent \cup
      \bp{
      \np{\ConditionalAscendent \cup \ConditionalCommonCause} \CousinhoodStar
      \ConverseConditionalAscendent
      }$ by definition~\eqref{eq:ConditionalActiveMinus} }
  \end{align*}
  We conclude that 
  $\ConditionalActivePlus
  \Delta_{\TransitiveReflexiveClosureOfSet{\AgentSubsetW}}
  \Converse{\Precedence} \subset \ConditionalActiveMinus$.
  \medskip

  This ends the proof. 
\end{proof}

\subsubsection{Case of active \undirectedEdgePath\ of length~2}

The following Lemma~\ref{le:active-length-2} is instrumental in the proof
of Proposition~\ref{th:main_isimplied}. It covers the easy case of 
active \undirectedEdgePaths\ of length~2.

\begin{lemma}  (Coq \PY{k+kn}{Lemma}~\PY{n+nf}{D\_L\ref{le:active-length-2}})
  \label{le:active-length-2}
  Let \( \npOrientedGraph \) be a graph, 
  and $\AgentSubsetW\subset\VERTEX$ be a subset of vertices.
  Let  \(\UndirectedPath \in  \ActiveUndirectedPaths{\OrientedGraph}  \) be
  a given active \undirectedEdgePath\ of length~2
  (see Definition~\ref{de:ActivePaths})
  with   $\bgent$ as head endpoint and $\cgent$ as tail endpoint, that is,
  there exists a vertex $\dgent\in\VERTEX$ such that
  \begin{equation}
    \UndirectedPath = \ba{ \np{\bgent,\dgent,\orient_1},
      \np{\dgent,\cgent,\orient_2} }
    \text{ with }
    \np{\orient_1,\orient_2} \in \Orientation^2
    \eqfinp
    \label{eq:le:active-length-2}    
  \end{equation}
  Then, one of the following two possibilities holds true:
  \begin{enumerate}
  \item
    the \undirectedEdgePath~$\UndirectedPath$ ends with $\orient_2=+1$ orientation,
    and then \(\bgent \ConditionalActivePlus \cgent\).
  \item
    the \undirectedEdgePath~$\UndirectedPath$ ends with $\orient_2=-1$ orientation,
    and then \(\bgent \ConditionalActiveMinus \cgent\).
  \end{enumerate}
\end{lemma}

\begin{proof}
  As \(\UndirectedPath \) in~\eqref{eq:le:active-length-2}
  belongs to~\( \ActiveUndirectedPaths{\OrientedGraph} \) --- 
  hence to \( \UPATH\np{\OrientedGraph} \),
  the set of \undirectedEdgePaths\ of positive length in~\eqref{def:eopaths_gtzero} --- 
  we have that 
  \( \np{\bgent,\dgent} \in \EDGE \) if $\orient_1={+1}$,
  \( \np{\bgent,\dgent} \in \Converse{\EDGE} \) if $\orient_1={-1}$,
  \( \np{\dgent,\cgent} \in \EDGE \) if $\orient_2={+1}$,
  and \( \np{\dgent,\cgent} \in \Converse{\EDGE} \) if $\orient_2={-1}$.
  Now, we consider the four conditions enumerated in Definition~\ref{de:ActivePaths} which must be satisfied
  for the \undirectedEdgePath~$\UndirectedPath$ to be active and which impose constraints on the
  vertex~$\dgent$ according to the possible orientations.

  \begin{enumerate}
  \item
    First, we consider the case when the \undirectedEdgePath~$\UndirectedPath$
    ends with $+1$ orientation, that is, when $\orient_2=+1$.
    \begin{itemize}
    \item
      Item~\ref{it:ActivePaths_case1} in Definition~\ref{de:ActivePaths}
      corresponds to 
      $\orient_1 = +1$, $\orient_{2}= +1$ and $\dgent \in
      \Complementary{\AgentSubsetW}$, which gives that 
      $\bgent \EDGE^{(+1)} \Delta_{\Complementary{\AgentSubsetW}}\dgent$ and $\dgent\Delta_{\Complementary{\AgentSubsetW}}
      \EDGE^{(+1)}\cgent$. Hence, by composition of binary relations, we get that
      $\bgent \EDGE \Delta_{\Complementary{\AgentSubsetW}} \EDGE\cgent$,
      using the property \( \Delta_{\Complementary{\AgentSubsetW}}\Delta_{\Complementary{\AgentSubsetW}}=\Delta_{\Complementary{\AgentSubsetW}}\).
    \item
      Item~\ref{it:ActivePaths_case3} in Definition~\ref{de:ActivePaths} corresponds to 
      $\orient_1 = -1$, $\orient_{2}= +1$ and $\dgent \in
      \Complementary{\AgentSubsetW}$, which gives that 
      $\bgent \EDGE^{(-1)} \Delta_{\Complementary{\AgentSubsetW}}\dgent$ and $\dgent\Delta_{\Complementary{\AgentSubsetW}}
      \EDGE^{(+1)}\cgent$. Hence, we get that 
      $\bgent \Converse\EDGE \Delta_{\Complementary{\AgentSubsetW}} \EDGE\cgent$.
    \end{itemize}
    We have obtained that the \undirectedEdgePath~$\UndirectedPath$ ends with orientation~$+1$  and is
    such that
    either \(\bgent  \Precedence \Delta_{\Complementary{\AgentSubsetW}}\Precedence  \cgent\)
    or \(\bgent   \Converse{\Precedence} \Delta_{\Complementary{\AgentSubsetW}} \Precedence  \cgent\).
    Using the properties that
    \( \Precedence \Delta_{\Complementary{\AgentSubsetW}}\Precedence
    \subset \ConditionalAscendent
    \subset \ConditionalActivePlus\)
    (by~\eqref{eq:conditional_ascendent_relation} and~\eqref{eq:ConditionalActivePlus})
    and that 
    \( \Converse{\Precedence} \Delta_{\Complementary{\AgentSubsetW}} \Precedence
    \subset \ConditionalCommonCause \subset \ConditionalActivePlus\)
    (by~\eqref{eq:common_cause} and~\eqref{eq:ConditionalActivePlus}),
    we obtain that \(\bgent \ConditionalActivePlus \cgent\).
        
  \item
    Second, we consider the case when the \undirectedEdgePath~$\UndirectedPath$ ends with $-1$ orientation,
    that is, when $\orient_2=-1$:
    \begin{itemize}
    \item
      Item~\ref{it:ActivePaths_case2} in Definition~\ref{de:ActivePaths} corresponds to 
     $\orient_1 = -1$, $\orient_{2}= -1$ and $\dgent \in \Complementary{\AgentSubsetW}$, which gives that 
       $\bgent \EDGE^{(-1)} \Delta_{\Complementary{\AgentSubsetW}}\dgent$ and $\dgent\Delta_{\Complementary{\AgentSubsetW}}
      \EDGE^{(-1)}\cgent$. Hence, we get that 
      $\bgent \Converse\EDGE \Delta_{\Complementary{\AgentSubsetW}}
      \Converse\EDGE\cgent$,
    \item
      Item~\ref{it:ActivePaths_case4} in Definition~\ref{de:ActivePaths} corresponds to 
      $\orient_1 = +1$, $\orient_{2}= -1$ and $\dgent \in \TransitiveReflexiveClosureOfSet{\AgentSubsetW}$, which gives that 
       $\bgent \EDGE^{(+1)}
       \Delta_{\TransitiveReflexiveClosureOfSet{\AgentSubsetW}}\dgent$
       and $\dgent\Delta_{\TransitiveReflexiveClosureOfSet{\AgentSubsetW}} \EDGE^{(-1)}\cgent$. Hence, we get that 
      $\bgent \EDGE \Delta_{\TransitiveReflexiveClosureOfSet{\AgentSubsetW}}
      \Converse\EDGE\cgent$,
      using the property \(
      \Delta_{\TransitiveReflexiveClosureOfSet{\AgentSubsetW}}\Delta_{\TransitiveReflexiveClosureOfSet{\AgentSubsetW}}
      = \Delta_{\TransitiveReflexiveClosureOfSet{\AgentSubsetW}} \). 
    \end{itemize}
    We have obtained that the \undirectedEdgePath~$\UndirectedPath$ ends with orientation~$-1$ and is
    such that
    either 
    \( \bgent\Precedence \Delta_{\TransitiveReflexiveClosureOfSet{\AgentSubsetW}} \Converse{\Precedence}\cgent \)
    or \( \bgent \Converse{\Precedence} \Delta_{\Complementary{\AgentSubsetW}}\Converse{\Precedence} \cgent\).
    Using the properties that
    \(\Precedence \Delta_{\TransitiveReflexiveClosureOfSet{\AgentSubsetW}} \Converse{\Precedence}
    \subset   \ConditionalAscendent  \Delta_{\TransitiveReflexiveClosureOfSet{\AgentSubsetW}} \ConverseConditionalAscendent
    \subset \ConditionalActiveMinus \)
    (by~\eqref{eq:conditional_ascendent_relation}
    and~\eqref{eq:ConditionalActiveMinus})
    and that
    \(\Converse{\Precedence} \Delta_{\Complementary{\AgentSubsetW}}\Converse{\Precedence}
    \subset \ConverseConditionalAscendent\subset \ConditionalActiveMinus\)
    (by~\eqref{eq:converse_conditional_ascendent_relation}
    and~\eqref{eq:ConditionalActiveMinus}),
    we obtain that  \(\bgent \ConditionalActiveMinus \cgent\).

  \end{enumerate}

  This ends the proof.
\end{proof}

\subsubsection{Case of active \undirectedEdgePath\ of length greater than~$2$}

The following Lemma~\ref{lem:induction-lemma} is instrumental in the proof of
Proposition~\ref{th:main_isimplied}.  It covers by induction the case of active
\undirectedEdgePaths\ of length greater than~2.

\begin{figure}[hbtp]
  \begin{center}
    \mbox{\includegraphics[width=0.9\textwidth]{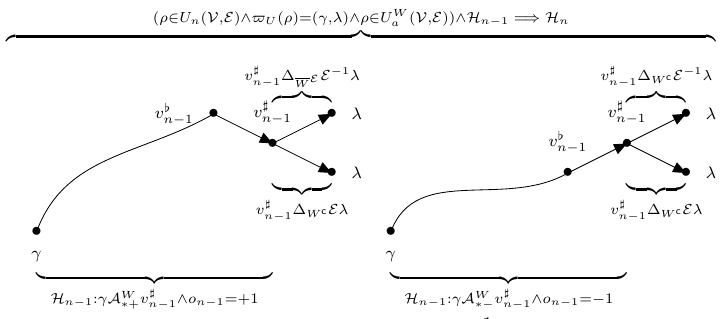}}
    \caption{Sketch of proof of
      induction Lemma~\ref{lem:induction-lemma} used in
      Proposition~\ref{th:main_isimplied} in Appendix~\ref{Proof_of_Theorem_isimplied}
      \label{fig:isimplied}}
  \end{center}
\end{figure}

\begin{lemma}  (Coq \PY{k+kn}{Lemma}~\PY{n+nf}{D\_L\ref{lem:induction-lemma}})
  \label{lem:induction-lemma}
  Let \( \npOrientedGraph \) be a graph, 
  and $\AgentSubsetW\subset\VERTEX$ be a subset of vertices.
  For any $n\ge 2$ and $\bgent$, $\cgent \in \AGENT$, the following statement
  holds true.
  For any
  \undirectedEdgePath\ \( \UndirectedPath \in \UPATH_{n}\np{\graph} \) of length~$n$
  joining vertices~$\bgent$ and $\cgent$ (that is,
  \( \Projection_{\UPATH}(\UndirectedPath)= \np{\bgent,\cgent}\) as in~\eqref{eq:Projection_UPATH_all}), and
  such that \( \UndirectedPath \) is active (that is,
  \( \UndirectedPath \in \ActiveUndirectedPaths{\graph} \) as in
  Definition~\ref{de:ActivePaths}), 
  one of the two following properties is fullfiled:
  \begin{enumerate}
  \item
    Either $\bgent\ConditionalActivePlus \cgent$ and
    the last orientation of $\UndirectedPath$ is $\orient_n=+1$,
    \label{it:lem:induction-lemma_+1}
  \item
    Or $\bgent\ConditionalActiveMinus \cgent$ and
    the last orientation of $\UndirectedPath$ is $\orient_n=-1$.
    \label{it:lem:induction-lemma_-1}
  \end{enumerate}
\end{lemma}

\begin{proof}
  We call ${\mathcal H}_{n}$ the statement in Lemma~\ref{lem:induction-lemma}
  and we prove by induction that it is satisfied for all $n\ge 2$.
  
  The proof of ${\mathcal H}_{2}$ is given by Lemma~\ref{le:active-length-2}.
  \medskip
  
  We suppose that the induction assumption~${\mathcal H}_{n-1}$ holds true, where $n-1 \ge 2$,
  and we are going to show that ${\mathcal H}_{n}$ holds true. 
  For this purpose, we consider an \undirectedEdgePath~$\UndirectedPath$ of length~$n$ ($n \geq 2$), 
  joining vertices~$\bgent$ and $\cgent$ in the graph~$(\graph)$, and 
  which is active, that is, 
  \begin{equation*}
    \UndirectedPath \in \UPATH_{n}\np{\graph}
    \text{ and }
    \Projection_{\UPATH}(\UndirectedPath)= \np{\bgent,\cgent}
    \text{ and }
    \UndirectedPath \in \ActiveUndirectedPaths{\graph}
    \eqfinp 
  \end{equation*}
  We decompose the \undirectedEdgePath~$\UndirectedPath$ as
  \[
    \UndirectedPath=\nseqa{(\tail{\vertex_i},\head{\vertex_i},\orient_i)}{i\in \ic{1,n}}
    = \UndirectedPathbis \times \UndirectedPathter
    \eqfinv
  \]
  where
  \(
  \UndirectedPathbis=\nseqa{(\tail{\vertex_i},\head{\vertex_i},\orient_i)}{i \in
      \ic{1,n{-}1}}
  \in \UPATH_{n{-}1}\np{\graph}
  \)
  is an \undirectedEdgePath\ of length~$n{-}1$,
  and where $\UndirectedPathter=\na{\np{\tail{\vertex_n},\head{\vertex_n}, \orient_{n}}}\in \UPATH_{1}\np{\graph}$
  is an \undirectedEdgePath\ of length~$1$.
  We have that \( \tail{\vertex_1}=\bgent \) and \( \head{\vertex_n}=\cgent \).
  It is clear that the \undirectedEdgePath~$\UndirectedPathbis$ is active, 
  that is, \(   \UndirectedPathbis \in \ActiveUndirectedPaths{\graph}\).
  Indeed, otherwise, the \undirectedEdgePath~$\UndirectedPathbis$
  would be in one of the four cases listed in 
  Definition~\ref{de:vertices-d-separated}, hence so would be the \undirectedEdgePath~$\UndirectedPath$.
  But this would contradict the assumption that
  \(  \UndirectedPath \in \ActiveUndirectedPaths{\graph} \).
  As the \undirectedEdgePath~$\UndirectedPathbis$ is active and of length~$n{-}1$,
  it satisfies the induction assumption~${\mathcal H}_{n-1}$. We deduce that
  either $\bgent \ConditionalActivePlus \head{\vertex_{n-1}}$
  and the last orientation of $\UndirectedPathbis$ is
  $\orient_{n-1}=+1$, or
  $\bgent \ConditionalActiveMinus \head{\vertex_{n-1}}$ and the last orientation of $\UndirectedPathbis$ is
  $\orient_{n-1}=-1$.
  We analyze the two cases separately. Each case being subdivided in two cases, we will
  analyse four different cases as summarized in Figure~\ref{fig:isimplied}.
  
  \medskip
  \noindent $\bullet$
  Assume that we have
  $\bgent \ConditionalActivePlus \head{\vertex_{n-1}}$ and that
  the last orientation of~$\UndirectedPathbis$ is $\orient_{n-1}=+1$,
  that is, $\UndirectedPathbis$ ends with \(
  \np{\tail{\vertex_{n-1}},\head{\vertex_{n-1}}, +1} \). 
  There are two possibilities for the
  \undirectedEdgePath~$\UndirectedPathter=
  \na{\np{\tail{\vertex_n},\head{\vertex_n}, \orient_{n}}}$.
  
  \begin{itemize}
    
  \item
    Suppose that $\UndirectedPathter=\na{\np{\tail{\vertex_n},\head{\vertex_n}, +1}}
    = \na{\np{\tail{\vertex_{n}},\cgent, +1}}$, that
    is, \( \np{\tail{\vertex_{n}},\cgent} \in \Precedence \)
    by~\eqref{eq:eopaths_defs}.
    As the path $\UndirectedPath$ is active by assumption, the pattern
    \( \ba{ \np{\tail{\vertex_{n-1}},\head{\vertex_{n-1}}, +1},
      \np{\tail{\vertex_n},\head{\vertex_n}, +1} } \) must satisfy Item~\ref{it:ActivePaths_case1} in
    Definition~\ref{de:ActivePaths}. We deduce that 
    $\head{\vertex_{n-1}}=\tail{\vertex_{n}} \in \Complementary{\AgentSubsetW}$. 
    Now, we wrap up the results obtained so far.
    On the one hand, from $\bgent \ConditionalActivePlus \head{\vertex_{n-1}}$,
    $\head{\vertex_{n-1}}=\tail{\vertex_{n}} \in \Complementary{\AgentSubsetW}$
    and \( \np{\tail{\vertex_{n}},\cgent} \in \Precedence \),
    we get that  $\bgent\ConditionalActivePlus
    \Delta_{\Complementary{\AgentSubsetW}}
    \Precedence\cgent$, hence that 
    $\bgent \ConditionalActivePlus\cgent$ because
    the relation~$\ConditionalActivePlus$ in~\eqref{eq:ConditionalActivePlus}
    ends with the relation
    $\ConditionalAscendent$ and as 
    $\ConditionalAscendent \Delta_{\Complementary{\AgentSubsetW}}\Precedence
    =\ConditionalDown (\Delta_{\Complementary{\AgentSubsetW}}\Precedence)
    \subset \ConditionalAscendent$ 
    by~\eqref{eq:conditional_ascendent_relation}.
    On the other hand, the \undirectedEdgePath~$\UndirectedPath$ ends with~$+1$,
    as it is the case for~$\UndirectedPathter$.
    We conclude that the \undirectedEdgePath~$\UndirectedPath$ of length~$n$
    satisfies the case~\ref{it:lem:induction-lemma_+1} of~${\mathcal H}_{n}$,
    since it ends with~$+1$ and its endpoints are such 
    that $\bgent \ConditionalActivePlus\cgent$.
    Therefore, we have proven the case~\ref{it:lem:induction-lemma_+1} of
    the induction assumption~${\mathcal H}_{n}$
    for the \undirectedEdgePaths\ of length~$n$.
  \item
    Suppose that $\UndirectedPathter=
    \na{\np{\tail{\vertex_n},\head{\vertex_n}, -1}}=
    \na{\np{\tail{\vertex_{n}},\cgent, -1}}$, that
    is, \( \np{\tail{\vertex_{n}},\cgent} \in \Converse{\Precedence} \)
    by~\eqref{eq:eopaths_defs}.
    As the path $\UndirectedPath$ is active, the pattern
    \( \ba{\np{\tail{\vertex_{n-1}},\head{\vertex_{n-1}}, +1},
      \np{\tail{\vertex_n},\head{\vertex_n}, -1} } \) must satisfy Item~\ref{it:ActivePaths_case4} in
    Definition~\ref{de:ActivePaths}. We deduce that 
    $\head{\vertex_{n-1}}=\tail{\vertex_{n}} \in \TransitiveReflexiveClosureOfSet{\AgentSubsetW}$.
    Now, we wrap up the results obtained so far.
    On the one hand, from $\bgent \ConditionalActivePlus \head{\vertex_{n-1}}$,
    $\head{\vertex_{n-1}}=\tail{\vertex_{n}} \in \TransitiveReflexiveClosureOfSet{\AgentSubsetW}$
    and \( \np{\tail{\vertex_{n}},\cgent} \in \Converse{\Precedence} \),
    we get that  $\bgent\ConditionalActivePlus
    \Delta_{\TransitiveReflexiveClosureOfSet{\AgentSubsetW}}
    \Converse{\Precedence}\cgent$, hence that 
    $\bgent \ConditionalActiveMinus\cgent$  by~\eqref{eq:ConditionalActiveMinus_supset}.
    On the other hand, the \undirectedEdgePath~$\UndirectedPath$ ends with~$-1$,
    as it is the case for~$\UndirectedPathter$.
    We conclude that the \undirectedEdgePath~$\UndirectedPath$ of length~$n$
    satisfies the case~\ref{it:lem:induction-lemma_-1} of~${\mathcal H}_{n}$,
    since it ends with~$-1$ and its endpoints are such 
    that $\bgent \ConditionalActiveMinus\cgent$.
    Therefore, we have proven the case~\ref{it:lem:induction-lemma_-1} of
    the induction assumption~${\mathcal H}_{n}$
    for the \undirectedEdgePaths\ of length~$n$.
  \end{itemize}

  \medskip
  \noindent $\bullet$
  Assume that we have
  $\bgent \ConditionalActiveMinus \head{\vertex_{n-1}}$
  and that the last orientation of~$\UndirectedPathbis$ is
  $\orient_{n-1}=-1$, that
  is, \( \np{\tail{\vertex_{n}},\cgent} \in \Converse{\Precedence} \). 
  There are two possibilities for the 
  \undirectedEdgePath~$\UndirectedPathter=
  \na{\np{\tail{\vertex_n},\head{\vertex_n},\orient_{n}}}$.
  
  \begin{itemize}
  \item
    Suppose that $\UndirectedPathter=
    \na{\np{\tail{\vertex_n},\head{\vertex_n}, +1}}= \na{\np{\tail{\vertex_{n}},\cgent, +1}}$, that
    is, \( \np{\tail{\vertex_{n}},\cgent} \in \Precedence \)
    by~\eqref{eq:eopaths_defs}.
    As the path $\UndirectedPath$ is active, the pattern
    \( \ba{\np{\tail{\vertex_{n-1}},\head{\vertex_{n-1}}, -1},
      \np{\tail{\vertex_n},\head{\vertex_n}, +1} } \) must satisfy Item~\ref{it:ActivePaths_case3} in
    Definition~\ref{de:ActivePaths}. We deduce that 
    $\head{\vertex_{n-1}}=\tail{\vertex_{n}} \in \Complementary{\AgentSubsetW}$. 
    Now, we wrap up the results obtained so far.
    On the one hand, from $\bgent \ConditionalActiveMinus \head{\vertex_{n-1}}$,
    $\head{\vertex_{n-1}}=\tail{\vertex_{n}} \in \Complementary{\AgentSubsetW}$
    and \( \np{\tail{\vertex_{n}},\cgent} \in \Precedence \),
    we get that  $\bgent\ConditionalActiveMinus
    \Delta_{\Complementary{\AgentSubsetW}}
    \Converse{\Precedence}\cgent$, hence that 
    $\bgent \ConditionalActivePlus\cgent$  by~\eqref{eq:ConditionalActiveMinus_supset}.
    On the other hand, the \undirectedEdgePath~$\UndirectedPath$ ends with~$+1$,
    as it is the case for~$\UndirectedPathter$.
    We conclude that the \undirectedEdgePath~$\UndirectedPath$ of length~$n$
    satisfies the case~\ref{it:lem:induction-lemma_+1} of~${\mathcal H}_{n}$,
    since it ends with~$+1$ and its endpoints are such 
    that $\bgent \ConditionalActivePlus\cgent$.
    Therefore, we have proven the case~\ref{it:lem:induction-lemma_+1} of
    the induction assumption~${\mathcal H}_{n}$
    for the \undirectedEdgePaths\ of length~$n$.
  \item
    Suppose that $\UndirectedPathter=
    \na{\np{\tail{\vertex_n},\head{\vertex_n}, -1}}= \na{\np{\tail{\vertex_{n}},\cgent, -1}}$, that
    is, \( \np{\tail{\vertex_{n}},\cgent} \in \Converse{\Precedence} \)
    by~\eqref{eq:eopaths_defs}.
    As the path $\UndirectedPath$ is active, the pattern
    \( \ba{ \np{\tail{\vertex_{n-1}},\head{\vertex_{n-1}}, -1},
      \np{\tail{\vertex_n},\head{\vertex_n}, -1} } \)
    must satisfy Item~\ref{it:ActivePaths_case2} in
    Definition~\ref{de:ActivePaths}. We deduce that 
    $\head{\vertex_{n-1}}=\tail{\vertex_{n}} \in \Complementary{\AgentSubsetW}$. 
    Now, we wrap up the results obtained so far.
    On the one hand, from $\bgent \ConditionalActiveMinus \head{\vertex_{n-1}}$,
    $\head{\vertex_{n-1}}=\tail{\vertex_{n}} \in \Complementary{\AgentSubsetW}$
    and \( \np{\tail{\vertex_{n}},\cgent} \in \Converse{\Precedence} \),
    we get that  $\bgent\ConditionalActiveMinus
    \Delta_{\Complementary{\AgentSubsetW}}
    \Converse{\Precedence}\cgent$, hence that 
    $\bgent \ConditionalActiveMinus\cgent$ because
    the relation~$\ConditionalActiveMinus$ ends with the relation
    $\ConverseConditionalAscendent$ and we have that
    $\ConverseConditionalAscendent \Delta_{\Complementary{\AgentSubsetW}}\Converse{\Precedence}
    = \ConditionalUp\Delta_{\Complementary{\AgentSubsetW}}\Converse{\Precedence}
    \subset \ConverseConditionalAscendent$  which implies that
    $\ConditionalActiveMinus \Delta_{\Complementary{\AgentSubsetW}}\Converse{\Precedence}
    \subset \ConditionalActiveMinus$. 
    On the other hand, the \undirectedEdgePath~$\UndirectedPath$ ends with~$-1$,
    as it is the case for~$\UndirectedPathter$.
    We conclude that the \undirectedEdgePath~$\UndirectedPath$ of length~$n$
    satisfies the case~\ref{it:lem:induction-lemma_-1} of~${\mathcal H}_{n}$,
    since it ends with~$-1$ and its endpoints are such 
    that $\bgent \ConditionalActiveMinus\cgent$.
    Therefore, we have proven the case~\ref{it:lem:induction-lemma_-1} of
    the induction assumption~${\mathcal H}_{n}$
    for the \undirectedEdgePaths\ of length~$n$.
  \end{itemize}
  \medskip

  This ends the proof. 
\end{proof}

\newcommand{\etalchar}[1]{$^{#1}$}
\newcommand{\noopsort}[1]{} \ifx\undefined\allcaps\def\allcaps#1{#1}\fi


\begin{thebibliography}{ABC{\etalchar{+}}22}

\bibitem[ABC{\etalchar{+}}22]{MathComp:2022}
Reynald Affeldt, Yves Bertot, Cyril Cohen, Marie Kerjean, Assia Mahboubi,
  Damien Rouhling, Pierre Roux, Kazuhiko Sakaguchi, Zachary Stone, Pierre-Yves
  Strub, and Laurent Th\'ery.
\newblock {MathCompAnalysis}: Mathematical components compliant analysis
  library.
\newblock Technical Report Version 0.5.4, 2022.

\bibitem[AGT20]{Affeldt-et-al:2020}
Reynald Affeldt, Jacques Garrigue, and Saikawa Takafumi.
\newblock Reasoning with conditional probabilities and joint distributions in
  {Coq}.
\newblock Technical report, 2020.

\bibitem[CDH21]{Chancelier-De-Lara-Heymann-2021}
Jean-Philippe Chancelier, Michel {De Lara}, and Benjamin Heymann.
\newblock Conditional separation as a binary relation, 2021.
\newblock Preprint.

\bibitem[CDLS06]{cowell2006probabilistic}
Robert~G Cowell, Philip Dawid, Steffen~L Lauritzen, and David~J Spiegelhalter.
\newblock {\em Probabilistic networks and expert systems: Exact computational
  methods for Bayesian networks}.
\newblock Springer Science \& Business Media, 2006.

\bibitem[Cha24]{Chancelier-Coq:2024}
Jean-Philippe Chancelier.
\newblock Coq proofs for ``{Conditional Separation as a Binary Relation}''.
\newblock swh:1:dir:eb5510adb3a8a76b2c1360e779b625404793d285, 2024.

\bibitem[DCH21]{De-Lara-Chancelier-Heymann-2021}
Michel {De Lara}, Jean-Philippe Chancelier, and Benjamin Heymann.
\newblock Topological conditional separation, 2021.
\newblock Preprint.

\bibitem[Die18]{Diestel}
Reinhard Diestel.
\newblock {\em Graph theory}, volume 173 of {\em Graduate Texts in
  Mathematics}.
\newblock Springer, Berlin, fifth edition, 2018.
\newblock Paperback edition of [ MR3644391].

\bibitem[GMT16]{Gonthier-Mahboubi-Tassi:2016}
Georges Gonthier, Assia Mahboubi, and Enrico Tassi.
\newblock {A Small Scale Reflection Extension for the {Coq} system}.
\newblock Research Report RR-6455, {Inria Saclay Ile de France}, 2016.

\bibitem[HDC21]{Heymann-De-Lara-Chancelier-2021}
Benjamin Heymann, Michel {De Lara}, and Jean-Philippe Chancelier.
\newblock Causal inference theory with information dependency models, 2021.
\newblock Preprint.

\bibitem[LDLL90]{Lauritzen-et-al-1990}
S.~L. Lauritzen, A.~P. Dawid, B.~N. Larsen, and H.-G. Leimer.
\newblock Independence properties of directed {Markov} fields.
\newblock {\em Networks}, 20(5):491--505, 1990.

\bibitem[Pea86]{PEARL1986357}
Judea Pearl.
\newblock A constraint - propagation approach to probabilistic reasoning.
\newblock In Laveen~N. Kanal and John~F. Lemmer, editors, {\em Uncertainty in
  Artificial Intelligence}, volume~4 of {\em Machine Intelligence and Pattern
  Recognition}, pages 357--369. North-Holland, 1986.

\bibitem[Pea95]{pearl1995causal}
Judea Pearl.
\newblock Causal diagrams for empirical research.
\newblock {\em Biometrika}, 82(4):669--688, 1995.

\bibitem[PM18]{pearl2018book}
Judea Pearl and Dana Mackenzie.
\newblock {\em The book of {W}hy: the new science of cause and effect}.
\newblock Basic Books, 2018.

\bibitem[Wan10]{Wang:2010}
Jinfang Wang.
\newblock A universal algebraic approach for conditional independence.
\newblock {\em Annals of the Institute of Statistical Mathematics},
  62:747--773, 2010.

\bibitem[YKS{\etalchar{+}}16]{Yamaguchi-et-al:2016}
Rutaro Yamaguchi, Ken Kin, Shugo Shimoyama, Manabu Hagiwara, Mitsuharu
  Yamamoto, and Jinfang Wang.
\newblock Formalization of the conditional independence using {Coq/SSReflect}.
\newblock Technical report, 2016.

\end{thebibliography}
\end{document}